\documentclass[12pt,hidelinks]{article}
\usepackage{amsfonts,amsmath,float,color,algorithm,amsthm,
  verbatim,subcaption,algorithmic,bm,booktabs,xcolor,enumerate,
  url,cases,indentfirst,graphicx,multicol,hyperref,multirow,amssymb}
\usepackage{arydshln}
\usepackage[T1]{fontenc}
\usepackage[letterpaper,top=1in,bottom=1in,left=1in,right=1in]{geometry}

\usepackage{natbib}

\graphicspath{{pictures/}}
\allowdisplaybreaks

\newcommand{\rp}{\mathrm{p}}
\newcommand{\rw}{\mathrm{w}}
\newcommand{\ruw}{\mathrm{uw}}

\newcommand{\onen}{\frac{1}{n}}
\newcommand{\oner}{\frac{1}{r}}

\newcommand{\sumn}{\sum_{i=1}^{n}}
\newcommand{\sumjn}{\sum_{j=1}^{n}}
\newcommand{\sumr}{\sum_{i=1}^{r}}

\renewcommand{\Pr}{\mathbb{P}}
\newcommand{\Exp}{\mathbb{E}}
\newcommand{\Var}{\mathbb{V}}
\newcommand{\mmse}{\mathrm{A-OS}}
\newcommand{\mvc}{\mathrm{L-OS}}

\newcommand{\os}{\mathrm{OS}}

\newcommand{\hbeta}{\hat{\beta}}

\newtheorem{assumption}{Assumption}
\newtheorem{lemma}{Lemma}

\newtheorem{theorem}{Theorem}

\newtheorem{remark}{Remark}

\usepackage{xr}
\begin{document}
\title{Unweighted estimation based on optimal sample under measurement constraints}
\author{Jing Wang $^{1,2,3}$, HaiYing Wang $^{3}$\footnote{Author to whom
    correspondence may be addressed. Email:haiying.wang@uconn.edu}, Shifeng Xiong $^1$}
\maketitle
\centerline{\it
  NCMIS, KLSC, Academy of Mathematics and Systems Science, CAS, Beijing 100190, China $^1$}
\centerline{\it School of Mathematical Sciences, University of Chinese Academy
  of Sciences, Beijing 100049, China $^2$}
\centerline{\it Department of Statistics, University of Connecticut, Storrs,
  CT 06269, U.S.A. $^3$}
\vspace{.55cm}
\begin{abstract}
  To tackle massive data, subsampling is a practical approach to select the
  more informative data points. However, when responses are expensive to
  measure, developing efficient subsampling schemes is challenging, and an
  optimal sampling approach under measurement constraints was developed to meet
  this challenge. This method uses the inverses of optimal sampling
  probabilities to reweight the objective function, which assigns smaller
  weights to the more important data points. Thus the estimation efficiency of
  the resulting estimator can be improved. In this paper, we propose an
  unweighted estimating procedure based on optimal subsamples to obtain a more
  efficient estimator. We obtain the unconditional asymptotic distribution of
  the estimator via martingale techniques without conditioning on the pilot
  estimate, which has been less investigated in the existing subsampling
  literature. Both asymptotic results and numerical results show that the
  unweighted estimator is more efficient in parameter estimation.

  {\it keywords}: Generalized Linear Models; Massive Data; Martingale Central Limit Theorem

  {\it MSC2020}: Primary 62D05; secondary 62J12
\end{abstract}

\newpage
\section{INTRODUCTION}\label{sec:intro}

Data acquisition is becoming easier nowadays, and massive data bring new challenges
to data storage and processing. Conventional statistical models may not be
applicable due to limited computational resources. Facing such problems,
subsampling has become a popular approach to reduce computational burdens. The
key idea of subsampling is to collect more informative data points from the full
data and perform calculations on a smaller data set, see
\cite{drineas2006sampling}; \cite{Drineas:11}; \cite{mahoney2011randomized}. In
some circumstances, covariates $\{X_{i}\}$ are available for all the data
points, but responses $\{Y_{i}\}$ can be obtained for only a small portion
because they are expensive to measure. For example, the extremely large size of
modern galaxy datasets has made visual classification of galaxies
impractical. Most subsampling probabilities developed recently for generalized
linear models (GLMs) rely on complete responses in the full data set, see
\cite{WangZhuMa2017}; \cite{wang2019more}, \cite{ai2019optimal}. In order to
handle the difficulty when responses are hard to measure,
\cite{Zhang2020optimal} proposed a response-free optimal sampling scheme under
measurement constraints (OSUMC) for GLMs. However, their method uses the
reweighted estimator which is not the most efficient one, since it assigns
smaller weights to the more informative data points in the objective
function. The robust sampling probabilities proposed in \cite{nie2018minimax} do
not depend on the responses either, but their investigation focused on linear
regression models.

In this paper, we focus on a subsampling method under measurement constraints
and propose a more efficient estimator based on the same subsamples taken
according to OSUMC for GLMs. We use martingale techniques to derive the
unconditional asymptotic distribution of the unweighted estimator and show that
its asymptotic covariance matrix is smaller, in the Loewner ordering, than that
of the weighted estimator. Before showing the structure of the paper, we first
give a short overview of the emerging field of subsampling methods.

Various subsampling methods have been studied in recent years. For linear
regression, \cite{drineas2006sampling} developed a subsampling method based on
statistical leveraging scores. \cite{Drineas:11} developed an algorithm using
randomized Hardamard transform. \cite{ma2015statistical} investigated the
statistical perspective of leverage sampling. \cite{WangYangStufken2018}
developed an information-based procedure to select optimal subdata for linear
regression deteministically. \cite{zhang2020distributed} proposed a distributed
sampling-based approach for linear models. \cite{ma2020asymptotic} studied the
statistical properties of sampling estimators and proposed several estimators
based on asymptotic results which are related to leveraging scores. Beyond
linear models, \cite{fithian2014local} proposed a local case-control subsampling
method to handle imbalanced data sets for logistic
regression. \cite{WangZhuMa2017} developed an optimal sampling method under
A-optimality criterion (OSMAC) for logistic regression. Their estimator can be
improved because inverse probability reweighting is applied on the objective
function, and \cite{wang2019more} developed a more efficient estimator for
logistic regression based on optimal subsample. They proposed an unweighted
estimator with bias correction using an idea similar to
\cite{fithian2014local}. They also introduced a Poisson sampling algorithm to
reduce RAM usage when calculating optimal sampling probabilities.
\cite{ai2019optimal} generalized OSMAC to GLMs and obtained optimal subsampling
probabilities under A- and L-optimality criteria for GLMs. These optimal
sampling methods require all the responses in order to construct optimal
probabilities, which is not possible under measurement
constraints. \cite{Zhang2020optimal} developed an optimal sampling method under
measurement constraints. Their estimator is also based on the weighted objective
function and thus the performance can be improved. Recently,
\cite{cheng2020information} extended an information-based data selection
approach for linear models to logistic regression. \cite{yu2020optimal} derived
optimal Poisson subsampling probabilies under the A- and L-optimality criteria
for quasi-likelihood estimation, and developed a distributed subsampling
framwork to deal with data stored in different machines. \cite{wang2020optimal}
developed an optimal sampling method for quantile
regression. \cite{pronzato2020sequential} proposed a sequential online
subsampling procedure based on optimal bounded design measures.

We focus on GLMs in this paper, which include commonly used models such as
linear, logistic and Poisson regression. The rest of the paper is organized as
follows. Section \ref{sec:background} presents the model setup and briefly
reviews the OSUMC method. The more efficient estimator and its asymptotic
properties are presented in Section \ref{sec:theory}. Section
\ref{sec:numerical} provides numerical simulations. We summarize our paper in
Section \ref{sec:conclusion}. Proofs and technical details are presented in the
Supplementary Material.

\section{BACKGROUND AND MODEL SETUP}\label{sec:background}

We start by reviewing GLMs. Consider %
independent and identically distributed
(i.i.d) %
data $(X_{1},Y_{1})$, $(X_{2},Y_{2})$,..., $(X_{n},Y_{n})$ from the distribution
of $(X,Y)$,
where $X\in\mathbb{R}^{p}$ is the covariate vector and
$Y$ is the response variable. %
Assume that the conditional density of $Y$ given $X$ satisfies that
\begin{equation*}
	f(y\vert x,\beta_{0},\sigma)\propto \exp\left\{\frac{yx^{T}\beta_{0}-b(x^{T}\beta_{0})}{c(\sigma)}\right\},%
\end{equation*}
where $\beta_{0}$ is the unknown parameter we need to estimate from data,
$b(\cdot)$ and $c(\cdot)$ are known functions, and $\sigma$ is the dispersion
parameter.  In this paper, we are only interested in estimating
$\beta_{0}$. Thus, we take $c(\sigma)=1$ without loss of generality. We also
include an intercept in the model, as is almost always the case in
practice. We obtain the maximum likelihood estimator (MLE) of $\beta_{0}$
through maximizing the loglikelihood function, namely,
\begin{equation} \label{eq:3}
	\hbeta_{\text{MLE}}:=\mathop{\arg\max}_{\beta}
  \onen\sumn\left\{Y_{i}X_{i}^{T}\beta-b(X_{i}^{T}\beta)\right\},
\end{equation}
which is the same as solving the following score equation:
\begin{equation*}
	\Psi_{n}(\beta):=\onen\sumn\{b^{'}(X_{i}^{T}\beta)-Y_{i}\}X_{i}=0,
\end{equation*}
where $b^{'}(\cdot)$ is the derivative of $b(\cdot)$. There is no general
closed-form solution to $\hbeta_{\text{MLE}}$, and iterative algorithms
such as Newton's method are often used. Therefore, when the data are massive,
the computational burden of estimating $\beta_{0}$ is very heavy. To handle this
problem, \cite{ai2019optimal} proposed a subsampling-based approach, which
constructs sampling probabilities $\{\pi_{i}\}_{i=1}^{n}$ that depend on both
the covariates $\{X_i\}$ and the responses $\{Y_i\}$.  However, it is infeasible to
obtain all the responses under measurement constraints. For example, it costs
considerable money and time to synthesize superconductors. When we use
data-driven methods to predict the critical temperature with the chemical
composition of superconductors, it may be more pratical to measure a small
number of materials to build a data-driven model. To tackle this type of ``many
$X$, few $Y$'' scenario, \cite{Zhang2020optimal} developed OSUMC subsampling
probabilities.

Assume we obtain a subsample of size $r$ by sampling with replacement according
to the probabilities $\pi=\{\pi_{i}\}_{i=1}^{n}$. A reweighted estimator is
often used in subsample literature, defined as the minimizer of the reweighted
target function, namely
\begin{equation}\label{eq:def_beta_w}
	\hbeta_{\rw}:=\mathop{\arg\max}_{\beta}
  \oner\sumr\frac{Y_{i}^{*}X_{i}^{*{T}}\beta-b(X_{i}^{*{T}}\beta)}{n\pi_{i}^{*}},
\end{equation}
where $(X_{i}^{*},Y_{i}^{*})$ is the data sampled in the $i$th step, and
$\pi_{i}^{*}$ denotes the corresponding sampling probability. Equivalently, we
can solve the reweighted score function
\begin{equation*}
	\Psi^{*}_{\rw}(\beta):=
  \oner\sumr\frac{b^{'}(X_{i}^{*{T}}\beta)-Y_{i}^{*}}{n\pi_{i}^{*}}X_{i}^{*}=0,
\end{equation*}
to obtain the reweighted estimator. \cite{Zhang2020optimal} proposed a scheme to
derive the optimal subsampling probabilities for GLMs under measurement
constraints. They first proved that $\hbeta_{\rw}$ is asymptotically
normal:
\begin{equation*}
 	\Var\{\Psi^{*}_{\rw}(\beta_{0})\}^{-\frac{1}{2}}\Phi(\hbeta_{\rw}-\beta_{0})\xrightarrow{d}N(0,I),
\end{equation*}
where the notation ``$\xrightarrow{d}$'' denotes convergence in distribution,
\begin{equation*}
 	\Var\{\Psi^{*}_{\rw}(\beta_{0})\}:=\Exp\left[\Var\{\Psi^{*}_{\rw}(\beta_{0})\vert
          X_{1}^{n}\}\right]=\Exp\left\{\frac{1}{n^{2}}\sumn b^{''}(X_{i}^{T}\beta_{0})X_{i}X_{i}^{T}\left(\frac{1}{r\pi_{i}}-\oner+1\right)\right\},
\end{equation*}
$X_{1}^{n}:=(X_{1},X_{2},...,X_{n})$, $b^{''}(\cdot)$ is the second derivative
of $b(\cdot)$, and
\begin{equation}\label{eq:phi}
	\Phi:=\Exp\left\{\onen\sumn b^{''}(X_{i}^{T}\beta_{0})X_{i}X_{i}^{T}\right\}.
\end{equation}
Since the matrix
$\Phi^{-1}\Var\{\Psi^{*}_{\rw}(\beta_{0})\vert
X_{1}^{n}\}\Phi^{-1}$ converges to the asymptotic variance of
$\hbeta_{\rw}$, \cite{Zhang2020optimal} minimized its trace,
$\text{tr}(\Phi^{-1}\Var\{\Psi^{*}_{\rw}(\beta_{0})\vert
X_{1}^{n}\}\Phi^{-1})$, to obtain the optimal sampling probabilities which
depend only on covariate vectors $X_{1}$ ,..., $X_{n}$:
\begin{equation}\label{eq:Api}
  \pi_{i}^{\mmse}(\beta_{0},\Phi)
  =\frac{\sqrt{b^{''}(X_{i}^{T}\beta_{0})}\|\Phi^{-1}X_{i}\|}{\sumjn\sqrt{b^{''}(X_{j}^{T}\beta_{0})}\|\Phi^{-1}X_{j}\|}.
 \end{equation}
To avoid the matrix multiplication in $\|\Phi^{-1}X_{i}\|$ in \eqref{eq:Api}, we
can consider a variant of \eqref{eq:Api} which omits the inverse matrix
$\Phi^{-1}$:
\begin{equation}
\label{eq:Lpi}
\pi_{i}^{\mvc}(\beta_{0})
   =\frac{\sqrt{b^{''}(X_{i}^{T}\beta_{0})}\| X_{i}\|}{\sumjn\sqrt{b^{''}(X_{j}^{T}\beta_{0})}\| X_{j}\|}.
\end{equation}
Here, $\{\pi_i^{\mvc}\}_{i=1}^n$ are other widely used optimal
probabilities, derived by minimizing the quantity
$\text{tr}(L\Phi^{-1}\Var\{\Psi^{*}_{\rw}(\beta_{0})\vert
X_{1}^{n}\}\Phi^{-1}L^{T})$ with $L=\Phi$. This is a special case of using the
L-optimality criterion to obtain optimal subsampling probabilities
\citep[see][]{WangZhuMa2017,ai2019optimal}.
The probabilities in \eqref{eq:Api} and \eqref{eq:Lpi} are useful when
the responses are not available, as we discussed before. However, as pointed out
in \cite{wang2019more}, under the logistic model framework, the weighting
scheme adopted in \eqref{eq:def_beta_w} does not bring us the most efficient
estimator. Intuitively, if a data point $(X_{i}, Y_{i})$ has a larger sampling
probability, it contains more information about $\beta_{0}$. However, data
points with higher sampling probabilities have smaller weights in
\eqref{eq:def_beta_w}. This will reduce the efficiency of the estimator.
We propose a more efficient estimator based on the unweighted target function.

\section{UNWEIGHTED ESTIMATION AND ASYMPTOTIC THEORY} \label{sec:theory}

In this section, we present an algorithm with an unweighted estimator and derive
its asymptotic
property. %
As we discussed before, the reweighted estimator reduces the importance of more
informative data points. To overcome this problem, \cite{wang2019more} developed
a method to correct the bias of the unweighted estimator in logistic
regression. In this section, we show that, using the optimal probabilities
under measurement constraints, the unweighted estimator is
asymptotically unbiased
and therefore it is a better estimator since it has a smaller asymptotic variance
matrix in the Loewner ordering.
To make our investigation more general and put the probabilities in
  \eqref{eq:Api} and \eqref{eq:Lpi} in an unified class, we consider the
  following general class of subsampling probabilities in the rest of the paper:
  \begin{equation}
    \label{eq:pi}
    \pi_{i}^{\os}(\beta_{0},\Phi)
    =\frac{\sqrt{b^{''}(X_{i}^{T}\beta_{0})}\| L\Phi^{-1}X_{i}\|}
    {\sumjn\sqrt{b^{''}(X_{j}^{T}\beta_{0})}\| L\Phi^{-1}X_{j}\|},
  \end{equation}
  where $L$ is a fixed matrix. Here the probabilities
  $\{\pi_{i}^{\os}(\beta_{0},\Phi)\}_{i=1}^n$ are optimal in that they
  minimize the asymptotic variance of $L\hbeta_{\rw}$. Specifically, when $L=I$,
  the probabilities in \eqref{eq:pi} reduce to those in \eqref{eq:Api} and when
  $L=\Phi$, they reduce to those in
  \eqref{eq:Lpi}.

We define our unweighted estimator as:
\begin{equation}\label{eq:def_beta_uw}
	\hbeta_{\ruw}:=\mathop{\arg\max}_{\beta}
  \oner\sumr\left\{Y_{i}^{*}X_{i}^{*{T}}\beta-b(X_{i}^{*{T}}\beta)\right\},
\end{equation}
where $(X_{i}^{*},Y_{i}^{*})$'s are sampled according to the probabilities in
\eqref{eq:pi}.

\subsection{Notation and main algorithm}

We first introduce some notations and the main algorithm. Recall that
$X_{1}^{n}:=(X_{1},X_{2},...,X_{n})$ and denote
$Y_{1}^{n}:=(Y_{1},Y_{2},...,Y_{n})$. For a vector $X\in\mathbb{R}^{p}$, we use
$\| X\|$ to denote its Euclidean norm. For a matrix
$A\in\mathbb{R}^{p\times p}$, we use $\lambda_{\min}(A)$ and $\lambda_{\max}(A)$
to denote its minimum and maxmum eigenvalues, respectively,
$\| A\|_{F}$ to denote its Frobenius norm, and $\text{tr}(A)$ to
denote its trace. For two positive semi-definite matrices $A$ and $B$, $A\geq B$
if and only if $A-B$ is positive semi-definite; this is known as the Loewner
ordering. For parameter $\beta$, we assume that $\beta$ takes values in a
compact set $\beta\in\mathbb{B}$. Now, we present the main algorithm in
Algorithm \ref{alg:algo}. Since the probabilities in \eqref{eq:pi} involve 
unknown quantities, $\beta_{0}$ and $\Phi$, we use pilot estimates to replace
them in Algorithm \ref{alg:algo}.
\begin{algorithm}[H]%
  \caption{Unweighted estimation for GLM under measurement constraints}
  \label{alg:algo}
  \begin{algorithmic}[1]
    \STATE Take a pilot subsample of size $r_{\rp}$:
    $\{(X_{i}^{*_{\rp}},Y_{i}^{*_{\rp}})\}_{i=1}^{r_{\rp}}$
    with simple random sampling from the full data set
    $\{(X_{i},Y_{i})\}_{i=1}^{n}$. Calculate the pilot estimate of $\beta_{0}$:
    \begin{equation*}
      \hbeta_{\rp}:=\mathop{\arg\max}_{\beta}\frac{1}{r_{\rp}}\sum_{i=1}^{r_{\rp}}\left\{Y_{i}^{*_{\rp}}X_{i}^{*_{\rp}{T}}\beta-b(X_{i}^{*_{\rp}{T}}\beta)\right\},
    \end{equation*}
    and the pilot estimate of $\Phi$:
    \begin{equation*}
      \hat{\Phi}_{\rp}:=\frac{1}{r_{\rp}}\sum_{i=1}^{r_{\rp}}b^{''}(X_{i}^{*_{\rp}
        {T}}\hbeta_{\rp})X_{i}^{*_{\rp}}X_{i}^{*_{\rp}{T}}.
    \end{equation*}
    \STATE Use $\hbeta_{\rp}$ and $\hat{\Phi}_{\rp}$ to
    replace $\beta_{0}$ and $\Phi$ in \eqref{eq:pi}, respectively, and calculate
    the sampling probabilities
    $\{\pi_{i}^{\os}(\hbeta_{\rp},\hat{\Phi}_{\rp})\}_{i=1}^{n}$.
    \STATE Obtain a subsample $\{(X_{i}^{*},Y_{i}^{*})\}_{i=1}^{r}$ of size $r$
    according to the sampling probabilities
    $\{\pi_{i}^{\os}(\hbeta_{\rp},\hat{\Phi}_{\rp})\}_{i=1}^{n}$
    using sampling with replacement, and solve the estimation equation:
    \begin{equation*}
      \Psi_{\ruw}^{*}(\beta):=\oner\sumr\{b^{'}(X_{i}^{*{T}}\beta)-Y_{i}^{*}\}X_{i}^{*}=0,
    \end{equation*}
    to obtain the unweighted estimator defined in \eqref{eq:def_beta_uw}.
  \end{algorithmic}
\end{algorithm}
\begin{remark}\label{rm:twoweights}
Our Algorithm~\ref{alg:algo} is different from the subsampling algorithm in
\cite{Zhang2020optimal} at step 3 of obtaining the subsampling estimators.
There are two types of weights in the subsampling algorithms: one is the sampling
weights which we call subsampling probabilities in this paper, and the other is
the estimation weights used to construct the target function.
Algorithm~\ref{alg:algo} and \cite{Zhang2020optimal}'s algorithm share the same
sampling probabilities (sampling weights) but they use
different estimation weights.
\cite{Zhang2020optimal} use the estimation weights $1/\pi_{i}^{\os}(\hbeta_{\rp},\hat{\Phi}_{\rp})$
while we set the estimation weights to be uniformly one, i.e., the target
function is unweighted.
We will show in Section~\ref{sec:efficiency} that our estimator improves the
estimation efficiency. This does not contradict the fact that
$\{\pi_{i}^{\os}\}_{i=1}^n$ are optimal for the algorithm in
\cite{Zhang2020optimal}, because they force the estimation weights to be the
inverses of the sampling weights while we do not enforce this requirement. 
\end{remark}

\begin{remark}\label{rm:computime}
The computational complexity of our two-step Algorithm~\ref{alg:algo} is the
same as the OSUMC estimator in \cite{Zhang2020optimal}, because we use the same
sampling probabilities and the two methods differ only in the weights of the
target function.
With Newton's method, it requires $O(\zeta_{\rp}r_{\rp}p^2)$ time to compute the
pilot estimates, where $\zeta_{\rp}$ is the number of iterations for the
algorithm to convergence based on the pilot sample.
The time complexities of calculating sampling probabilities
$\{\pi^{\mmse}\}_{i=1}^n$ and $\{\pi^{\mvc}\}_{i=1}^n$ are $O(np^2)$ and
$O(np)$, respectively.
After obtaining the second stage subsample with the optimal sampling
probabilities, it takes $O(\zeta rp^2)$ time to solve the unweighted target
function where $\zeta$ is the number of iterations of Newton's algorithm.  Thus,
the total computational time is $O(np^2+\zeta_{\rp}r_{\rp}p^2+\zeta rp^2)$ for
A-optimality and $O(np+\zeta_{\rp}r_{\rp}p^2+\zeta rp^2)$ for L-optimality.  The
computational complexity of our algorithm based on the A-optimality criterion is
the same as the OSUMC algorithm in \cite{Zhang2020optimal}.  Therefore, our
method increase the estimation efficiency without increasing the computational
burden.
\end{remark}

\subsection{Asymptotic normality of
  $\hbeta_{\ruw}$}\label{sec:asymptotic}

We focus on unconditional asymptotic results for the unweighted algorithm, and
use martingale techniques to prove
theorems. %
To present the asymptotic results, we summarize some regularity conditions
first.

\begin{assumption}\label{A1}
  The second derivative $b^{''}(\cdot)$ is bounded and continuous.
\end{assumption}
\begin{assumption}\label{A2}
  The fourth moment of the covariate is finite, i.e.,
  $\Exp\left(\| X\|^{4}\right)<\infty$.
\end{assumption}
\begin{assumption}\label{A3}
  Let $g(x):=\inf_{\beta\in\mathbb{B}} b^{''}(x^{T}\beta)$. Assume that
  $\lambda_{\min}[\Exp\{g(X)XX^{T}\}]>0$. Assume that there exists a
  function $h(x)$ such that $|b^{'''}(x^{T}\beta)|\leq h(x)$ and
  $\Exp\{h(X)\| X\|^{4}\}<\infty$, where $b^{'''}(\cdot)$ denotes
  the third derivative of $b(\cdot)$.
\end{assumption}

Assumption \ref{A1} is  commonly used in GLM literature, e.g.,
\cite{Zhang2020optimal}. Assumption \ref{A2} is a moment condition on
$X$. %
The second part of Assumption \ref{A3} is similar to the third-derivative
condition used in the classical theory of MLE. However, here we need a stronger
moment condition, $\Exp[h(X)\| X\|^{4}]<\infty$, since we use an
unequal probability sampling
method. %
Before we prove the asymptotic normality of $\hbeta_{\ruw}$, we need
to prove some lemmas. First, we present the convergence of
$\Dot{\Psi}_{\ruw}^{*}(\beta)$.
\begin{lemma}\label{lm:Convergence_of_Gamma}%
  Under assumptions A\ref{A1}-A\ref{A3}, for every sequence
  $\beta_{n}\xrightarrow{p}\beta_{0}$, %
  \begin{equation*}
    \hat{m}\Dot{\Psi}_{\ruw}^{*}(\beta_{n})
    \xrightarrow{p}\Gamma
    :=\Exp\left[\left\{b^{''}(X^{T}\beta_{0})\right\}^{\frac{3}{2}}
      \|L\Phi^{-1}X\| XX^{T}\right],
  \end{equation*}
  where
  $\hat{m}=(1/n)\sumn\sqrt{b^{''}(X_{i}^{T}\hbeta_{\rp})}\|L\hat{\Phi}_{\rp}^{-1}X_{i}\|$ and the notation
  ``$\xrightarrow{p}$'' denotes convergence in probability.
\end{lemma}
Furthermore, to establish the asymptotic normality of
$\hbeta_{\ruw}$, we present the asymptotic normality of
$\Psi_{\ruw}^{*}(\beta_{0})$.
\begin{lemma}\label{lm:Normality_of_Psi}%
  Under assumptions A\ref{A1}-A\ref{A3}, if $r/n\to\rho\in[0,1)$,
  $r_{\rp}/\sqrt{n}\to0$ and $\exists \delta>0$ such that
  \begin{equation}\label{eq:1}
    \Exp\left\{\left| b^{'}(X^{T}\beta_{0})-Y\right|^{4+2\delta}
      \left\| X\right\|^{8+4\delta}\right\}<\infty,
  \end{equation}
  then
  \begin{equation*}
    \sqrt{r}\hat{m}\Psi_{\ruw}^{*}(\beta_{0})
    \xrightarrow{d}N(0,m\Gamma+\rho \Omega),		
  \end{equation*}
  where
  $m:=\Exp\left\{\sqrt{b^{''}(X^{T}\beta_{0})}\|L\Phi^{-1}X\|\right\}$ and
  $\Omega=\Exp\left[\left\{b^{''}(X^{T}\beta_{0})\right\}^{2}\|L\Phi^{-1}X\|^{2}XX^{T}\right]$.  If specifically $\rho=0$, then the
  required condition in~\eqref{eq:1} can be weakened to
  \begin{equation*}
    \Exp\left\{\left| b^{'}(X^{T}\beta_{0})-Y\right|^{2+\delta}
      \left\| X\right\|^{4+2\delta}\right\}<\infty.
  \end{equation*}
\end{lemma}

We will use the central limit theorem for martingales described in
\cite{jakubowski1980limit} and \cite{Zhang2020optimal} to prove this Lemma in
the supplementary material. In Algorithm~\ref{alg:algo}, the pilot subsample and the
optimal subsample are from the same full data so the unconditional distributions
of the two subsamples are not independent and it is possible to have
overlaps. The assumption $r_{\rp}/\sqrt{n}\to0$ is to ensure that the data
points used in the pilot subsample are asymptotically negligible when deriving
the unconditional asymptotic distribution of $\Psi_{\ruw}^{*}(\beta_{0})$ which
depends on both subsamples. This assumption can be replaced by other
alternatives such as assuming that the pilot estimator is independent of the full data
\citep[e.g.,][]{fithian2014local} and this is appropriate if we modify step 3 of
Algorithm~\ref{alg:algo} to sample from the rest of the data with the pilot
subsample data points removed.

Now, we are ready to show the asymptotic normality of the unweighted estimator.
\begin{theorem}\label{th:Normality}
  Under assumptions A\ref{A1}-A\ref{A3}, assuming that $\Gamma$ is positive-definite, we have
  \begin{equation*}
    \hbeta_{\ruw}-\beta_{0}
    =-\Gamma^{-1}\hat{m}\Psi_{\ruw}^{*}(\beta_{0})
    +o_{p}\left(1/\sqrt{r}\right).
  \end{equation*}
  In addition, under the conditions of Lemma \ref{lm:Normality_of_Psi}
  \begin{equation*}
    \sqrt{r}(\hbeta_{\ruw}-\beta_{0})
    \xrightarrow{d}N(0, \Sigma_{\ruw}^{\rho}),
  \end{equation*}
  where %
  \begin{equation}\label{eq:sigma_uw_rho}
    \Sigma_{\ruw}^{\rho}
    :=m\Gamma^{-1}+\rho\Gamma^{-1}\Omega\Gamma^{-1}.
  \end{equation}

\end{theorem}
Theorem \ref{th:Normality} shows that $\hbeta_{\ruw}$ is asymptotically
unbiased, and from \eqref{eq:sigma_uw_rho} %
we see that the asymptotic variance of $\hbeta_{\ruw}$ can be
split into two parts, $m\Gamma^{-1}$ and
$\rho\Gamma^{-1}\Omega\Gamma^{-1}$. Here, $m\Gamma^{-1}$ is the contribution
from the randomness of subsampling and $\rho\Gamma^{-1}\Omega\Gamma^{-1}$ is due
to the randomness of the full data. If the subsample size $r$ is of a smaller
order than the full data sample size $n$, i.e., $\rho=0$, then the randomness of
the full data is negligible. If $r$ is of the same order as $n$, we need a
stronger moment condition (as stated in Lemma~\ref{lm:Normality_of_Psi}) to
establish asymptotic normality.
In the subsampling setting, we usually expect $r\ll n$, and therefore
$m\Gamma^{-1}$ %
is the dominating term of the asymptotic variance of
$\hbeta_{\ruw}$. %

To estimate the asymptotic variance of
$\hbeta_{\ruw}$, %
we propose the following formulas involving only the selected
subsample: %
\begin{equation}\label{eq:estMSE}
  \hat{\Var}(\hbeta_{\ruw})
  =\oner\hat{m}\hat{\Gamma}^{-1}
  +\onen\hat{\Gamma}^{-1}\hat{\Omega}\hat{\Gamma}^{-1},
\end{equation}
where
\begin{equation*}
  \hat{\Gamma}=\frac{\hat{m}}{r}\sumr
  b^{''}(X_{i}^{*{T}}\hbeta_{\ruw})X_{i}^{*}X_{i}^{*{T}},
\end{equation*}
and
\begin{equation*}
  \hat{\Omega}=\frac{n\hat{m}^{2}}{r}\sumr
  \pi_{i}^{*}b^{''}(X_{i}^{*{T}}\hbeta_{\ruw})X_{i}^{*}X_{i}^{*{T}}.
\end{equation*}
Our estimator of the asymptotic variance follows the similar idea that is
described in \cite{WangZhuMa2017} and \cite{wang2019more}.

\subsection{Efficiency of the unweighted estimator}
\label{sec:efficiency}

In this section, we compare the efficiency of the unweighted estimator
$\hbeta_{\ruw}$ with the weighted estimator
$\hbeta_{\rw}$ defined in~\eqref{eq:def_beta_w}. We first restate
the asymptotic result in \cite{Zhang2020optimal}. In their paper, they proved
that under some regularity conditions, $\hbeta_{\rw}$ is
asymptotically normal
\begin{equation*}
	\Var\{\Psi^{*}_{\rw}(\beta_{0})\}^{-\frac{1}{2}}\Phi(\hbeta_{\rw}-\beta_{0})\xrightarrow{d}N(0,I),
\end{equation*}
where
\begin{equation*}
  \Var\{\Psi^{*}_{\rw}(\beta_{0})\}=\Exp
  \left[\frac{1}{n^{2}}\sumn b^{''}(X_{i}^{T}\beta_{0})X_{i}X_{i}^{T}\left\{\frac{1}{r\pi_{i}}-\oner+1\right\}\right].
\end{equation*}
Denote
\begin{equation*}
  \Lambda:=\Exp\left\{\frac{b^{''}(X^{T}\beta_{0})XX^{T}}{\sqrt{b^{''}(X^{T}\beta_{0})}\|L\Phi^{-1}X\|}\right\},
\end{equation*}
and replace $\pi=\{\pi_{i}\}_{i=1}^{n}$ in
$\Var\{\Psi^{*}_{\rw}(\beta_{0})\}$ with the optimal sampling
probabilities defined in \eqref{eq:pi}. We then have that
\begin{equation}\label{eq:weightedvariance}
  \Var\{\Psi^{*}_{\rw}(\beta_{0})\}%
  =\oner{\frac{n-1}{n}}m\Lambda+\onen\Phi,
\end{equation}
where $m$ is defined in Lemma \ref{lm:Normality_of_Psi} and $\Phi$ is defined in
\eqref{eq:phi}. The details of the calculation are presented in the
supplementary material. From \eqref{eq:weightedvariance}, if $r/n\to\rho$, the asymptotic
variance of $\sqrt{r}(\hbeta_{\rw}-\beta_{0})$ is
\begin{equation}\label{eq:sigma_w_rho}
  \Sigma_{\rw}^{\rho}
  :=m\Phi^{-1}\Lambda\Phi^{-1}+\rho\Phi^{-1}.
\end{equation}
The asymptotic variance $\Sigma_{\rw}^{\rho}$ consists of two parts: the
term $m\Phi^{-1}\Lambda\Phi^{-1}$ is due to the randomness of subsampling while
the term $\rho\Phi^{-1}$ is due to the randomness of the full data. Similarly,
in the asymptotic variance $\Sigma_{\ruw}^{\rho}$ defined in
(\ref{eq:sigma_uw_rho}) for the unweighted estimator, $m\Gamma^{-1}$ is due to
the randomness of subsampling and $\rho\Gamma^{-1}\Omega\Gamma^{-1}$ is due to
the randomness of the full data. We have the following results comparing the
aforementioned terms for the weighted and unweighted estimators.
\begin{theorem}\label{th:Efficiency}
  If $\Phi$, $\Gamma$ and $\Lambda$ are finite and positive-definite, then
  \begin{equation*}
    \Gamma^{-1}\leq \Phi^{-1}\Lambda\Phi^{-1},\quad
    and\quad\Gamma^{-1}\Omega\Gamma^{-1}\geq \Phi^{-1},
  \end{equation*}
  where the inequalities are in the Loewner ordering.
\end{theorem}
From Theorem~\ref{th:Efficiency}, $m\Gamma^{-1}\leq
m\Phi^{-1}\Lambda\Phi^{-1}$. Thus, compared with the weighted estimator, the
unweighted estimator has a smaller asymptotic variance component from the
randomness of subsampling. On the other hand, since
$\rho\Gamma^{-1}\Omega\Gamma^{-1}\geq \rho\Phi^{-1}$, the asymptotic variance
component due to the full data randomness is larger for the unweighted
estimator. A major motivation of subsampling is to reduce the computational or
data measurement cost significantly, so it is typical that $r\ll n$ and 
therefore $\rho$ is typically very small. In this scenario, the asymptotic
variance component due to subsampling is the dominating term, and the unweighted
estimator has a higher estimation efficiency than the weighted estimator. In the
case that $r/n\to0$, the asymptotic variance component due to full data
randomness is negligible.

We can also get some insights on the difference between the weighted and
unweighted estimators by considering them conditionally on the full data. Given
the full data, the subsample weighted estimator $\hbeta_{\rw}$ is asymptotically
unbiased for the full data unweighted MLE $\hbeta_{\text{MLE}}$ in
\eqref{eq:3},
while the subsample unweighted estimator $\hbeta_{\ruw}$ is asymptotically
unbiased for the full data weighted MLE defined as
\begin{equation*}
  \hbeta_{\text{wMLE}}:=\mathop{\arg\max}_{\beta}\onen\sumn{w_i}
  \left\{Y_{i}X_{i}^{T}\beta-b(X_{i}^{T}\beta)\right\},
\end{equation*}
where $w_i=\sqrt{b^{''}(X_i^T\beta_{0})}\|L\Phi^{-1}X_i\|$ does not
depend on the $\{Y_i\}$. Here, $\hbeta_{\text{wMLE}}$ is asymptotically unbiased
for the true parameter because the weights $\{w_i\}$ are only related to
the $\{X_i\}$.  We see that $\hbeta_{\rw}$ and $\hbeta_{\ruw}$ essentially approximate
different full data estimators $\hbeta_{\text{MLE}}$ and $\hbeta_{\text{wMLE}}$,
respectively. It is well known that $\hbeta_{\text{MLE}}$ is more efficient than
$\hbeta_{\text{wMLE}}$ based on the full data, but its variation is much
smaller than that of $\hbeta_{\rw}$ or $\hbeta_{\ruw}$, and it is negligible if
$r/n\rightarrow0$. Thus the variation of $\hbeta_{\rw}$ around
$\hbeta_{\text{MLE}}$ and the variation of $\hbeta_{\ruw}$ around
$\hbeta_{\text{wMLE}}$ are the major components of the asymptotic variances of
$\hbeta_{\rw}$ and $\hbeta_{\ruw}$ in terms of estimating the true parameter.

When the model is correctly specified, then $\hbeta_{\rw}$ and $\hbeta_{\ruw}$
are consistent for the same true parameter. However, if the model is
misspecified, then $\hbeta_{\rw}$ and $\hbeta_{\ruw}$ will typically converge
to different limits. Heuristically, $\hbeta_{\rw}$ will converge to the
solution of $\Exp[X\{Y-b^{'}(X^{T}\beta)\}]$ while $\hbeta_{\ruw}$ will converge
to the solution of $\Exp[wX\{Y-b^{'}(X^{T}\beta)\}]$ with
$w=\sqrt{b^{''}(X^T\beta_0)}\|L\Phi^{-1}X\|$. In this scenario, it is
difficult to compare the efficiency of $\hbeta_{\rw}$ with that of
$\hbeta_{\ruw}$, because it is unknown which solution is closer to the true
data-generating parameter.

\section{NUMERICAL RESULTS} \label{sec:numerical}

We investigate the efficiency of the unweighted estimator in parameter
estimation through numerical experiments in this section. We present simulation
results in Section \ref{sec:simulation} and experiments for real data in Section
\ref{sec:realdata}.

\subsection{Simulation Results}\label{sec:simulation}

In this section, we use simulations to evaluate the performance of the more
efficient estimator we proposed. To compare with the original OSUMC estimator,
we use the same setups as described in Section 5 and in the appendix of
\cite{Zhang2020optimal}, and show numerical results for logistic, Poisson and
linear regressions.

\subsubsection{Logistic Regression and Poisson regression}\label{sec:simulog}

We first present simulations for logistic regression for which the conditional
density of the response has the form
\begin{equation*}
f(y|x,\beta_0)=\exp\left\{yx^T\beta_0-\log(1+e^{x^T\beta_0})\right\},
\text{ for } y= 0, 1.
\end{equation*}
This model implies that the probability of $Y=1$ given $X$ is
\begin{equation*}
P(Y=1|X,\beta_0)=\frac{e^{X^T\beta_0}}{1+e^{X^T\beta_0}}.
\end{equation*}
To generate the full data, we set the true parameter $\beta_0$ as a 20
dimensional vector with all entries equal to 1. The full data sample size
is $n=100,000$ and four distributions of $X$ are
considered, which are exactly the same distributions used in
\cite{Zhang2020optimal}. We present these four covariate distributions below for
completeness:
\begin{enumerate}
\item \textbf{mzNormal}: The covariate $X$ follows a multivariate normal
  distribution $N(0,\Sigma)$, where $\Sigma_{ij}=0.5^{I{(i\neq j)}}$ and
  $I(\cdot)$ represents the indicator function. We have almost equal numbers of
  1's and 0's in this scenario.
\item \textbf{nzNormal}: The covariate $X$ follows a multivariate normal
  distribution $N(0.5,\Sigma)$, where $\Sigma$ is defined in mzNormal. In this
  scenario, roughly 75\% of the responses are 1's.
\item \textbf{unNormal}: The covariate $X$ follows a multivariate normal
  distribution $N(0,\Sigma_{1})$, where $\Sigma_{1}=U_{1}\Sigma U_{1}$,
  $U_{1}=diag(1,1/2,...,1/20)$ and $\Sigma$ is the same covariance matrix as we
  used in mzNormal. For this case, the components of $X$ have different
  variances.
\item \textbf{mixNormal}: The covariate $X$ follows a mixed multivariate normal
  distribution, namely, $X\sim0.5N(0.5,\Sigma)+0.5N(-0.5,\Sigma)$, where
  $\Sigma$ is the same as what we used in mzNormal.
\end{enumerate}
To compare the performance of the new estimator with the weighted one, we use
the empirical MSE of $\hbeta$:
\begin{equation}\label{eq:MSE}
	\text{eMSE}(\hbeta)=\frac{1}{S}\sum_{s=1}^{S}\|\hbeta^{(s)}-\beta_{0}\|.
\end{equation}
Here, $\hbeta^{(s)}$ is the estimated parameter we obtained in the $s$-th
repetition of the simulation. We repeated the simulation for $S=500$ times to
calculate $\text{eMSE}(\hbeta)$. For the pilot estimate, we used
$r_{\rp}=500$ for both weighted and unweighted methods. In every
repetition, we generated the full data, which means we focus on the
unconditional empirical MSE. Figure \ref{fig:unconditionMSE_logistic} shows that
our unweighted estimator performs better than the original OSUMC weighted
estimator under each setting when applied to logistic regression. This is
  true for both A-optimality and L-optimality criteria. For instance, when $X$ has a
mixNormal distribution, the emprical MSE of the weighted estimator is over 1.15
times as large as that of the unweighted one. In most cases,
  $\pi_i^{\mmse}$ and $\pi_i^{\mvc}$ perform similarly. When $X$ has a unNormal
  design, $\pi_i^{\mmse}$ performs significantly better than $\pi_i^{\mvc}$
  because the A-optimality aims to directly minimize asymptotic MSE.

To evaluate the performance of \eqref{eq:estMSE} in estimating the asymptotic
variance, we compare $\text{tr}\{\hat{\Var}(\hbeta_{\ruw})\}$ with the
empirical variance. Figure \ref{fig:unconditionEMSE_logistic} shows that the
estimated variances are very close to the empirical variances under the logistic
regression model.
\begin{figure}[H]
  \centering
  \begin{subfigure}{0.48\textwidth}
    \includegraphics[width=\textwidth]{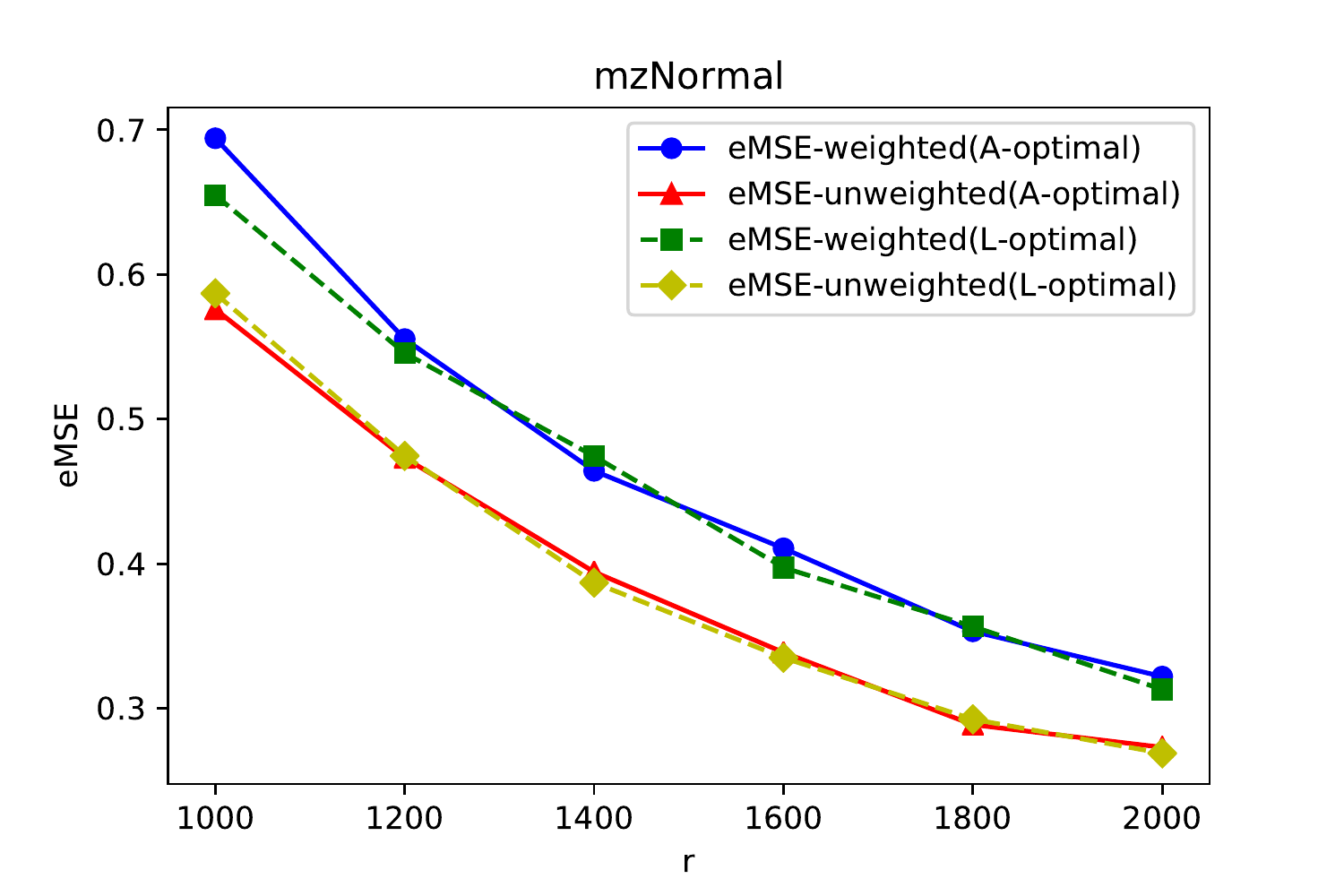}
    \caption{mzNormal}
  \end{subfigure}
  \begin{subfigure}{0.48\textwidth}
    \includegraphics[width=\textwidth]{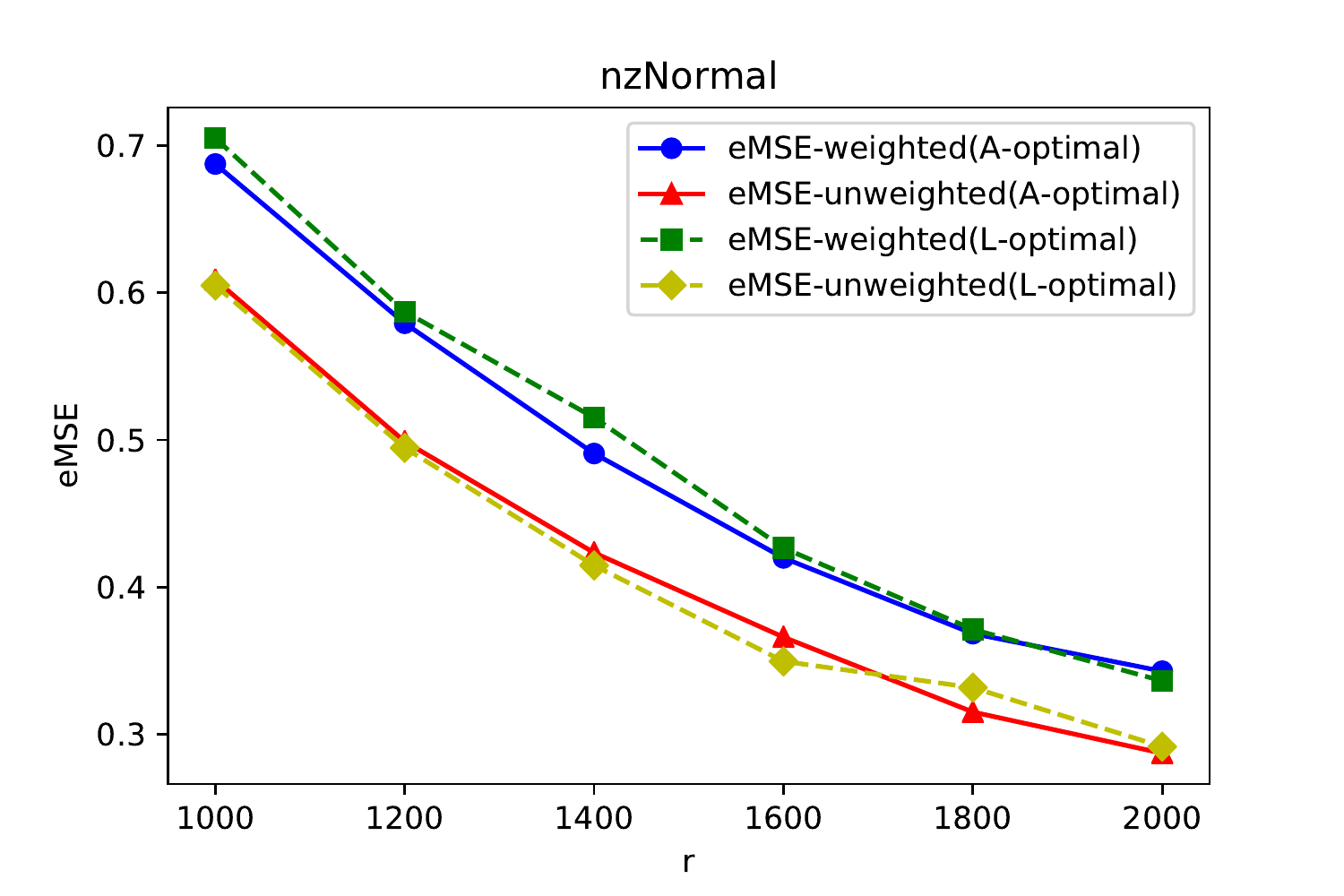}
    \caption{nzNormal}
  \end{subfigure}
  \begin{subfigure}{0.48\textwidth}
    \includegraphics[width=\textwidth]{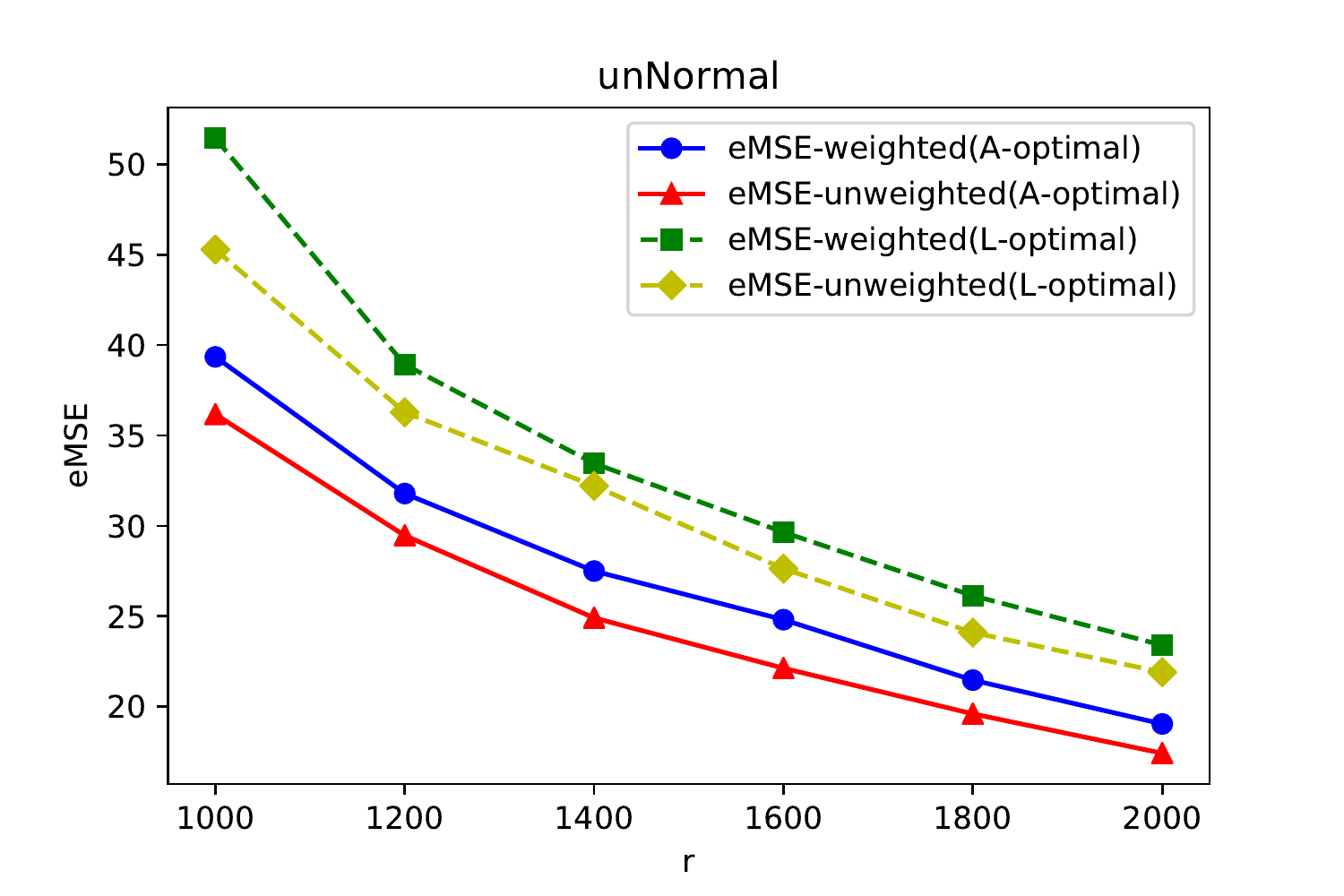}
    \caption{unNormal}
  \end{subfigure}
  \begin{subfigure}{0.48\textwidth}
    \includegraphics[width=\textwidth]{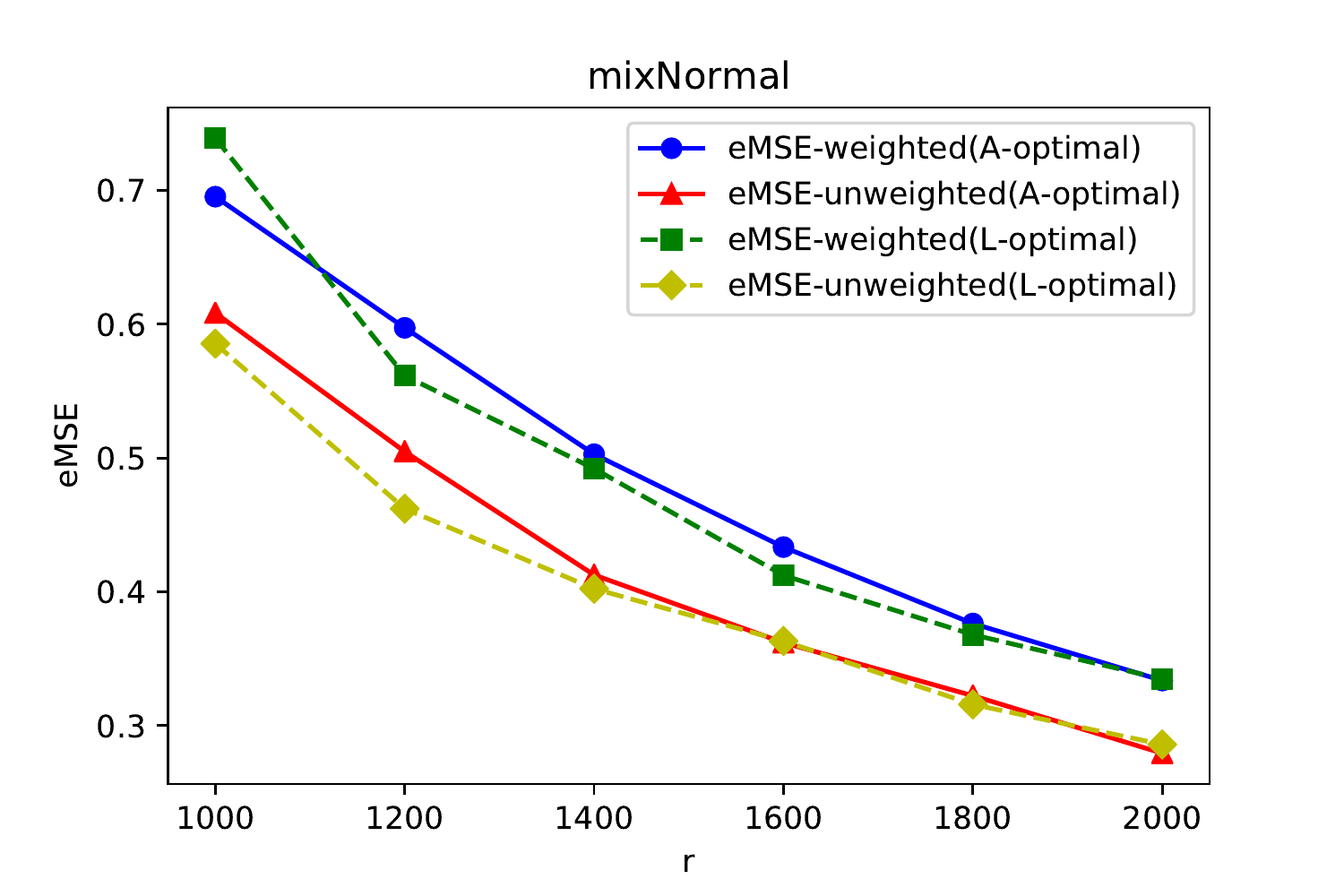}
    \caption{mixNormal}
  \end{subfigure}
  \caption{eMSE for different subsample sizes $r$ with a pilot sample size
    $r_{\rp}=500$ for logistic regression under different settings.}
  \label{fig:unconditionMSE_logistic}
\end{figure}
\begin{figure}[H]
	\centering
	\begin{subfigure}{0.48\textwidth}
		\includegraphics[width=\textwidth]{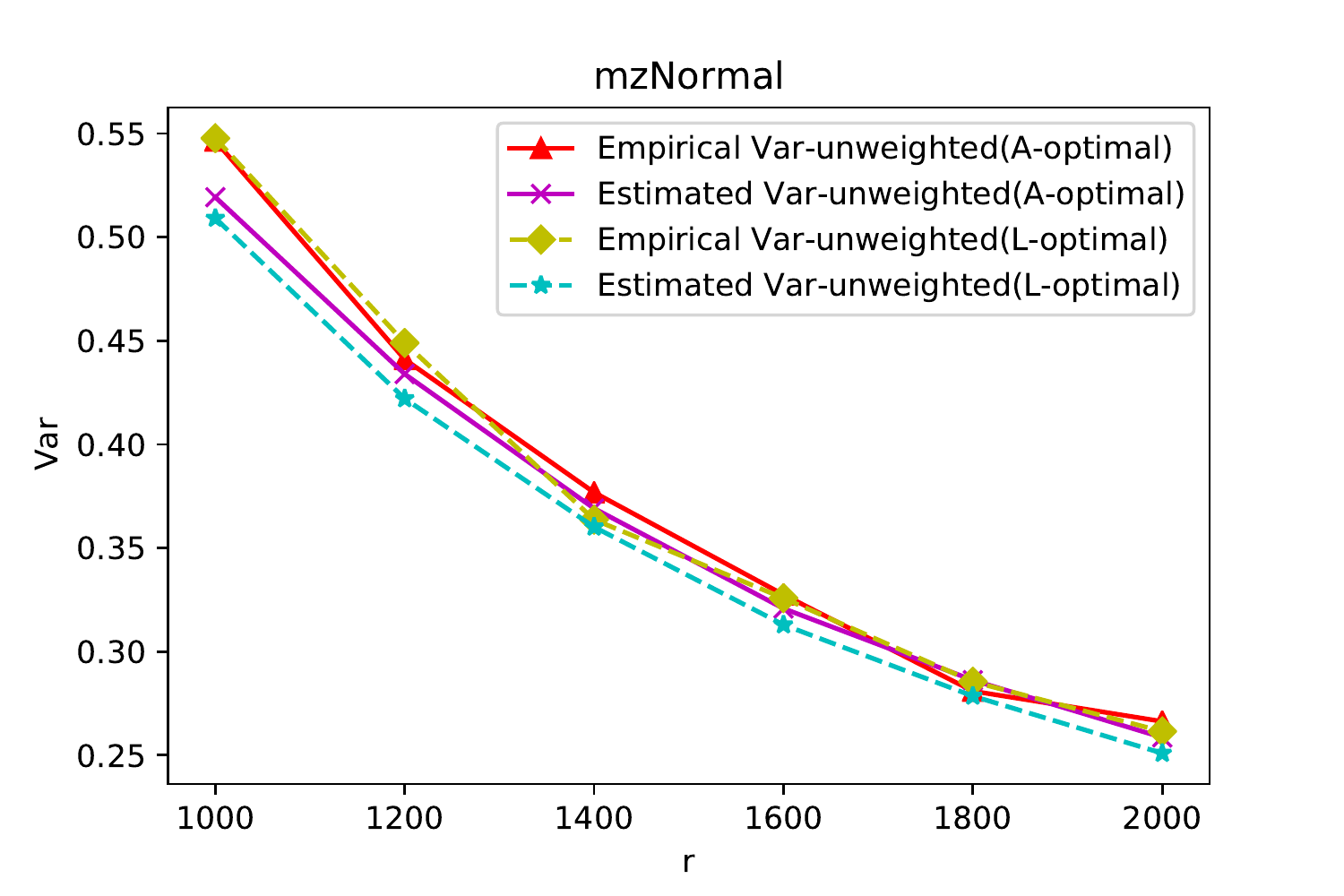}
		\caption{mzNormal}
	\end{subfigure}
	\begin{subfigure}{0.48\textwidth}
		\includegraphics[width=\textwidth]{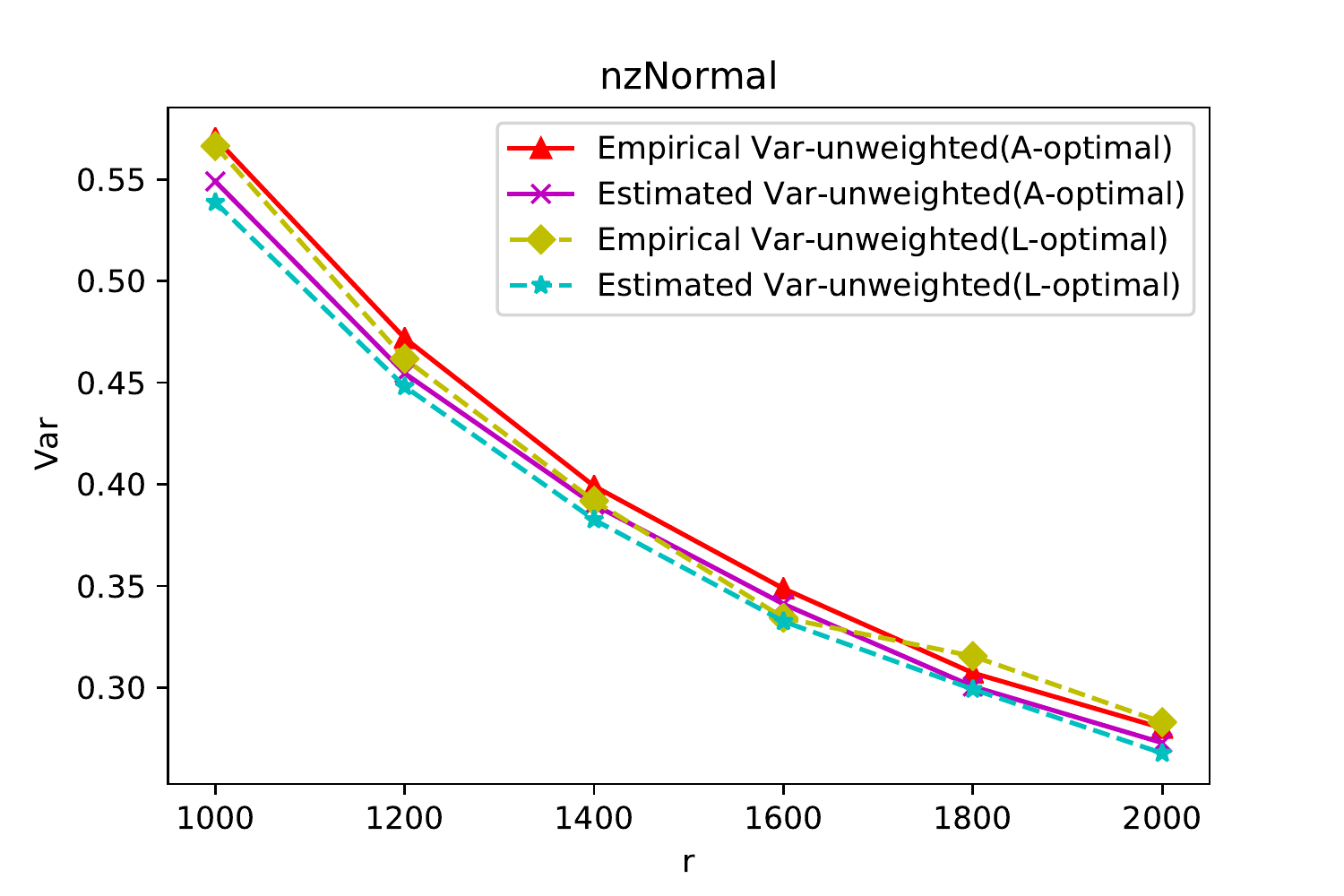}
		\caption{nzNormal}
	\end{subfigure}
	\begin{subfigure}{0.48\textwidth}
		\includegraphics[width=\textwidth]{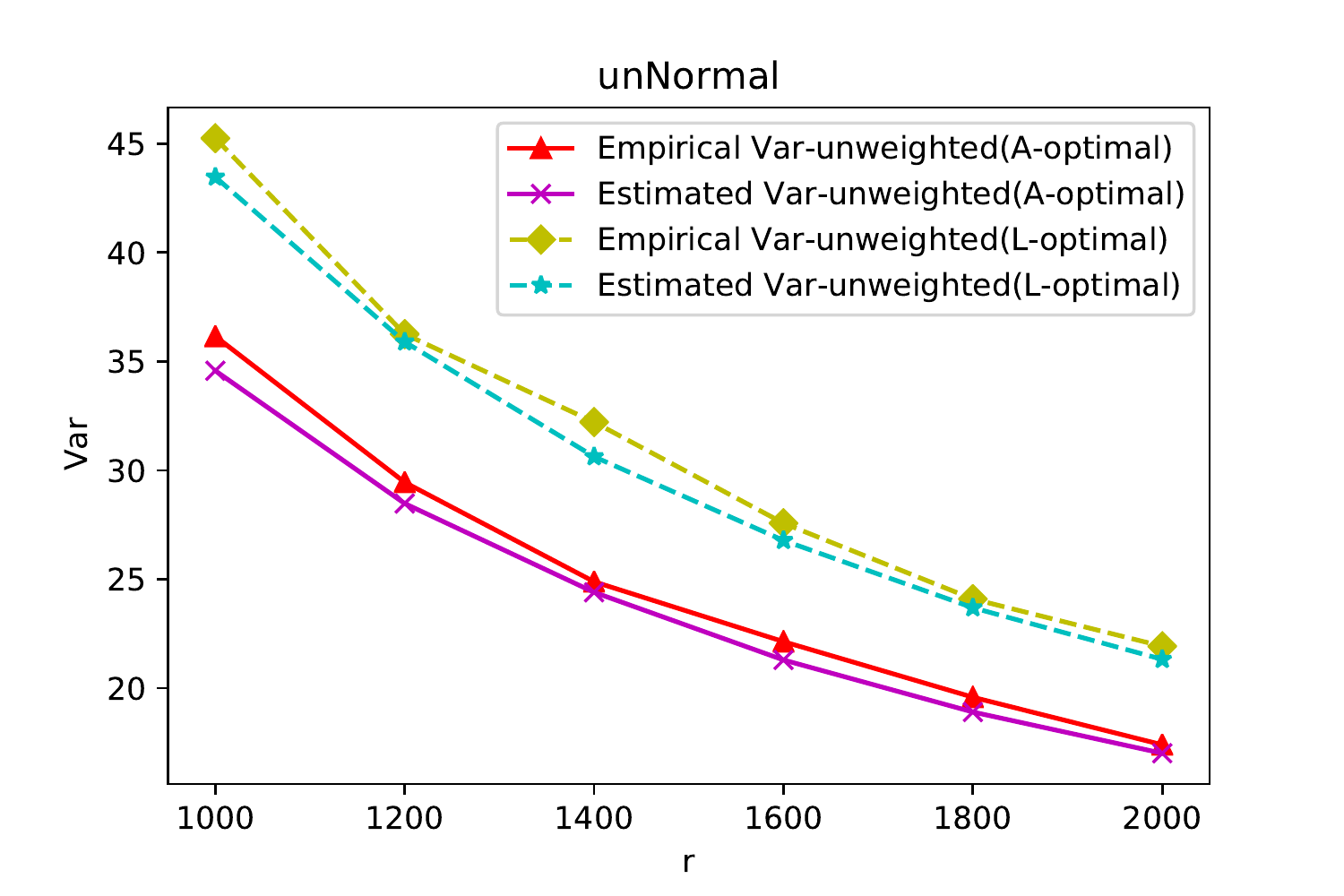}
		\caption{unNormal}
	\end{subfigure}
	\begin{subfigure}{0.48\textwidth}
		\includegraphics[width=\textwidth]{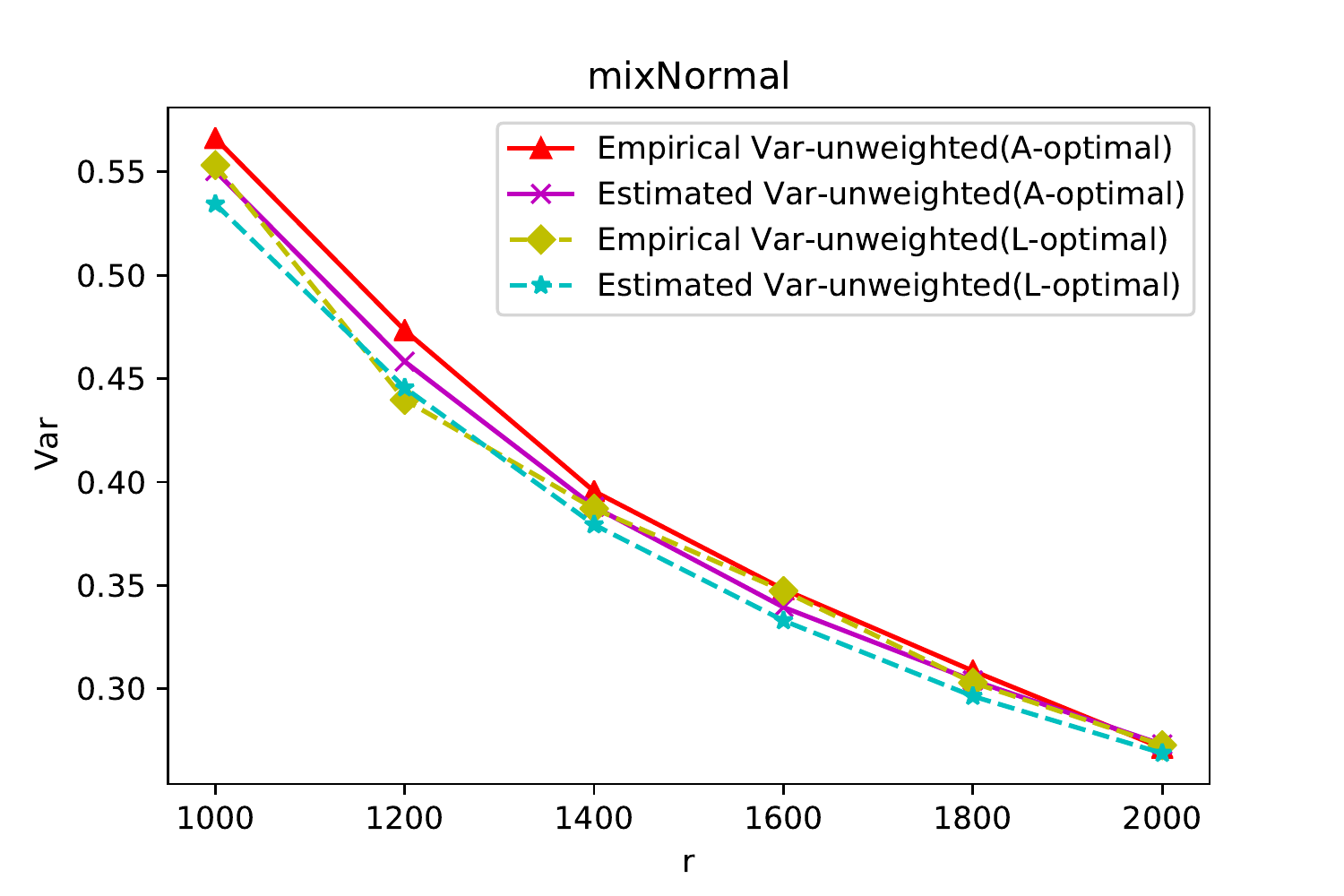}
		\caption{mixNormal}
	\end{subfigure}
	\caption{Empirical variance and estimated variance,
    $\text{tr}\{\hat{\Var}(\hbeta_{\ruw})\}$, for different subsample sizes $r$
    with a pilot sample size $r_{\rp}=500$ for the unweighted estimator under
    different settings.}
	\label{fig:unconditionEMSE_logistic}
\end{figure}

Performances of the unweighted estimator under the Poisson regression are also
investigated. The Poisson regression model has a form of
\begin{equation*}
  f(y|x,\beta_0)=\exp\left\{ yx^T\beta_0-e^{x^T\beta_0}-\log(y!) \right\},
  \text{ for } y = 0, 1, 2, ...
\end{equation*}
We generated $n=100,000$ data points. A $100\times 1$ vector of
$0.5$ is used as the true value of the parameter, $\beta_{0}$, in this
scenario. We use the same settings discussed in the appendix of
\cite{Zhang2020optimal}. Specifically, covariates are generated using
the following two settings:
\begin{enumerate}
\item \textbf{Case 1}: Each component of $X$ is generated independently from the
  uniform distribution over $[-0.5, 0.5]$.
\item \textbf{Case 2}: First half of the components of $X$ are generated
  independently from the uniform distribution over $[-0.5, 0.5]$, and the other
  half of the components of $X$ are generated indepedently from the uniform
  distribution over $[-1, 1]$.
\end{enumerate}
Again we repeated the experiment for $S=500$ times and in each repetition we
sampled $r_{\rp}=500$ data points to obtain pilot estimates. We also
compared the empirical MSE defined in \eqref{eq:MSE} and calculated
$\text{tr}\{\hat{\Var}(\hbeta_{\ruw})\}$ to investigate the
performance of the estimated variance defined in \eqref{eq:estMSE}. Empirical
MSEs of the unweighted and weighted estimators are presented in Figure
\ref{fig:uncondition_MSE_poisson}. For Poisson regression, our unweighted
  estimator also outperforms the weighted OSUMC estimator under both criteria,
  and $\pi_i^{\mmse}$ and $\pi_i^{\mvc}$ perform similarly. For Case
  1, The empirical MSE of the weighted estimator is around 1.5 times as large as
  that of the unweighted estimator we proposed. For Case 2, the empirical MSE of
  our estimator is about half of that of the weighted estimator. The results for
  the estimated variances are presented in Figure
  \ref{fig:unconditionEMSE_poisson}. The estimated variance we proposed in
  \eqref{eq:estMSE} also works well under the Poisson model.
\begin{figure}[H]
	\centering
	\begin{subfigure}{0.48\textwidth}
		\includegraphics[width=\textwidth]{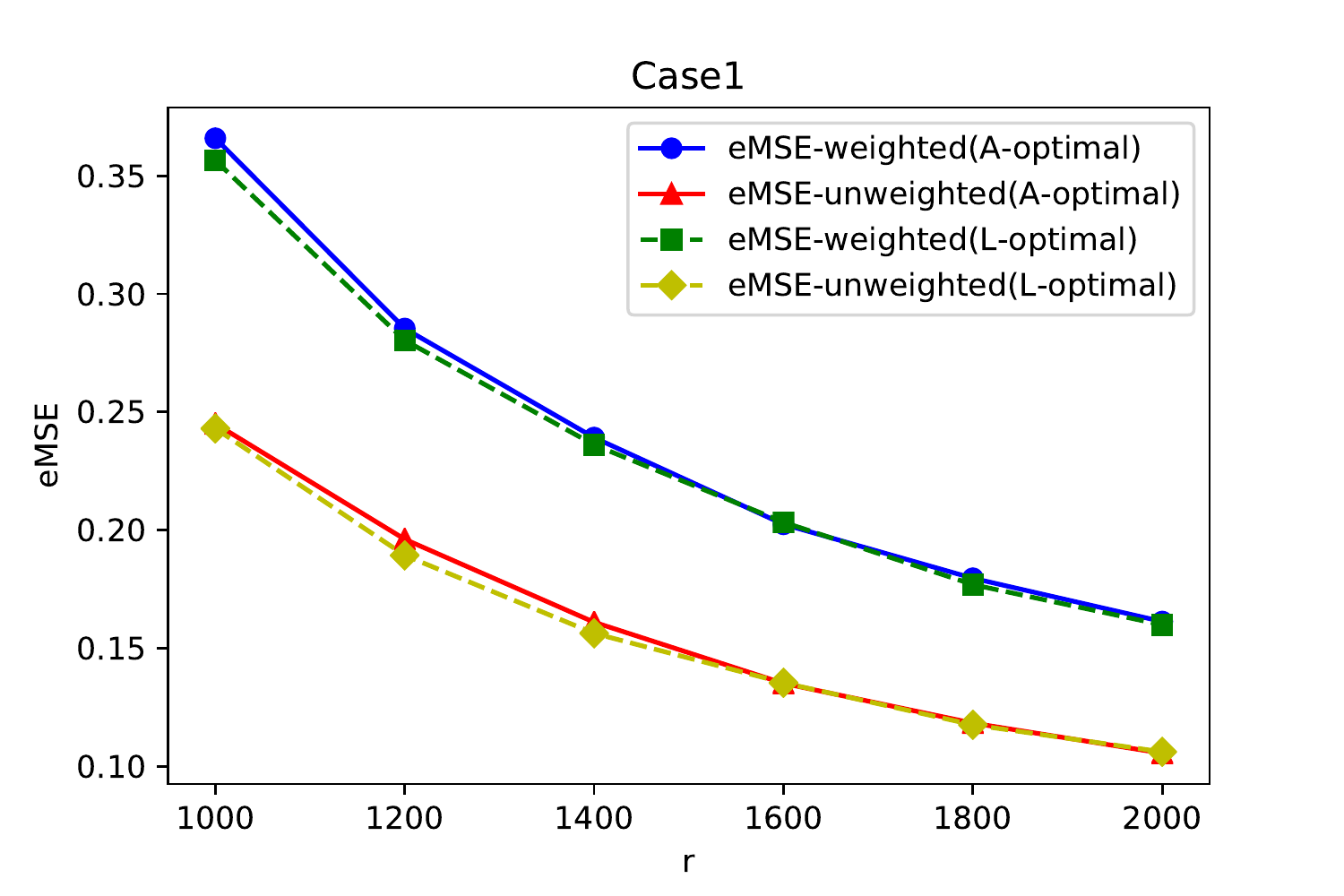}\\[-9mm]
		\caption{Case1}
	\end{subfigure}
	\begin{subfigure}{0.48\textwidth}
		\includegraphics[width=\textwidth]{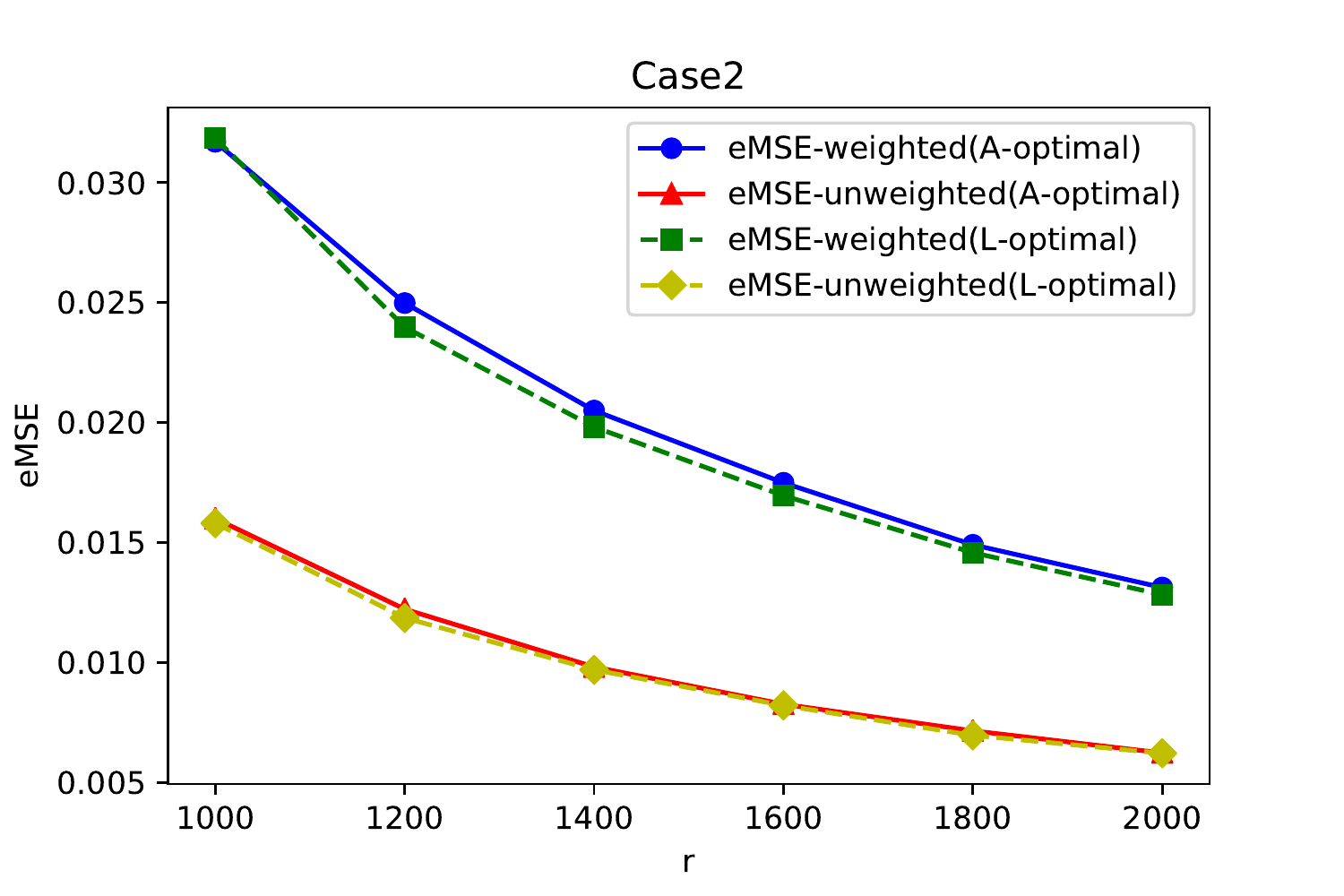}\\[-9mm]
		\caption{Case2}
	\end{subfigure}
	\caption{eMSE for different subsample sizes $r$ with a pilot sample size
    $r_{\rp}=500$ for Poisson regression under different settings.}
	\label{fig:uncondition_MSE_poisson}
\end{figure}
\begin{figure}[H]
	\centering
	\begin{subfigure}{0.48\textwidth}
		\includegraphics[width=\textwidth]{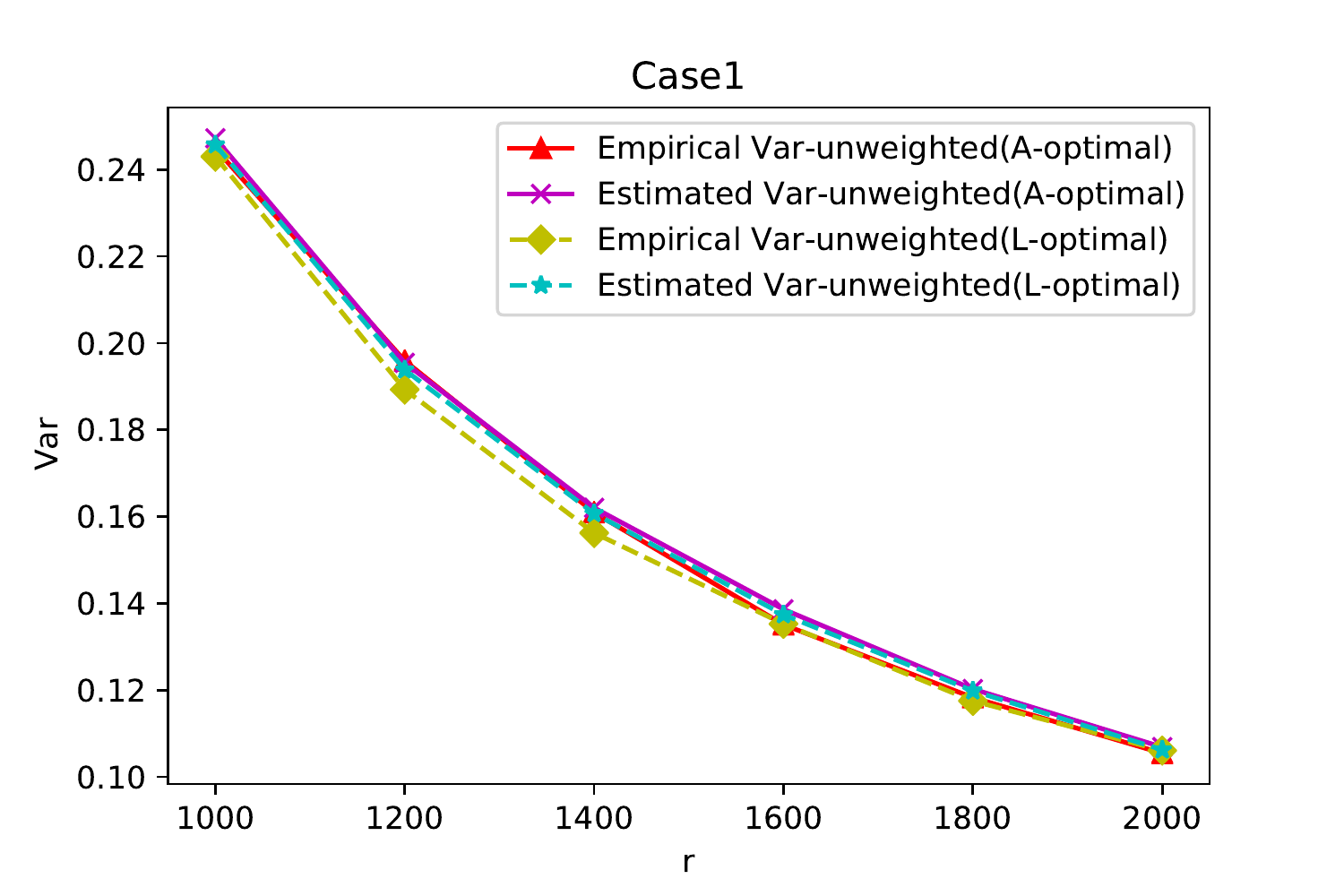}\\[-9mm]
		\caption{Case1}
	\end{subfigure}
	\begin{subfigure}{0.48\textwidth}
		\includegraphics[width=\textwidth]{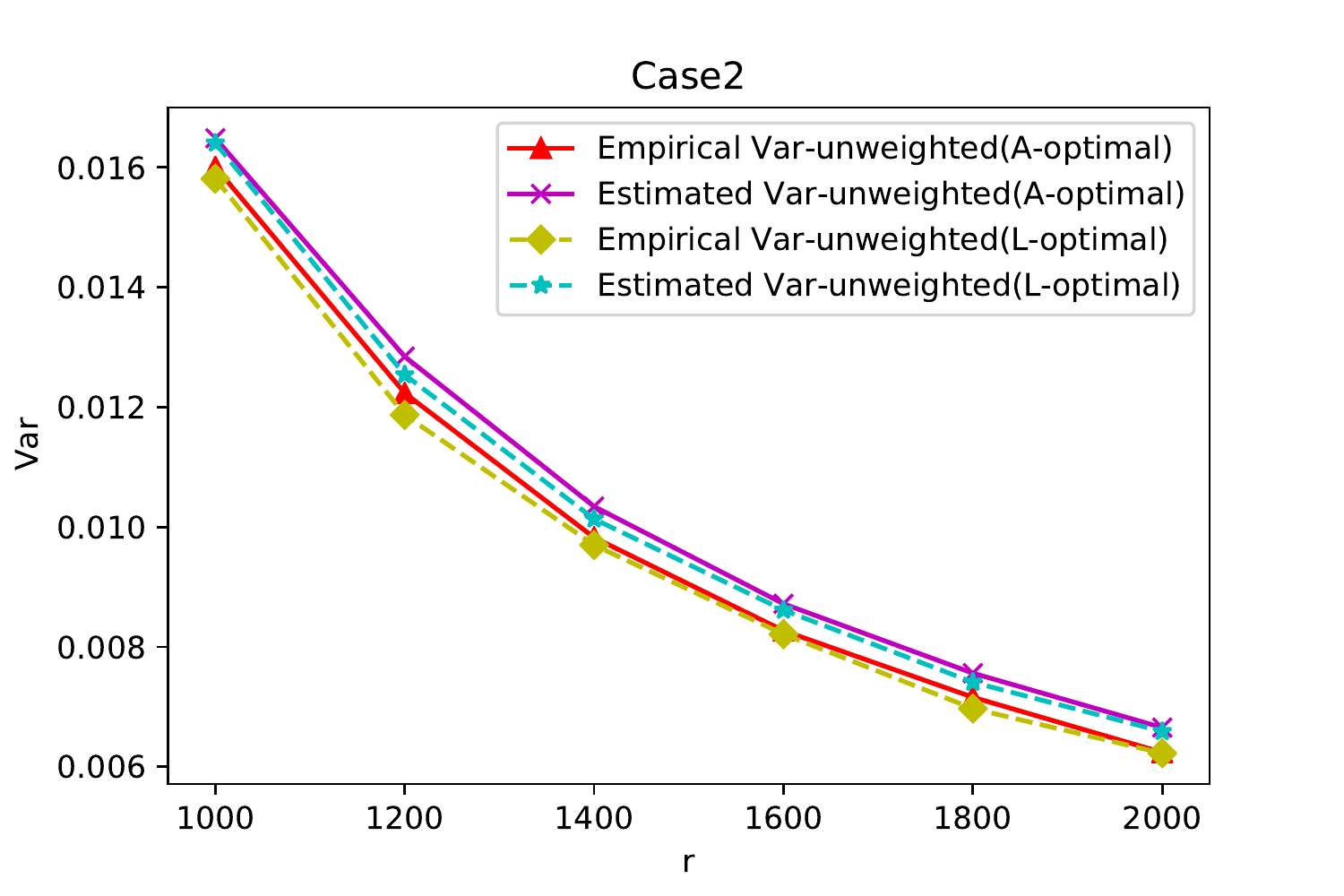}\\[-9mm]
		\caption{Case2}
	\end{subfigure}
	\caption{Empirical variance and estimated variance,
    $\text{tr}\{\hat{\Var}(\hbeta_{\ruw})\}$, for different
    subsample sizes $r$ with a pilot sample size $r_{\rp}=500$ for the
    unweighted estimator under different settings.}
	\label{fig:unconditionEMSE_poisson}
\end{figure}

\subsubsection{Linear Model}

We now present simulation results for linear regression. We used the
settings in \cite{Zhang2020optimal} which generated full data of size
$n=100,000$ from the following model:
\begin{equation*}
	Y=X\beta_{0}+\epsilon,
\end{equation*}
where
$\beta_{0}=(\underbrace{0.1,...,0.1}_\text{5},\underbrace{10,...,10}_\text{20},\underbrace{0.1,...,0.1}_\text{5})^{T}$
is a 30 dimensional vector, and $\epsilon\sim N(0,9I_{n})$. We used the
following distributions of $X$:
\begin{enumerate}
\item \textbf{GA}: The covariate $X$ follows a multivariate normal distribution
  $N(1_{p},\Sigma_{2})$, where $p=30$, $\Sigma_{2}=U_{2}\Sigma U_{2}$ and
  $U_{2}=diag(5,5/2,...,5/30)$. The entries of $\Sigma$ are
  $\Sigma_{ij}=0.5^{I(i\neq j)}$, which is the same as we defined before.
\item \textbf{T3}: The covariate $X$ follows a multivariate t-distribution which
  has degrees of freedom 3, $T_{3}(0,\Sigma_{2})$, and $\Sigma_{2}$ is defined
  in GA above.
\item \textbf{T1}: The covariate $X$ follows a multivariate t-distribution which
  has degrees of freedom 1, $T_{1}(0,\Sigma_{2})$, and $\Sigma_{2}$ is the same
  as GA.
\item \textbf{EXP}: Components of $X$ are i.i.d. from an exponential
  distribution with a rate parameter of 2.
\end{enumerate}

The first three settings are exactly the same settings used in
\cite{Zhang2020optimal}. The last setting is used in \cite{WangYangStufken2018}
and \cite{wang2019more}. Since the sampling probabilities are not related to the
responses for linear models, Algorithm \ref{alg:algo} can be simplified. For
completeness, we present the simplified algorithm as Algorithm
\ref{alg:algolinear}, which is similar to the algorithm used in
\cite{ma2015statistical}.

\begin{algorithm}[tbp]%
  \caption{Unweighted estimation for linear model under measurement constraints}
  \label{alg:algolinear}
  \begin{algorithmic}[1]
    \STATE Caculate the sampling probabilities $\{\pi_{i}\}_{i=1}^{n}$ using the
    following formula:
    \begin{equation*}
      \pi_{i}^{\mmse}=\frac{\left\|\left(\sumjn
            X_{j}X_{j}^{T}\right)^{-1}X_{i}\right\|}{\sum_{k=1}^{n}\left\|\left(\sumjn
            X_{j}X_{j}^{T}\right)^{-1}X_{k}\right\|}
    \end{equation*}
    \STATE Obtain a subsample $\{(X_{i}^{*},Y_{i}^{*})\}_{i=1}^{r}$ of size $r$
    according to the sampling probabilities
    $\{\pi_{i}^{\mmse}\}_{i=1}^{n}$ using sampling with replacement, and
    solve the estimation equation
    \begin{equation*}
      \Psi_{\ruw}^{*}(\beta):=\oner\sumr(X_{i}^{*{T}}\beta-Y_{i}^{*})X_{i}^{*}=0,
    \end{equation*}
    to obtain the unweighted estimator.
  \end{algorithmic}
\end{algorithm}

We also repeated the simulation for $S=500$ times and
compared the empirical MSEs. In this section, we present the numerical results
under A-optimality only. The results under L-optimality are similar and we
present them in the supplementary material. Simulation results of
unconditionally empirical MSE are presented in Figure
\ref{fig:unconditionMSE_linear}. We see that the unweighted estimator is more
efficient in every case. Especially, when $X$ has a $T_{3}$ or $T_{1}$
distribution, the unweighted estimator performs significantly better than the
weighted estimator. As described in \cite{Zhang2020optimal}, the OSUMC estimator
outperforms other sampling methods more obviously when $X$ is heavy-tailed. We
notice that using the unweighted estimator, the advantage of OSUMC can be
significantly reinforced when the design is heavy-tailed, despite $X$ 
not meeting the regularity conditions we presented in Section \ref{sec:theory}.
\begin{figure}[H]
	\centering
	\begin{subfigure}{0.48\textwidth}
		\includegraphics[width=\textwidth]{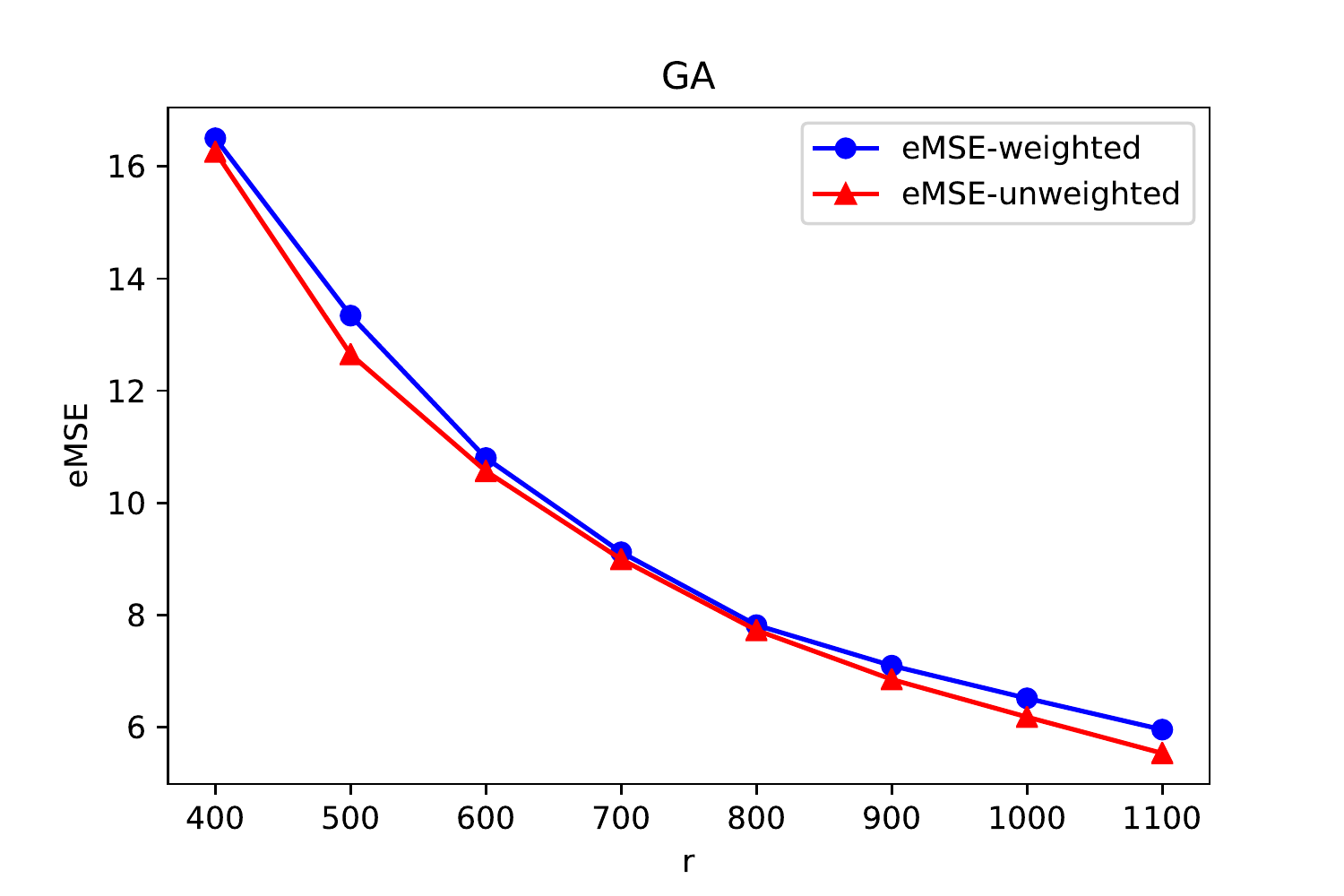}
		\caption{GA}
	\end{subfigure}
	\begin{subfigure}{0.48\textwidth}
		\includegraphics[width=\textwidth]{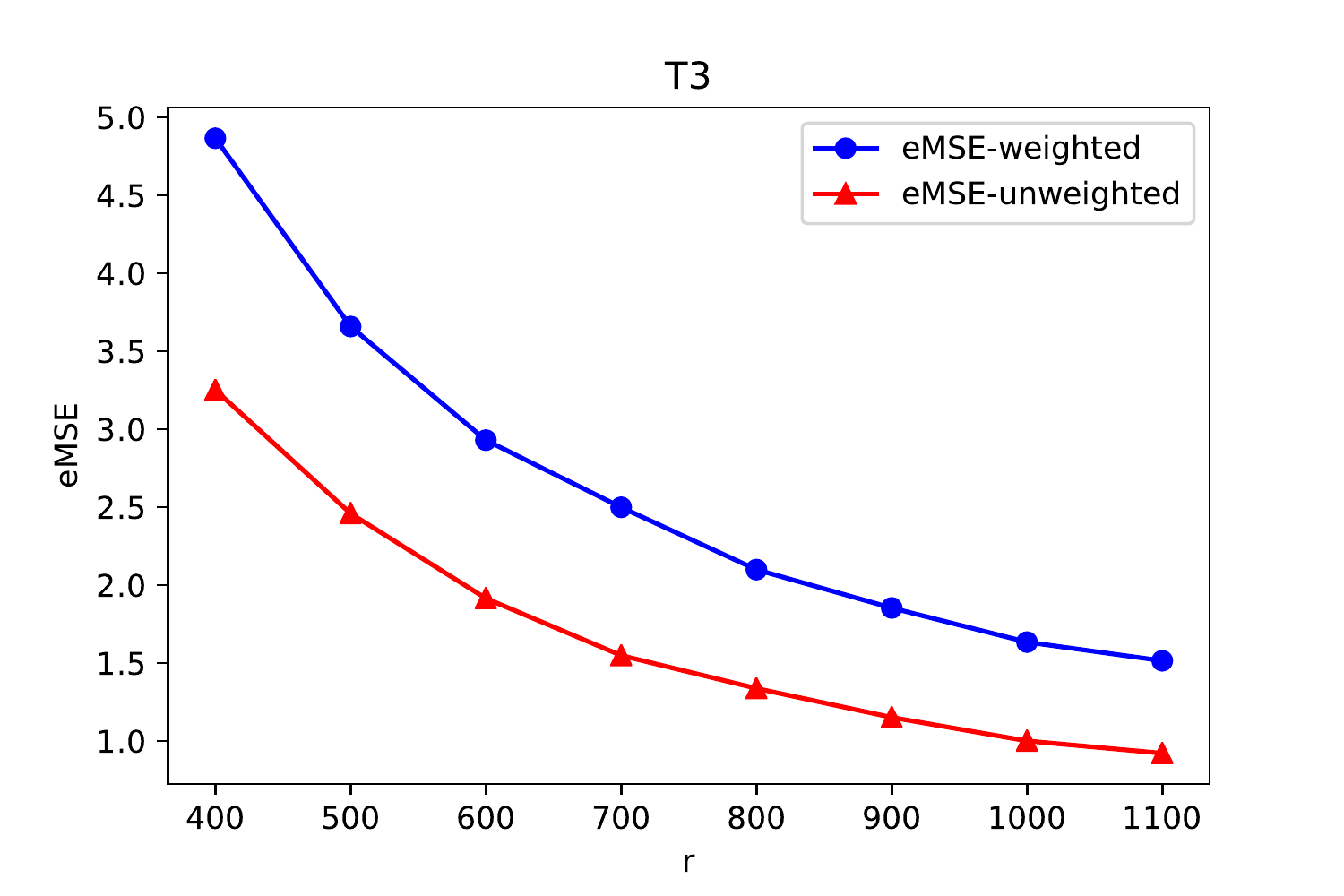}
		\caption{T3}
	\end{subfigure}
	\begin{subfigure}{0.48\textwidth}
		\includegraphics[width=\textwidth]{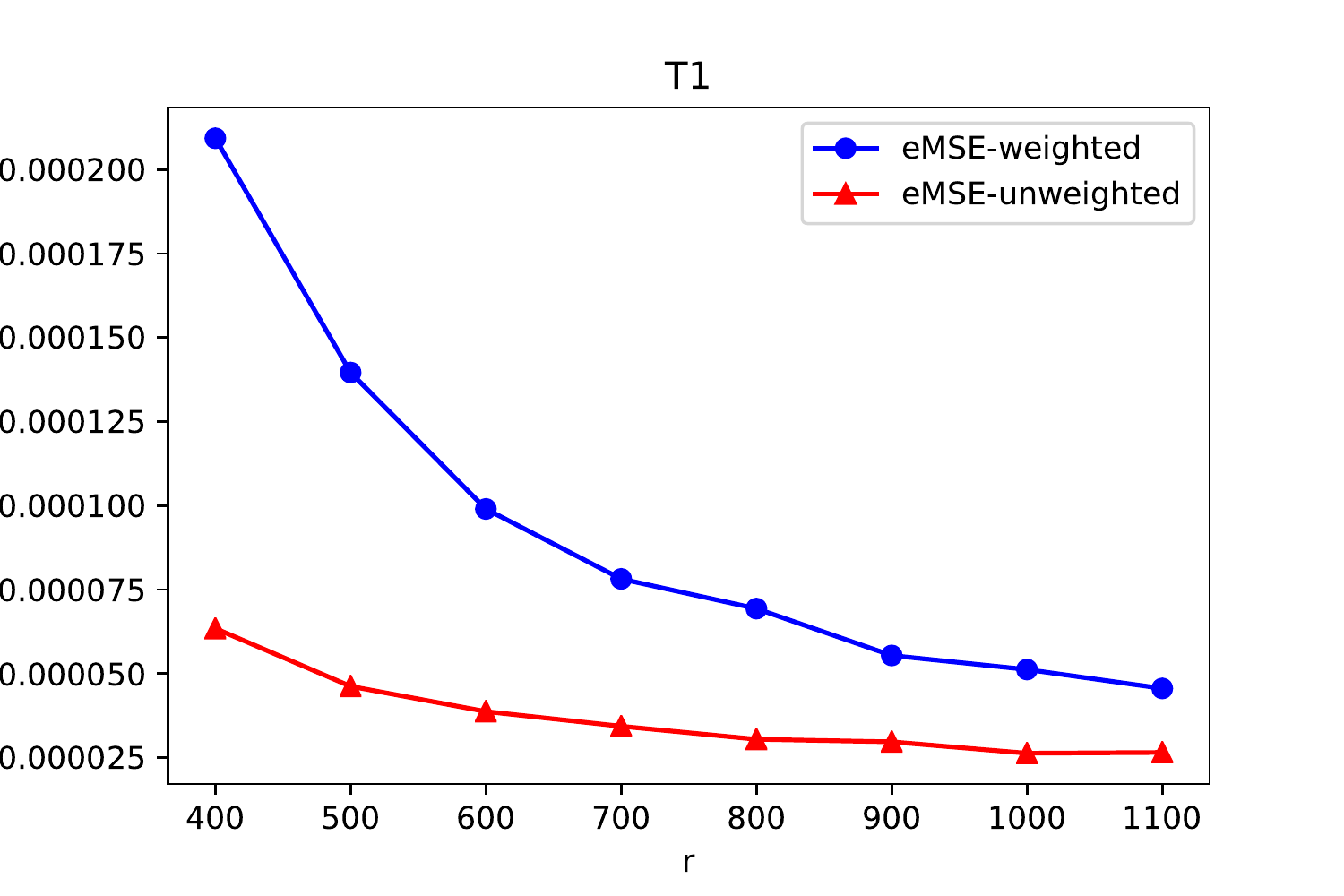}
		\caption{T1}
	\end{subfigure}
	\begin{subfigure}{0.48\textwidth}
		\includegraphics[width=\textwidth]{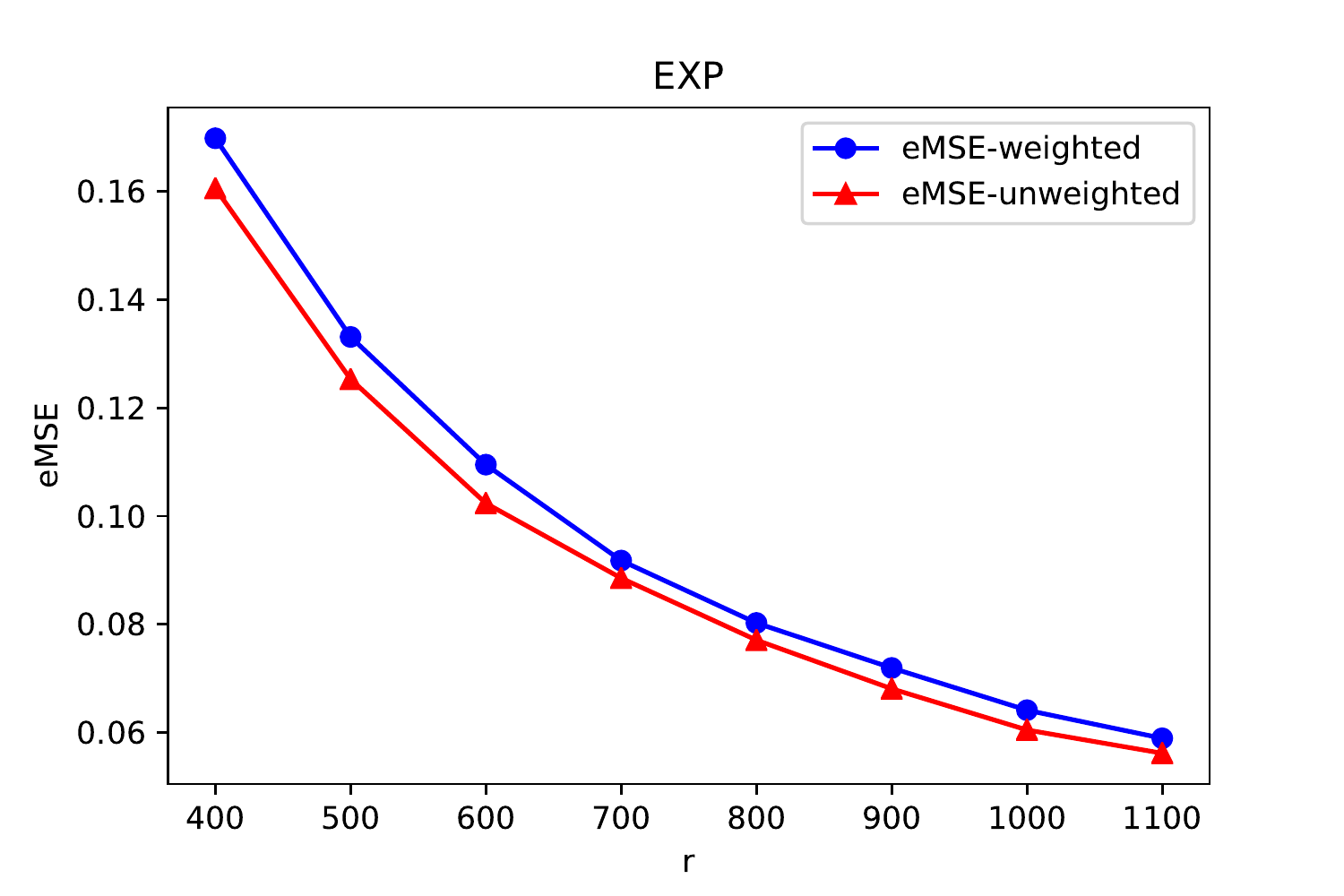}
		\caption{EXP}
	\end{subfigure}
	\caption{eMSE for different subsample sizes $r$ for linear regression under
    different settings.}
	\label{fig:unconditionMSE_linear}
\end{figure}

\subsubsection{Computational Complexity}

We present the computation times for the simulations based
on logistic regression in Table~\ref{tb:computime}. We used the same four
settings for the logistic regression in Section~\ref{sec:simulog}, and
repeated the experiments for $S=500$ times. We recorded the computing time for
the weighted and unweighted procedures and implemented both $\pi^{\mmse}$
and $\pi^{\mvc}$ using Python. Our computations were carried out on a laptop
running Windows 10 with an Intel i5 processor and 8GB memory, and we used the
package: \verb|sklearn.linear_model.LogisticRegression| for optimization. We present the results with
subsample size $r=1000$. The results for other subsample size are similar and
thus are omitted.
\begin{table}[H]
\caption{Computational time (seconds)}
\label{tb:computime}
\centering
\begin{tabular}{lcccccc}
\hline
 & \multicolumn{2}{c}{A-optimiality} && \multicolumn{2}{c}{L-optimality} &  \multirow{2}{*}{Full data}\\
  \cline{2-3}  \cline{5-6}
          & weighted & unweighted &  & weighted & unweighted                 \\ \hline
mzNormal  & 42.38    & 36.70      &  & 32.15    & 26.42      & 177.10 \\
nzNormal  & 39.98    & 36.38      &  & 30.31    & 26.89      & 165.15 \\
unNormal  & 41.20    & 37.47      &  & 32.61    & 28.69      & 256.45 \\
mixNormal & 40.97    & 36.14      &  & 31.92    & 27.45      & 162.44 \\ \hline
\end{tabular}
\end{table}

In Table~\ref{tb:computime}, both the weighted and unweighted subsample
estimators significantly reduce the computation time compared with the
MLE. The computation time of the unweighted estimator is not significantly
different from that of the weighted estimator. The probabilities based on
L-optimality reduce computation time more than the probabilities
based on the A-optimality, in agreement with the analysis in
Section~\ref{sec:theory}. Interestingly, we see that the unweighted estimator
is faster than the weighted estimator. This is because the target function of
the unweighted estimator is usually smoother than that of the weighted
estimator, and thus it takes fewer iterations for the algorithm to
converge. To confirm this, we present the average
numbers of iterations for optimizing the weighted and unweighted target functions
in Table~\ref{tb:itertimes}.
\begin{table}[H]
\caption{Average number of iterations of the optimization algorithm}
\label{tb:itertimes}
\centering
\begin{tabular}{lccccc}
\hline
 & \multicolumn{2}{c}{A-optimiality} && \multicolumn{2}{c}{L-optimality}\\
  \cline{2-3}  \cline{5-6}
          & weighted & unweighted &  & weighted & unweighted \\ \hline
mzNormal  & 18.53    & 10.68      &  & 18.51    & 10.77 \\
nzNormal  & 19.08    & 10.85      &  & 18.74    & 10.88 \\
unNormal  & 22.71    & 11.81      &  & 22.77    & 12.42 \\
mixNormal & 19.04    & 10.85      &  & 18.88    & 10.82 \\ \hline
\end{tabular}
\end{table}

\subsection{Experiments for real data}\label{sec:realdata}

We apply our more efficient unweighted estimator to real data and evaluate its
performance in this section.

\subsubsection{Superconductivty Data Set}

In this section, we apply our more efficient estimator to the superconductivty
data set used in \cite{Zhang2020optimal}. The data set is
available from the Machine Learning Repository at
\url{https://archive.ics.uci.edu/ml/datasets/Superconductivty+Data\#}. It
contains 21,263 different data points, and every data point has 81 features with
one continuous response. Each data point represents a superconductor. The
response is the superconductor's critical temperature and the features are
extracted from its chemical formula. For example, the 81st column is the number
of elements of the superconductor. We use standardized features as covariate
variables and adopted a multiple linear regression model to fit the critical
temperature from the chemical formula of the superconductor. Specially, the
linear regression model is
  \begin{equation*}
    Y=\beta^{(0)}+\beta^{(1)}Z_1+\beta^{(2)}Z_2+...+\beta^{(81)}Z_{81}+\epsilon,
  \end{equation*}
  where $Z_i$s represent the standardized features, $Y$ is the critical temperature, and
  $\epsilon$ is the normally distributed error.  To measure the performances of sampling
methods in parameter estimation, we use the empirical MSE of the estimator
\begin{equation}\label{eq:MSE_real}
	\text{eMSE}(\hbeta)=\frac{1}{S}\sum_{s=1}^{S}\|\hbeta^{(s)}-\hbeta_{\text{MLE}}\|,
\end{equation}
and the relative efficiency
\begin{equation}\label{eq:RE_real}
	\text{Relative Efficiency}=\frac{\text{eMSE}(\hbeta_{\rw})}{\text{eMSE}(\hbeta_{\ruw})},
\end{equation}
where $\hbeta^{(s)}$ represents the estimate in the $s$-th repetition. Here
we use the full data estimator $\hbeta_{\text{MLE}}$ instead of the ``true''
parameter $\beta_{0}$ to calculate eMSE because the true parameter is unknown
for real data sets. We repeated the experiment for $S=1000$ times, and
present the numerical results in Figure \ref{fig:Real_SuperConduct}. Our
unweighted estimator also outperforms the weighted estimator when applied to the
Superconductivity data set and $\pi_i^{\mmse}$ result in smaller eMSE than
  $\pi_i^{\mvc}$ for both the weighted and unweighted estimators.
\begin{figure}[H]
	\centering
	\begin{subfigure}{0.48\textwidth}
		\includegraphics[width=\textwidth]{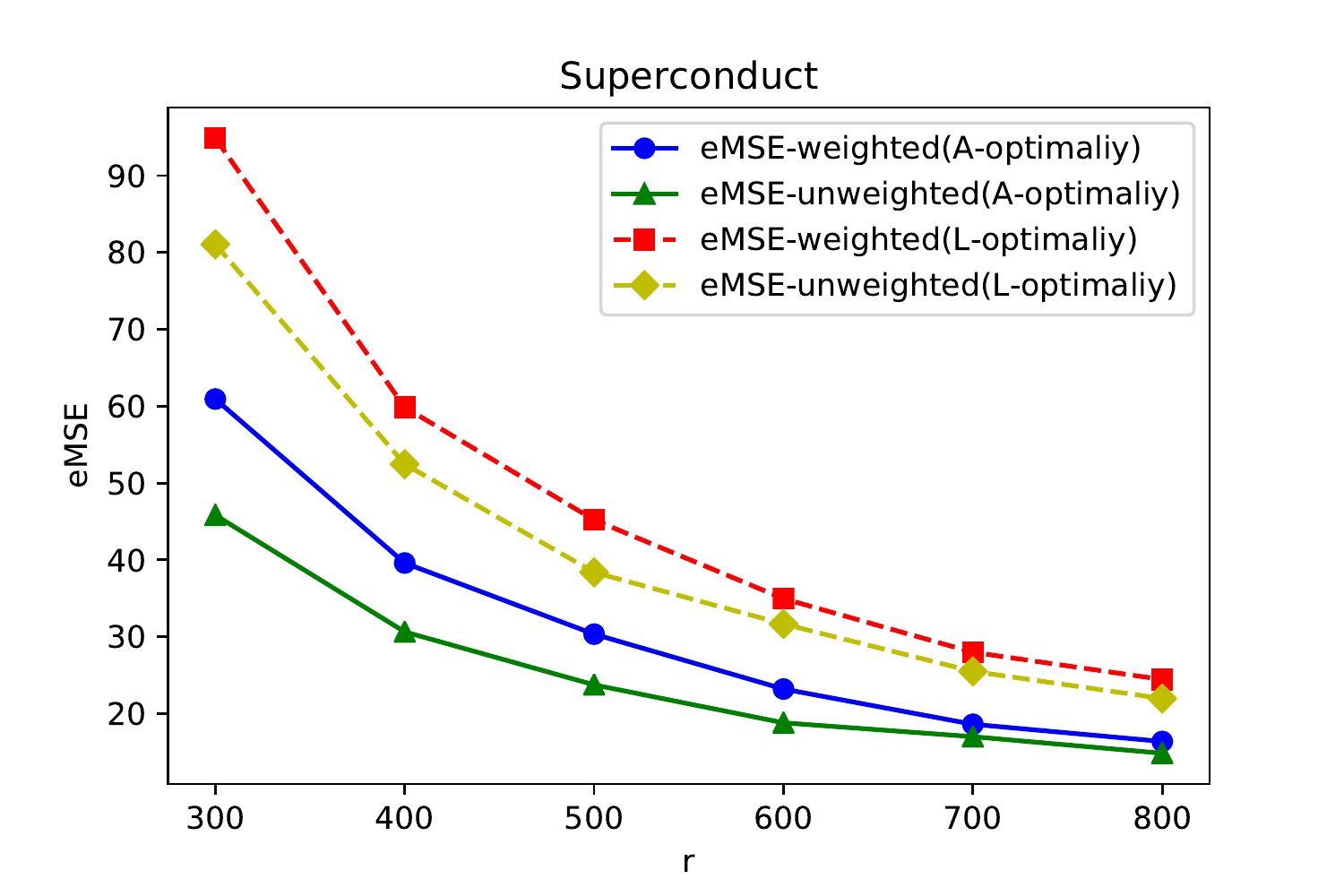}
		\caption{eMSE}
	\end{subfigure}
	\begin{subfigure}{0.48\textwidth}
		\includegraphics[width=\textwidth]{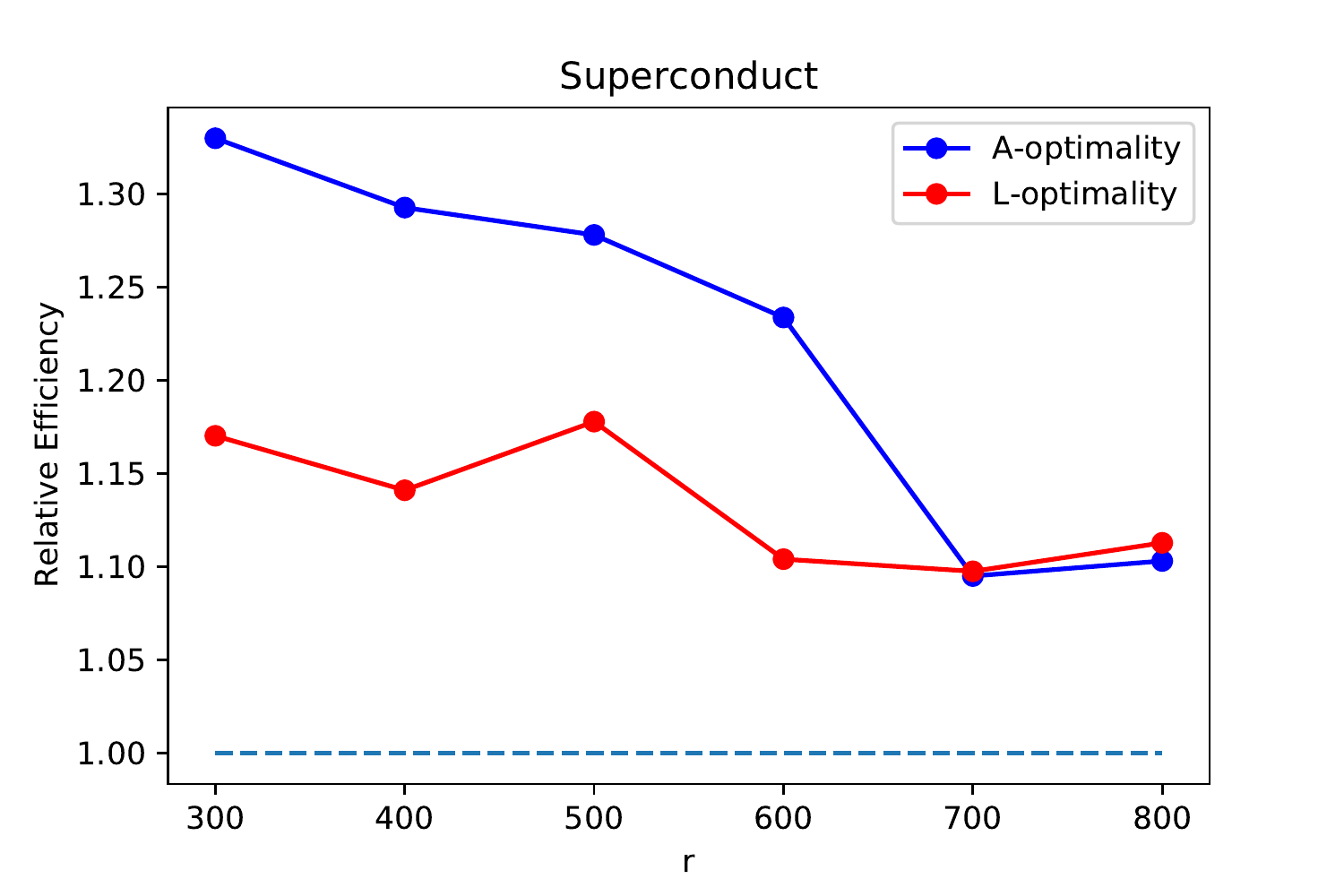}
		\caption{Relative Efficiency}
	\end{subfigure}
	\caption{eMSE and Relative Efficiency for Superconductivty data set with
    different subsample data sizes $r$.}
	\label{fig:Real_SuperConduct}
\end{figure}

\subsubsection{Supersymmetric Data Set}

In this section, the supersymmetric (SUSY) benchmark data set is used to evaluate
the performance of the unweighted estimator when applied to real data under
logistic model. The SUSY data set is available from the Machine Learning
Repository at \url{https://archive.ics.uci.edu/ml/datasets/SUSY}, and was also
used in \cite{WangZhuMa2017} and \cite{wang2019more}. The data are composed of
$n=5,000,000$ data points. Each data point represents a process and has one
binary response with 18 covariates. The response variable represents whether the
process produces new supersymmetric particles or the process is a background
process. The kinematic features of the process are used as covariates. We used a
logistic regression model to fit the data. Specifically, we model the probability
that a process produces new supersymmetric particles as
\begin{equation*}
P(Y=1|Z,\beta)=\frac{e^{\beta^{(0)}+\sum_{i=1}^{18}\beta^{(i)}Z_i}}{ 1 +
  e^{\beta^{(0)}+\sum_{i=1}^{18}\beta^{(i)}Z_i}},
\end{equation*}
where $Z_i$'s are the kinematic features of a process.
In order to compare
the efficiency of parameter estimation, we again use the regression coefficient
$\hbeta_{\text{MLE}}$ derived from the full data as the ``true
parameter''. The empirical MSE of the estimator defined in \eqref{eq:MSE_real}
and the relative efficiency defined in \eqref{eq:RE_real} are also
considered. We repeated the experiment for $S=1000$ times and drew a
pilot subsample of size $r_{\rp}=500$ in each repetition. Figure
\ref{fig:Real_SUSY} shows that the unweighted estimator is over 130\% more
efficient than the weighted one when applied to the SUSY data set when
  using $\pi_i^{\mmse}$, and over 110\% more efficient when using
  $\pi^{\mvc}$. Also, $\pi_i^{\mmse}$ performs better than
  $\pi_i^{\mvc}$ for the SUSY data set. %
\begin{figure}[H]
	\centering
	\begin{subfigure}{0.48\textwidth}
		\includegraphics[width=\textwidth]{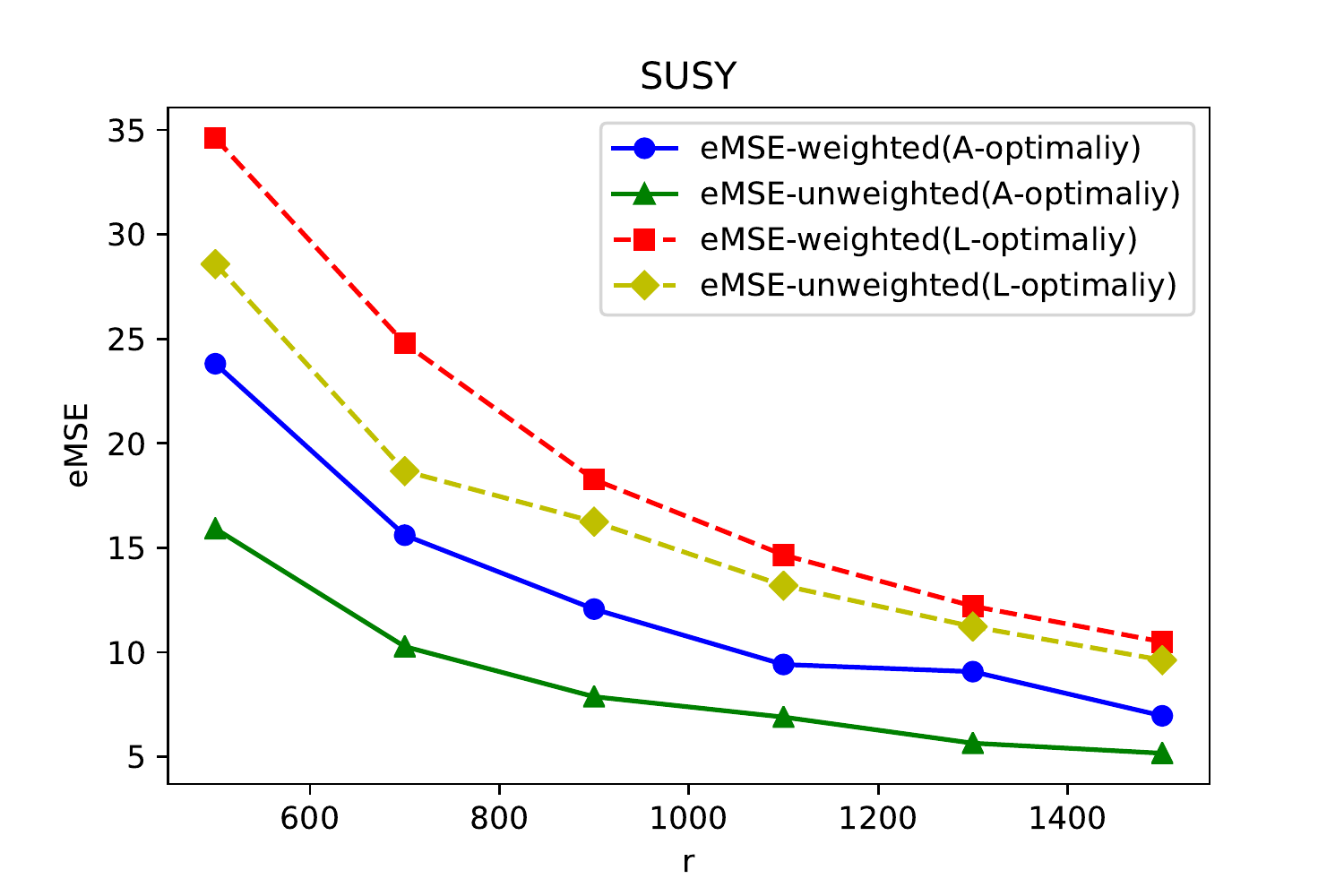}
		\caption{eMSE}
	\end{subfigure}
	\begin{subfigure}{0.48\textwidth}
		\includegraphics[width=\textwidth]{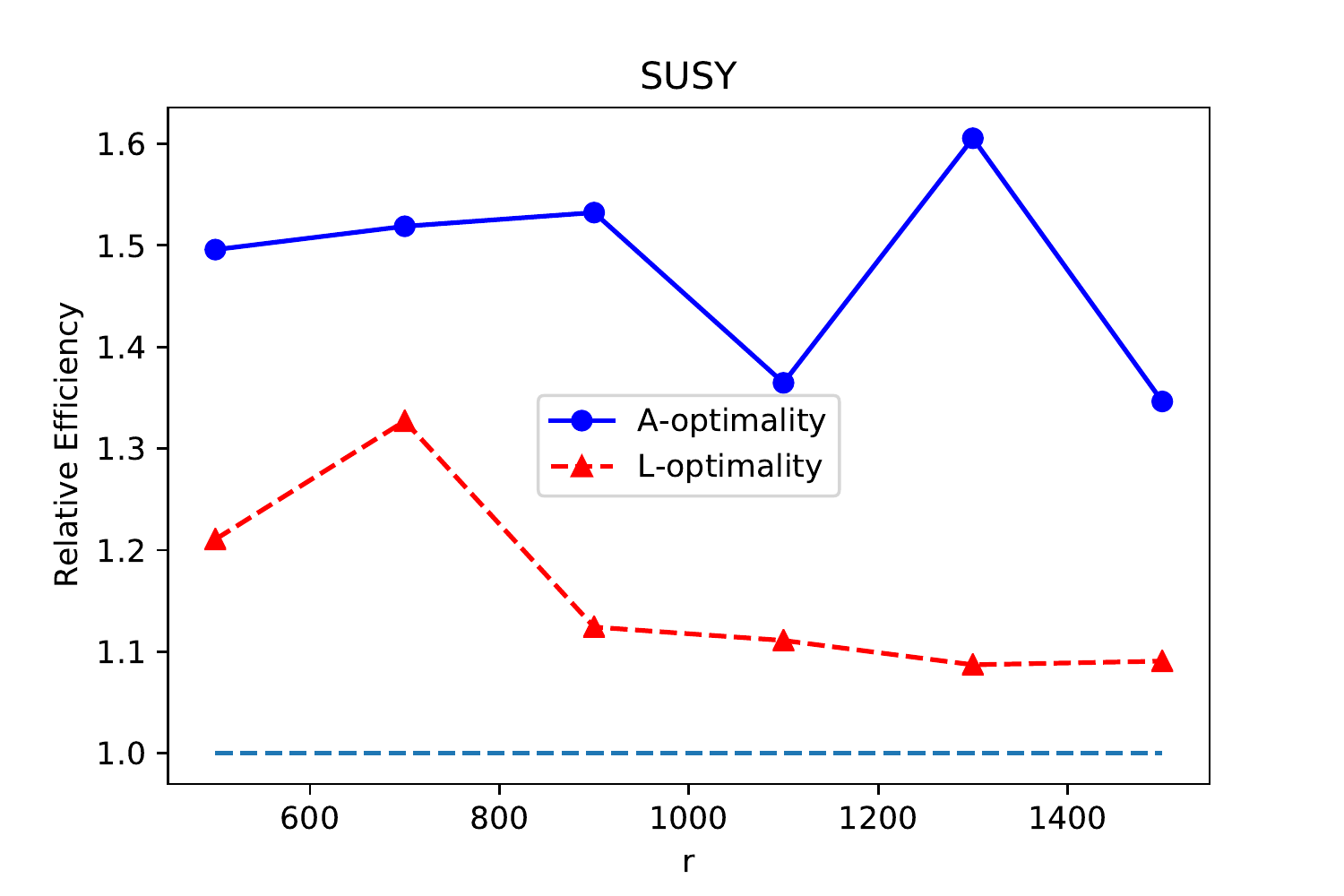}
		\caption{Relative Efficiency}
	\end{subfigure}
	\caption{eMSE and Relative Efficiency for SUSY data set with different
    subsample data sizes $r$ and a pilot sample size $r_{\rp}=500$.}
	\label{fig:Real_SUSY}
\end{figure}

\section{CONCLUSION} \label{sec:conclusion}

We proposed a novel unweighted estimator based on OSUMC subsample
for GLMs. It can be used to reduce computational burdens when responses are hard
to acquire. A two-step scheme was proposed and we showed the asymptotic
normality of the estimator unconditionally. We proved asymptotic results under a
martingale framework without conditioning on pilot estimates. Furthermore, we
showed that our new estimator is more efficient than the original OSUMC
estimator for parameter estimation under subsampling settings. Several numerical
experiments were implemented to demonstrate the superiouity of our
unweighted estimator over the original weighted esitmator.

Some extensions may be interesting for future research. Sampling with
replacement is used for both the weighted estimator and our new unweighted
estimator based on OSUMC. Poisson sampling is an alternative that reduces the
RAM usage for subsampling methods. Therefore, Poisson sampling is worth
developing under measurement constraints. Extending our subsampling
procedure to models beyond GLMs is also an interesting topic for future studies.

\section*{ACKNOWLEDGEMENTS}

The authors are very grateful to two anonymous referees and the editor for
constructive comments that helped to improve the paper. Jing Wang and Haiying Wang's research
is partially supported by US NSF grant CCF-2105571. Xiong's work is supported by the National
Science Foundation of China (Grant No. 12171462).

\clearpage

\appendix
\begin{center}
\LARGE Supplemental Material\\
\large for ``Unweighted estimation based on optimal sample under measurement
constraints''\\[20pt]
\end{center}

In this supplemental material, we present the proofs and technical details and some
addtional numerical experiments. Technical details are presented in Section
\ref{sec:proof_asymp} and Section \ref{sec:proof_effi}. Additional numerical
experiments are provided in Section \ref{sec:add_expriments}.

\setcounter{equation}{0}
\renewcommand{\theequation}{A.\arabic{equation}}
\setcounter{figure}{0}
\renewcommand{\thefigure}{A.\arabic{figure}}
\renewcommand{\thelemma}{A.\arabic{lemma}}

\section{Proofs and technical details}

In this appendix, we provide proofs in Section \ref{sec:theory}. Technical
details related to asymptotic results are presented in Section
\ref{sec:proof_asymp}; technical details related to estimation efficiency are
presented in Section \ref{sec:proof_effi}.

\subsection{Proofs of asymptotic results}\label{sec:proof_asymp}

We present proofs of the asymptotic results in this section. The proof of Lemma
\ref{lm:Convergence_of_Gamma} is presented in Section
\ref{sec:Convergence_of_Gamma}. The proof of Lemma \ref{lm:Normality_of_Psi} is
presented in Section \ref{sec:Normality}. The proof of Theorem
\ref{th:Normality} is presented in Section \ref{sec:Normality}. We first recall
some notations defined in the main paper:
\begin{equation*}
  \Psi_{\ruw}^{*}(\beta)=\oner\sumr\{b^{'}(X_{i}^{*{T}}\beta)-Y_{i}^{*}\}X_{i}^{*},
\end{equation*}
\begin{equation*}
  \hat{m}=\onen\sumn\sqrt{b^{''}(X_{i}^{T}\hbeta_{\rp})}\|L\hat{\Phi}_{\rp}^{-1}X_{i}\|,
\end{equation*}
\begin{equation*}
  \Phi=\Exp\left\{\onen\sumn b^{''}(X_{i}^{T}\beta_{0})X_{i}X_{i}^{T}\right\}=\Exp\{b^{''}(X^{T}\beta_{0})XX^{T}\},
\end{equation*}
and
\begin{equation*}
  \hat{\Phi}_{\rp}=\frac{1}{r_{\rp}}\sum_{i=1}^{r_{\rp}}b^{''}(X_{i}^{*_{\rp}{T}}\hbeta_{\rp})X_{i}^{*_{\rp}}X_{i}^{*_{\rp}{T}}.
\end{equation*}
In the following, we denote
\begin{equation*}
  \Psi_{\ruw}^{*\mathrm{m}}
  :=\hat{m}\Psi_{\ruw}^{*}(\beta)
  =\hat{m}\oner\sumr\{b^{'}(X_{i}^{*{T}}\beta)-Y_{i}^{*}\}X_{i}^{*}.
\end{equation*}
Here we use $\hat{m}$ to rescale the score function in order to simplify the
proof. We also denote
\begin{equation*}
  M_{n,i}:=\hat{m}\{b^{'}(X_{i}^{*{T}}\beta)-Y_{i}^{*}\}X_{i}^{*}-Q,
\end{equation*}
and
\begin{equation*}
  Q:=\onen\sumjn\sqrt{b^{''}(X_{j}^{T}\hbeta_{\rp})}\|L\hat{\Phi}_{\rp}^{-1}X_{j}\|\{b^{'}(X_{j}^{T}\beta)-Y_{j}\}X_{j}.
\end{equation*}
We then have
\begin{equation*}
  \Psi^{*\mathrm{m}}_{\ruw}(\beta)=\oner\sumr M_{n,i}+Q.
\end{equation*}
For filtration $\{\mathcal{F}_{n,i}\}_{i=1}^{r}$:
$\mathcal{F}_{n,0}=\sigma(X_{1}^{n},Y_{1}^{n},\hbeta_{\rp},\hat{\Phi}_{\rp})$;
$\mathcal{F}_{n,1}=\sigma(X_{1}^{n},Y_{1}^{n},\hbeta_{p},\hat{\Phi}_{\rp})\lor\sigma(*_{1})$;$\dots$;$\mathcal{F}_{n,i}=\sigma(X_{1}^{n},Y_{1}^{n},\hbeta_{\rp},\hat{\Phi}_{\rp})\lor\sigma(*_{1})\lor\dots\lor\sigma(*_{i})$,
where $\sigma(*_{i})$ is the $\sigma$-algebra generated by the $i$-th sampling
step, we have
\begin{align*}
  \Exp(M_{n,i}\vert\mathcal{F}_{n,i-1})&=\Exp\left[\hat{m}\{b^{'}(X_{i}^{*{T}}\beta)-Y_{i}^{*}\}X_{i}^{*}\vert\mathcal{F}_{n,i-1}\right]-Q\\
                                     &=\hat{m}\sumjn\pi_{j}^{\os}(\hbeta_{\rp},\hat{\Phi}_{\rp})\{b^{'}(X_{j}^{T}\beta)-Y_{j}\}X_{j}-Q.
\end{align*}
Since
\begin{equation*}
  \pi_{j}^{\os}(\hbeta_{\rp},\hat{\Phi}_{\rp})=\frac{\sqrt{b^{''}(X_{j}^{T}\hbeta_{\rp})}\|L\hat{\Phi}_{\rp}^{-1}X_{j}\|}{\sumn\sqrt{b^{''}(X_{i}^{T}\hbeta_{\rp})}\|L\hat{\Phi}_{\rp}^{-1}X_{i}\|}=\frac{\sqrt{b^{''}(X_{j}^{T}\hbeta_{\rp})}\|L\hat{\Phi}_{\rp}^{-1}X_{j}\|}{n\hat{m}},
\end{equation*}
we obtain that
\begin{equation*}
	\Exp(M_{n,i}\vert\mathcal{F}_{n,i-1})=0.
\end{equation*}
Therefore, $M_{n,i}$ is a martingale difference sequence, and we now present
some lemmas we will use to prove asymptotic results.
\begin{lemma}[Martingale Law of Large Numbers]\label{MartingaleLLN}
  If triangle array $(M_{n,i})_{i=1}^{r}$ is a martingale difference sequence,
  and uniformly integrable:
  \begin{equation*}
    \lim_{B\to\infty}\sup_{n,i}\Exp\left[\| M_{n,i}\| I\{\| M_{n,i}\|\geq B\}\right]=0,
  \end{equation*}
  then
  \begin{equation*}
    \Exp\left\{\left|\oner\sumr M_{n,i}\right|\right\}\to0.
  \end{equation*}
\end{lemma}
\begin{remark}
	This is a direct corollary of Theorem 2(b) in \cite{andrews1988laws}, which is
  mentioned in Section 3 of \cite{andrews1988laws}. Specially, if
  $(M_{n,i})_{i=1}^{r}$ is an identically distributed martingale difference
  sequence for fixed $n$, which is the case for sampling with replacement, we
  only need that $(M_{n,i})_{i=1}^{r}$ is uniformly integrable for index $n$,
  since
	\begin{equation*}
		\sup_{n}\Exp\left[\| M_{n,1}\| I\{\| M_{n,1}\|\geq B\}\right]=\sup_{n,i}\Exp\left[\| M_{n,i}\| I\{\| M_{n,i}\|\geq B\}\right].
	\end{equation*}
\end{remark}

\begin{lemma}\label{lm:L1_in_MoreEffi}
  Let $\mathbf{v_{1}}$,...,$\mathbf{v_{n}}$ be i.i.d random vector with the same
  distribution of $\mathbf{v}$. Let $g_{1}$ be a bounded function that may
  depend on $n$ and other random variables, and $g_{2}$ be a fixed function that
  does not depend on $n$. If $g_{1}(\mathbf{v_{i}})=o_{p}(1)$ for each $i$ as
  $n\to\infty$, and $\Exp\vert g_{2}(\mathbf{v})\vert<\infty$, then
  \begin{equation*}
    \onen\sumn g_{1}(\mathbf{v_{i}})g_{2}(\mathbf{v_{i}})=o_{p}(1).
  \end{equation*}
\end{lemma}
This is Lemma 1 in \cite{wang2019more}.
\begin{lemma}\label{lm:eigenvalue}
  Let $M$, $N$ be self-adjoint $k\times k$ matrices, and $m_{j}$, $n_{j}$ be
  their eigenvalues arranged in increasing order. Then
  \begin{equation*}
    \sum(n_{j}-m_{j})^{2}\leq\| N-M\|_{F}^{2}.
  \end{equation*}
\end{lemma}
This lemma is usually referred to as the Hoffman-Wielandt inequality, see
Theorem 18 in Chapter 10 of \cite{laxlinear}. Before we prove the main results,
we first prove Lemma \ref{lm:EigenvalueBound}.
\begin{lemma}\label{lm:EigenvalueBound}
  There exists an $r_{*}$ such that when $r_{\rp}\geq r_{*}$ ,
  $\lambda_{\max}(\hat{\Phi}_{\rp}^{-1})<C$ a.s., where $C$ is a finite
  constant. In addition, we have
  $\hbeta_{\rp}\xrightarrow{p}\beta_{0}$ and
  $\hat{\Phi}_{\rp}\xrightarrow{p}\Phi$, as $r_{\rp}\to\infty$.
\end{lemma}
\begin{proof}
  We have
  \begin{equation*}
    \hat{\Phi}_{\rp}=\frac{1}{r_{\rp}}\sum_{i=1}^{r_{\rp}}b^{''}(X_{i}^{*_{\rp}}\hbeta_{\rp})X_{i}^{*_{\rp}}X_{i}^{*_{\rp}{T}}\geq \frac{1}{r_{\rp}}\sum_{i=1}^{r_{\rp}}g(X_{i}^{*_{\rp}})X_{i}^{*_{\rp}}X_{i}^{*_{\rp}{T}},
  \end{equation*}
  and therefore
  \begin{equation*}
    \hat{\Phi}_{\rp}^{-1}\leq \left\{\frac{1}{r_{\rp}}\sum_{i=1}^{r_{\rp}}g(X_{i}^{*_{\rp}})X_{i}^{*_{\rp}}X_{i}^{*_{\rp}{T}}\right\}^{-1}.
  \end{equation*}
  Then, we have
  \begin{align*}
    &\lambda_{\max}\left(\hat{\Phi}_{\rp}^{-1}\right)\leq\lambda_{\max}\left[\left\{\frac{1}{r_{\rp}}\sum_{i=1}^{r_{\rp}}g(X_{i}^{*_{\rp}})X_{i}^{*_{\rp}}X_{i}^{*_{\rp}{T}}\right\}^{-1}\right]\\
    &\leq\lambda_{\max}\left[\Exp\{g(X)XX^{T}\}^{-1}\right]+\left|\lambda_{\max}\left[\left\{\frac{1}{r_{\rp}}\sum_{i=1}^{r_{\rp}}g(X_{i}^{*_{\rp}})X_{i}^{*_{\rp}}X_{i}^{^{*_{\rp}{T}}}\right\}^{-1}\right]-\lambda_{\max}\left[\Exp\{g(X)XX^{T}\}^{-1}\right]\right|\\
    &\leq\lambda_{\max}\left[\Exp\{g(X)XX^{T}\}^{-1}\right]+\left\|\left\{\frac{1}{r_{\rp}}\sum_{i=1}^{r_{\rp}}g(X_{i}^{*_{\rp}})X_{i}^{*_{\rp}}X_{i}^{*_{0}^{T}}\right\}^{-1}-\Exp\{g(X)XX^{T}\}^{-1}\right\|_{F},
  \end{align*}
  the last inequality is due to Lemma \ref{lm:eigenvalue}. Applying the Strong
  Law of Large Numbers, we obtain
  \begin{equation*}
    \left\|\left\{\frac{1}{r_{\rp}}\sum_{i=1}^{r_{\rp}}g(X_{i}^{*_{\rp}})X_{i}^{*_{\rp}}X_{i}^{*_{\rp}{T}}\right\}^{-1}-\Exp\{g(X)XX^{T}\}^{-1}\right\|_{F}=o(1),
  \end{equation*}
  and therefore we can find a constant $C$ and a $r_{*}$ such that
  \begin{equation*}
    \lambda_{\max}\left(\hat{\Phi}_{\rp}^{-1}\right)<C,
  \end{equation*}
  when
  $r>r_{*}$. %

  Next, we prove the second part of the lemma. Since we use simple random
  sampling in Algorithm \ref{alg:algo}, the asymptotic property of
  $\hbeta_{\rp}$ is the same as i.i.d data. Thus, the consistency of
  $\hbeta_{\rp}$ is easy to obtain, see
  \cite{mccullagh1989generalized}. Now, we know that
  $b^{''}(X_{i}^{*_{\rp}{T}}\hbeta_{\rp})-b^{''}(X_{i}^{*_{\rp}{T}}\beta_{0})$
  is $o_{p}(1)$ and bounded and
  \begin{equation*}
    \left\|\Exp\left(X_{i}^{*_{\rp}}X_{i}^{*_{\rp}{T}}\right)\right\|_{F}\leq\Exp(\| X_{i}^{*_{\rp}}\|^{2})<\infty.
  \end{equation*}
  By Lemma \ref{lm:L1_in_MoreEffi} and the Law of Large Numbers, we have
  \begin{equation*}
    \hat{\Phi}_{\rp}\xrightarrow{p}\Phi.
  \end{equation*}
	This complete the proof.
\end{proof}

Therefore, in the rest of the paper, we always assume that
\begin{equation*}
  \lambda_{\max}(\hat{\Phi}_{\rp}^{-1})\leq C.
\end{equation*}
Now, we show the proofs of asymptotic results presented in Section
\ref{sec:theory}.

\subsubsection{Proof of Lemma \ref{lm:Convergence_of_Gamma}}
\label{sec:Convergence_of_Gamma}

We prove Lemma \ref{lm:Convergence_of_Gamma} in this section, and we first
introduce some notations we use in this section. We have already defined
\begin{equation*}
  \Gamma(\beta)=\Exp\left\{\sqrt{b^{''}(X^{T}\beta_{0})}\|L\Phi^{-1}X\| b^{''}(X^{T}\beta) XX^{T}\right\}.
\end{equation*}
Denote
\begin{equation*}
  \tilde{M}_{n,i}=\hat{m}b^{''}(X_{i}^{*{T}}\beta)X_{i}^{*}X_{i}^{*{T}}-\tilde{Q},
\end{equation*}
\begin{equation*}
  \tilde{Q}=\onen\sumjn\sqrt{b^{''}(X_{j}^{T}\hbeta_{\rp})}\|L\hat{\Phi}_{\rp}^{-1}X_{j}\| b^{''}(X_{j}^{T}\beta)X_{j}X_{j}^{T}.
\end{equation*}
Then, we have
\begin{equation*}
  \Dot{\Psi}_{\ruw}^{*\mathrm{m}}(\beta)-\Gamma(\beta)=\oner\sumr\tilde{M}_{n,i}+\tilde{Q}-\Gamma(\beta).
\end{equation*}
For filtration $\{\mathcal{F}_{n,i}\}_{i=1}^{r}$:
$\mathcal{F}_{n,0}=\sigma(X_{1}^{n},Y_{1}^{n},\hbeta_{\rp},\hat{\Phi}_{\rp})$;
$\mathcal{F}_{n,1}=\sigma(X_{1}^{n},Y_{1}^{n},\hbeta_{\rp},\hat{\Phi}_{\rp})\lor\sigma(*_{1})$;
$\dots$;$\mathcal{F}_{n,i}=\sigma(X_{1}^{n},Y_{1}^{n},\hbeta_{\rp},\hat{\Phi}_{\rp})\lor\sigma(*_{1})\lor\dots\lor\sigma(*_{i})$,
we have
\begin{equation*}
  \Exp(\tilde{M}_{n,i}\vert\mathcal{F}_{n,i-1})=0.
\end{equation*}
Therefore, $\tilde{M}_{n,i}$ is a martingale difference sequence. If we denote
\begin{equation*}
  \hat{w}_{i}=\sqrt{b^{''}(X_{i}^{T}\hbeta_{\rp})}\|L\hat{\Phi}_{\rp}^{-1}X_{i}\|>0,
\end{equation*}
we will then know that
\begin{equation*}
  \hat{m}=\onen\sumjn\hat{w}_{j},
\end{equation*}
and
\begin{equation*}
  |\hat{w}_{i}|\lesssim\|L\hat{\Phi}_{\rp}^{-1}X_{i}\|\lesssim\| X_{i}\|,
\end{equation*}
due to Assumption \ref{A1} and Lemma \ref{lm:EigenvalueBound}, where the notion
``$\lesssim$'' means that there exist a constant $K$, such that
$\hat{w}_{i}\leq K\| X_{i}\|$. Now, we present the proof of Lemma
\ref{lm:Convergence_of_Gamma}.
\begin{proof}[\textbf{Proof of Lemma \ref{lm:Convergence_of_Gamma}}]
  To show that
  $\Dot{\Psi}_{\ruw}^{*\mathrm{m}}(\beta_{n})\xrightarrow{p}\Gamma$ for
  every sequence $\beta_{n}\xrightarrow{p}\beta_{0}$, we first show that for
  every $\beta\in\mathbb{B}$,
  \begin{equation*}
    \Dot{\Psi}_{\ruw}^{*\mathrm{m}}(\beta)-\Gamma(\beta)=\hat{m}\oner\sumr
          b^{''}(X_{i}^{*{T}}\beta)X_{i}^{*}X_{i}^{*{T}}-\Gamma(\beta)=o_{p}(1),
  \end{equation*}
  where
  \begin{equation*}
    \Gamma(\beta):=\Exp\left\{\sqrt{b^{''}(X^{T}\beta_{0})}
      \|L\Phi^{-1}X\|b^{''}(X^{T}\beta) XX^{T}\right\}.
  \end{equation*}
  We already know that
  \begin{equation*}
    \Dot{\Psi}_{\ruw}^{*\mathrm{m}}(\beta)-\Gamma(\beta)=\oner\sumr\tilde{M}_{n,i}
    +\tilde{Q}-\Gamma(\beta),
  \end{equation*}
   and $(\tilde{M}_{n,i})_{i}^{r}$ is an identically distributed martingale
  difference sequence. Therefore, for each entry of
  $\Dot{\Psi}_{\ruw}^{*\mathrm{m}}(\beta)$, assuming $0<\delta<1$, we
  have
  \begin{align*}
    &\left\{\Exp\left(\left| \tilde{M}_{n,1}^{(lk)}\right|^{1+\delta}\right)\right\}^{\frac{1}{1+\delta}}\\
    &\leq\left[\Exp\left\{\left|\hat{m}b^{''}(X_{1}^{*{T}}\beta)X_{1}^{*{(l)}}X_{1}^{*{(k)}}\right|^{1+\delta}\right\}\right]^{\frac{1}{1+\delta}}+\left[\Exp\left\{\left|\onen\sumjn\hat{w}_{j} b^{''}(X_{j}^{T}\beta)X_{j}^{(l)}X_{j}^{(k)}\right|^{1+\delta}\right\}\right]^{\frac{1}{1+\delta}}\\
    &:=I_{1}^{\frac{1}{1+\delta}}+I_{2}^{\frac{1}{1+\delta}},
  \end{align*}
  due to the Minkowski inequality. Then, for $I_{1}$, we have
  \begin{align*}
    I_{1}&\lesssim\Exp\left(\hat{m}^{1+\delta}\left| X_{1}^{*{(l)}}X_{1}^{*{(k)}}\right|^{1+\delta}\right)=\Exp\left\{\Exp\left(\hat{m}^{1+\delta}\left| X_{1}^{*{(l)}}X_{1}^{*{(k)}}\right|^{1+\delta}\bigg{\vert}\mathcal{F}_{n,0}\right)\right\}\\
         &=\Exp\left(\hat{m}^{\delta}\onen\sumjn\hat{w}_{j}\left| X_{j}^{(l)}X_{j}^{(k)}\right|^{1+\delta}\right)\lesssim\Exp\left\{\left(\onen\sum_{i=1}^{n}\| X_{i}\|\right)^{\delta}\onen\sumjn\| X_{j}\|\left| X_{j}^{(l)}X_{j}^{(k)}\right|^{1+\delta}\right\}\\
         &\leq\left\{\Exp\left(\onen\sum_{i=1}^{n}\| X_{i}\|\right)\right\}^{\delta}\left[\Exp\left\{\left(\onen\sumjn\| X_{j}\|\left| X_{j}^{(l)}X_{j}^{(k)}\right|^{1+\delta}\right)^{\frac{1}{1-\delta}}\right\}\right]^{1-\delta}.
  \end{align*}
  The last inequality is due to the H\"{o}lder inequality when $p=1/\delta$ and
  $q=1/(1-\delta)$. Since $1/(1-\delta)>1$, we then use the Minkowski inequality
  and obtain
  \begin{align*}
    I_{1}&\lesssim\left\{\Exp\left(\onen\sum_{i=1}^{n}\| X_{i}\|\right)\right\}^{\delta}\left[\onen\sumjn\Exp\left\{\left(\| X_{j}\|\left| X_{j}^{(l)}X_{j}^{(k)}\right|^{1+\delta}\right)^{\frac{1}{1-\delta}}\right\}\right]^{1-\delta}\\
         &\leq\left\{\Exp\left(\| X\|\right)\right\}^{\delta}\left\{\Exp\left(\| X\|^{\frac{1}{1-\delta}}\left| X^{(l)}X^{(k)}\right|^{\frac{1+\delta}{1-\delta}}\right)\right\}^{1-\delta}.
  \end{align*}
  For $I_{2}$, applying the Minkowski inequality, we obtain
  \begin{align*}
    I_{2}&\leq\onen\sumjn\Exp\left\{\left|\hat{w}_{j} b^{''}(X_{j}^{T}\beta)X_{j}^{(l)}X_{j}^{(k)}\right|^{1+\delta}\right\}\\
         &\lesssim\Exp\left(\| X\|^{1+\delta}| X^{(l)}X^{(k)}|^{1+\delta}\right).
  \end{align*}
  To see
  $\Exp\left(\| X\|^{\frac{1}{1-\delta}}\left|
      X^{(l)}X^{(k)}\right|^{\frac{1+\delta}{1-\delta}}\right)<\infty$, we
  only need to add all the entries:
  \begin{align*}
    \sum_{l,k}\Exp\left(\| X\|^{\frac{1}{1-\delta}}\left|
    X^{(l)}X^{(k)}\right|^{\frac{1+\delta}{1-\delta}}\right)=\sum_{l,k}\Exp\left(\|
    X\|^{\frac{1}{1-\delta}}\left|
    X^{(l)}\right|^{\frac{1+\delta}{1-\delta}}\left|
    X^{(k)}\right|^{\frac{1+\delta}{1-\delta}}\right)\lesssim\Exp\left(\| X\|^{\frac{3+2\delta}{1-\delta}}\right),
  \end{align*}
  Similarly, we also have
  \begin{equation*}
    \sum_{l,k}\Exp\left(\| X\|^{1+\delta}|
      X^{(l)}X^{(k)}|^{1+\delta}\right)\lesssim \Exp\left(\| X\|^{3+3\delta}\right).
  \end{equation*}
  If we let $\delta=1/6$, we then have
  \begin{equation*}
    \Exp\left(\| X\|^{\frac{3+2\delta}{1-\delta}}\right)=\Exp(\| X\|^{4})<\infty,
  \end{equation*}
  and
  \begin{equation*}
    \Exp\left(\| X\|^{3+3\delta}\right)=\Exp\left(\| X\|^{3.5}\right)<\infty.
  \end{equation*}
  Therefore,
  $\sup_{n}\Exp(| \tilde{M}_{n,1}^{(lk)}|^{1+\delta})<\infty$
  and we know that $(\tilde{M}_{n,i})_{i}^{r}$ is $L^{1}$ uniformly
  integrable. Now applying Lemma \ref{MartingaleLLN}, we can show that
  \begin{equation*}
    \oner\sumr\tilde{M}_{n,i}=o_{p}(1).
  \end{equation*}
  Next, we prove
  \begin{align*}
    &\tilde{Q}-\Gamma(\beta)\\
    &=\onen\sumjn\hat{w}_{j}b^{''}(X_{j}^{T}\beta)X_{j}X{j}^{T}-\Exp\left\{\sqrt{b^{''}(X^{T}\beta_{0})}\|L\Phi^{-1}X\| b^{''}(X^{T}\beta)XX^{T}\right\}\\
    &=o_{p}(1).
  \end{align*}
  For the consistency of $\hbeta_{\rp}$ and
  $\hat{\Phi}_{\rp}^{-1}$, it is easy to know that
  \begin{equation*}
    \sqrt{b^{''}(X_{j}^{T}\hbeta_{\rp})}\|L\hat{\Phi}_{\rp}^{-1}X_{j}\|-\sqrt{b^{''}(X_{j}^{T}\beta_{0})}\|L\Phi^{-1}X_{j}\|=o_{p}(1).
  \end{equation*}
  Considering each entry of $\tilde{Q}$ and $\Gamma(\beta)$, we know that
  \begin{align*}
    &\onen\sumjn\hat{w}_{j}b^{''}(X_{j}^{T}\beta)X_{j}^{(l)}X_{j}^{(k)}-\onen\sumjn\sqrt{b^{''}(X_{j}^{T}\beta_{0})}\|L\Phi^{-1}X_{j}\| b^{''}(X_{j}^{T}\beta)X_{j}^{(l)}X_{j}^{(k)}\\
    &=\onen\sumjn\left\{\sqrt{b^{''}(X_{j}^{T}\hbeta_{\rp})}\|L\hat{\Phi}_{\rp}^{-1}X_{j}\|-\sqrt{b^{''}(X_{j}^{T}\beta_{0})}\|L\Phi^{-1}X_{j}\|\right\}b^{''}(X_{j}^{T}\beta)X_{j}^{(l)}X_{j}^{(k)}.
  \end{align*}
  Then, since we include intercept in the model we know
  $\| X_{j}\|\geq 1$, and thus we have
  \begin{align*}
    &\onen\sumjn\left\{\sqrt{b^{''}(X_{j}^{T}\hbeta_{\rp})}\|L\hat{\Phi}_{\rp}^{-1}X_{j}\|-\sqrt{b^{''}(X_{j}^{T}\beta_{0})}\|L\Phi^{-1}X_{j}\|\right\}b^{''}(X_{j}^{T}\beta)X_{j}^{(l)}X_{j}^{(k)}\\
    &=\onen\sumjn\left\{\sqrt{b^{''}(X_{j}^{T}\hbeta_{\rp})}\frac{\|L\hat{\Phi}_{\rp}^{-1}X_{j}\|}{\|
     X_{j}\|}-\sqrt{b^{''}(X_{j}^{T}\beta_{0})}\frac{\|L\Phi^{-1}X_{j}\|}{\| X_{j}\|}\right\}\| X_{j}\| b^{''}(X_{j}^{T}\beta)X_{j}^{(l)}X_{j}^{(k)}.
  \end{align*}
  We know that
  \begin{equation*}
    \frac{\| L\Phi^{-1}X_{j}\|}{\| X_{j}\|}\leq C,\quad and\quad \frac{\| L\hat{\Phi}_{\rp}^{-1}X_{j}\|}{\| X_{j}\|}\leq C,
  \end{equation*}
  which means
  \begin{equation*}
    \left\{\sqrt{b^{''}(X_{j}^{T}\hbeta_{\rp})}
      \frac{\|L\hat{\Phi}_{\rp}^{-1}X_{j}\|}{\| X_{j}\|}
      -\sqrt{b^{''}(X_{j}^{T}\beta_{0})}
      \frac{\|L\Phi^{-1}X_{j}\|}{\| X_{j}\|}\right\}
  \end{equation*}
  is bounded and also is $o_{p}(1)$. We now show that
  $\| X_{j}\| b^{''}(X_{j}^{T}\beta)X_{j}^{(l)}X_{j}^{(k)}$ is
  integrable, since
  \begin{align*}
    \Exp\left\{\| X_{j}\| b^{''}(X_{j}^{T}\beta)| X^{(l)}|\cdot|
    X^{(k)}|\right\}&\lesssim\Exp\left(\| X_{j}\|\cdot | X^{(l)}|\cdot| X^{(k)}|\right)\\
                   &\leq\sum_{l,k}\Exp\left(\| X_{j}\|\cdot| X^{(l)}|\cdot| X^{(k)}|\right)\\
                   &\lesssim\Exp\left(\| X_{j}\|^{3}\right)<\infty.
  \end{align*}
  Then, applying Lemma \ref{lm:L1_in_MoreEffi} and the Law of Large Numbers, we
  have
  \begin{align*}
    &\tilde{Q}^{(lk)}-\Gamma(\beta)^{(lk)}\\
    &=\onen\sumjn\sqrt{b^{''}(X_{j}^{T}\beta_{0})}\|L\Phi^{-1}X_{j}\| b^{''}(X_{j}^{T}\beta)X_{j}^{(l)}X_{j}^{(k)}-\Gamma(\beta)^{(lk)}+o_{p}(1)\\
    &=o_{p}(1),
  \end{align*}
  which means $\tilde{Q}-\Gamma(\beta)=o_{p}(1)$. Therefore,
  \begin{equation*}
    \Dot{\Psi}_{\ruw}^{*\mathrm{m}}(\beta)-\Gamma(\beta)=o_{p}(1).
  \end{equation*}
  We have already proved that
  $\Dot{\Psi}_{\ruw}^{*\mathrm{m}}(\beta)\xrightarrow{p}\Gamma(\beta)$,
  for every $\beta\in\mathbb{B}$, and we next prove
  $\Dot{\Psi}_{\ruw}^{*\mathrm{m}}(\beta_{n})\xrightarrow{p}\Gamma$, for
  every sequence $\beta_{n}\xrightarrow{p}\beta_{0}$.  For each entry of
  $\Dot{\Psi}_{n}^{*}(\beta)$ and $\Gamma(\beta)$, taking
  $\beta_{1}, \beta_{2}\in\mathbb{B}$, we use the mean value theorem and obtain
  \begin{align*}
    &\left| \left\{\Dot{\Psi}_{\ruw}^{*\mathrm{m}}(\beta_{1})^{(lk)}-\Gamma(\beta_{1})^{(lk)}\right\}-\left\{\Dot{\Psi}_{\ruw}^{*\mathrm{m}}(\beta_{2})^{(lk)}-\Gamma(\beta_{2})^{(lk)}\right\}\right|\\
    &=\bigg{\|}\bigg{\{}\hat{m}\oner\sumr b^{'''}(X_{i}^{*{T}}\tilde{\beta_{1}})X_{i}^{*{(l)}}X_{i}^{*{(k)}}X_{i}^{*}\\
    &\quad -\Exp[\sqrt{b^{''}(X^{T}\beta_{0})}\|L\Phi^{-1}X\| b^{'''}(X^{T}\tilde{\beta}_{2})X^{(l)}X^{(k)}X]\bigg{\}}\cdot(\beta_{1}-\beta_{2})\bigg{\|}\\
    &\lesssim \left\{\hat{m}\oner\sumr h(X_{i}^{*})\left| X_{i}^{*{(l)}}X_{i}^{*{(k)}}\right|\| X_{i}^{*}\| +O(1)\right\}\|\beta_{1}-\beta_{2}\|\\
    &:=L_{n}\|\beta_{1}-\beta_{2}\|,
  \end{align*}
  the second inequality is due to Assumption \ref{A3}. Then, we need to show
  $L_{n}=O_{p}(1)$. Actually, we know that
  \begin{align*}\label{Newproof}
    \Exp\left\{\hat{m}\oner\sumr
    h(X_{i}^{*})\left|X_{i}^{*{(l)}}X_{i}^{*{(k)}}\right|\| X_{i}^{*}\|\right\}&=\Exp\left\{\onen\sumjn\hat{w}_{j}h(X_{j})\left| X_{j}^{(l)}X_{j}^{(k)}\right\rvert\| X_{j}\|\right\}\\
                                                                          &\lesssim\Exp\left\{\onen\sumjn
                                                                            h(X_{j})\left|X_{j}^{(l)}X_{j}^{(k)}\right\rvert\| X_{j}\|^{2}\right\}\\
                                                                          &=\Exp\left\{h(X)\left| X^{(l)}X^{(k)}\right|\| X\|^{2}\right\},
  \end{align*}
  and
  $\Exp\left\{h(X)\left| X^{(l)}X^{(k)}\right\rvert\|
    X\|^{2}\right\}\leq\sum_{l,k}\Exp\left\{h(X)\left|
      X^{(l)}X^{(k)}\right|\|
    X\|^{2}\right\}\lesssim\Exp\left\{h(X)\|
    X\|^{4}\right\}<\infty$, which is guaranteed because of Assumption
  \ref{A3}; therefore, $L_{n}$ is $O_{p}(1)$. Now we apply Theorem 21.10 in
  \cite{Davidson1994Stochastic} to conclude that
  $\Dot{\Psi}_{\ruw}^{*\mathrm{m}}(\beta)-\Gamma(\beta)$ is stochastic
  equicontinuous. Then consistency of
  $\Dot{\Psi}_{\ruw}^{*\mathrm{m}}(\beta)$ and stochastic equicontinuity
  implies
  \begin{equation*}
    \sup_{\beta\in\mathbb{B}}\|\Dot{\Psi}_{\ruw}^{*\mathrm{m}}(\beta)-\Gamma(\beta)\|\xrightarrow{p}0,
  \end{equation*}
  due to Theorem 21.9 in \cite{Davidson1994Stochastic}. Finally, applying
  Theorem 21.6 in \cite{Davidson1994Stochastic}, we conclude that for every
  sequence $\beta_{n}\xrightarrow{p}\beta_{0}$, we have
  \begin{equation*}
    \Dot{\Psi}_{\ruw}^{*\mathrm{m}}(\beta_{n})\xrightarrow{p}\Gamma(\beta_{0})=\Gamma.
  \end{equation*}
\end{proof}

\subsubsection{Proof of Lemma
  \ref{lm:Normality_of_Psi}}\label{sec:Normality_of_Psi}

In this section, our goal is to prove Lemma \ref{lm:Normality_of_Psi}.  We first
prove that when $r/n\to0$, under the condition
\begin{equation*}
  \Exp\left\{\left| b^{'}(X^{T}\beta_{0})-Y\right|^{2+\delta}\left\| X\right\|^{4+2\delta}\right\}<\infty,
\end{equation*}
we have
\begin{equation*}
  \sqrt{r}\Psi_{\ruw}^{*\mathrm{m}}(\beta_{0})\xrightarrow{d}N(0,m\Gamma).
\end{equation*}
We use the following martingale central limit theorem (CLT) in Hilbert space to
prove this weaker verision of asymptotic normality.
\begin{lemma}[Martingale Central Limit Theorem]\label{lm:MartingaleCLT}
  Let $H$ be separable Hilbert space, ${X_{nk}}$ be $H$-valued martingale
  difference sequence w.r.t. ${\mathcal{F}_{nk}}$, namely, $\{X_{nk}\}$ is
  adapted to $\{\mathcal{F}_{nk}\}$, $\Exp\| X_{nk}\|^{2}<\infty$,
  $\Exp(X_{nk}\vert\mathcal{F}_{nk-1})=0$, and $N(0,S)$ be Gaussian
  distribution, if the following conditions hold:
  \begin{enumerate}
  \item\label{c1}
    $\sum_{k=1}^{k(n)}\Exp\left(\|
      X_{nk}\|^{2}\vert\mathcal{F}_{n,k-1}\right)\xrightarrow{p}\text{tr}(S)$,
  \item\label{c2}
    $\sum_{k=1}^{k(n)}\Exp\left[\| X_{nk}\|^{2} I\{\|
      X_{nk}\|>\varepsilon\}\vert\mathcal{F}_{n,k-1}\right]\xrightarrow{p}0$,
    for every $\varepsilon>0$,
  \item\label{c3}
    $\sum_{k=1}^{k(n)}\Exp\left\{( X_{nk},e_{i})(
      X_{nk},e_{j})\vert\mathcal{F}_{n,k-1}\right\}\xrightarrow{p}(Se_{i},e_{j})$
    for some orthonormal basis ${e_{i}}$ in $H$ and $i,j\in\mathcal{N}$,
  \end{enumerate}
	then $S_{n}:=\sum_{k=1}^{k(n)} X_{nk}\xrightarrow{d}N(0,S)$.
\end{lemma}
See Theorem C in \cite{jakubowski1980limit}. We now present the
proof.%
\begin{proof}[\textbf{Proof for the case of $\rho=0$}]
  Note that
  \begin{equation*}
    \sqrt{r}\Psi_{\ruw}^{*\mathrm{m}}(\beta_{0})=\frac{1}{\sqrt{r}}\sumr
    M_{n,i}+\sqrt{r}Q.
  \end{equation*}
  First, we prove
  \begin{equation*}
    \sqrt{r}Q=o_{p}(1).
  \end{equation*}
  We know that
  \begin{align*}
    \sqrt{r}Q&=\sqrt{r}\onen\sumjn\sqrt{b^{''}(X_{j}^{T}\hbeta_{\rp})}\|L\hat{\Phi}_{\rp}^{-1}X_{j}\|\left\{b^{'}(X_{j}^{T}\beta_{0})-Y_{j}\right\}X_{j}\\
             &=\sqrt{\frac{r}{n}}\frac{1}{\sqrt{n}}\sumjn\sqrt{b^{''}(X_{j}^{T}\hbeta_{\rp})}\|L\hat{\Phi}_{\rp}^{-1}X_{j}\|\left\{b^{'}(X_{j}^{T}\beta_{0})-Y_{j}\right\}X_{j}.
  \end{align*}
  When $r/n\to0$, it is sufficient to show that
  \begin{equation*}
    \frac{1}{\sqrt{n}}\sumjn\sqrt{b^{''}(X_{j}^{T}\hbeta_{\rp})}\|L\hat{\Phi}_{\rp}^{-1}X_{j}\|\left\{b^{'}(X_{j}^{T}\beta_{0})-Y_{j}\right\}X_{j}=O_{p}(1).
  \end{equation*}
  The data points used to estimate $\hbeta_{\rp}$ and $\hat{\Phi}_{\rp}$
  can be ignored, because
  \begin{align*}
    &\frac{1}{\sqrt{n}}\sumjn\sqrt{b^{''}(X_j^T\hbeta_{\rp})}
      \|L\hat{\Phi}_{\rp}^{-1}X_j\|\left\{ b^{'}(X_j^T\beta_0)-Y_j\right\}X_j\\
    &=\frac{1}{\sqrt{n}}\sum_{k=1}^{r_{\rp}}\sqrt{b^{''}(X_k^{*_{\rp}T}\hbeta_{\rp})}
      \|L\Phi_{\rp}^{-1}X_k^{*_{\rp}}\|\left\{b^{'}(X_k^{*_{\rp}T}\beta_0)-Y_k^{*_{\rp}}\right\}X_k^{*_{\rp}}\\
    &\quad+\sqrt{1-\frac{r_{\rp}}{n}}\frac{1}{\sqrt{n-r_{\rp}}}
      \sum_{j=1}^{n-r_{\rp}}\sqrt{b^{''}(X_j^T\hbeta_{\rp})}
      \|L\hat{\Phi}_{\rp}^{-1}X_j\|\left\{ b^{'}(X_j^T\beta_0)-Y_j\right\}X_j\\
    &\lesssim\frac{1}{\sqrt{n}}\sum_{k=1}^{r_{\rp}}
      \|X_k^{*_{\rp}}\|\left\{b^{'}(X_k^{*_{\rp}T}\beta_0)-Y_k^{*_{\rp}}\right\}X_k^{*_{\rp}}\\
    &\quad+\sqrt{1-\frac{r_{\rp}}{n}}\frac{1}{\sqrt{n-r_{\rp}}}\sum_{j=1}^{n-r_{\rp}}\sqrt{b^{''}(X_j^T\hbeta_{\rp})}\|L\hat{\Phi}_{\rp}^{-1}X_j\|\left\{ b^{'}(X_j^T\beta_0)-Y_j\right\}X_j\\
    &=O_p \left( \frac{r_{\rp}}{\sqrt{n}}
      \right)+\sqrt{1-\frac{r_{\rp}}{n}}\frac{1}{\sqrt{n-r_{\rp}}}\sum_{j=1}^{n-r_{\rp}}\sqrt{b^{''}(X_j^T\hbeta_{\rp})}\|L\hat{\Phi}_{\rp}^{-1}X_j\|\left\{ b^{'}(X_j^T\beta_0)-Y_j\right\}X_j,
  \end{align*}
  where $X_k^{*_{\rp}}$ and $Y_k^{*_{\rp}}$ denote the data points in the pilot
  sample. The last equation is due to the fact that
  $\Exp\left[\|X\|\left\{b^{'}(X^{T}\beta_{0})-Y\right\}X\right]<\infty$. Since
  that $r_{\rp}/\sqrt{n}\to0$, we know that the
  data points in the pilot sample can be ignored. We know that
  $r_{\rp}/n\to0$, and thus to ease the notation, we can assume that $X_{j},
  Y_{j}$ are independent of $\hbeta_{\rp}$ and $\hat{\Phi}_{\rp}$. 
  Conditionally on $\hbeta_{\rp}$ and $\hat{\Phi}_{\rp}$, we have
  \begin{align*}
    \left\|\Var\left[\hat{w}_{1}\left\{b^{'}(X_{1}^{T}\beta_{0})-Y_{1}\right\}X_{1}\bigg{\vert}\hbeta_{\rp},\hat{\Phi}_{\rp}\right]\right\|_{F}&\lesssim\left\|\Var\left[\| X_{1}\|\left\{b^{'}(X_{1}^{T}\beta_{0})-Y_{1}\right\}X_{1}\bigg{\vert}\hbeta_{\rp},\hat{\Phi}_{\rp}\right]\right\|_{F}\\
                                                                                                                                                                             &=\Exp\left\{b^{''}(X_{1}^{T}\beta_{0})\| X_{1}\|^{4}\right\}\\
                                                                                                                                                                             &\lesssim\Exp\left(\| X_{1}\|^{4}\right)<\infty.
  \end{align*}
  Here, we can use a normal CLT for i.i.d. data to show that
  \begin{equation*}
    \frac{1}{\sqrt{n}}\sumjn\sqrt{b^{''}(X_{j}^{T}\hbeta_{\rp})}\|L\hat{\Phi}_{\rp}^{-1}X_{j}\|\left\{b^{'}(X_{j}^{T}\beta_{0})-Y_{j}\right\}X_{j}=O_{p\vert\hbeta_{\rp},\hat{\Phi}_{\rp}}(1),
  \end{equation*}
  and therefore $O_{p}(1)$ unconditionally \citep[see][]{xiong2008some, wang2019more}. This is sufficient to show that
  \begin{equation*}
    \sqrt{r}Q=o_{p}(1).
  \end{equation*}
  Therefore, we have
  \begin{equation*}
    \sqrt{r}\Psi_{\ruw}^{*\mathrm{m}}(\beta_{0})=\frac{1}{\sqrt{r}}\sumr
    M_{n,i}+o_{p}(1).
  \end{equation*}
  Next, our goal is to show the asymptotic normality of
  $\left(\sumr M_{n,i}\right)/\sqrt{r}$, using the martingale CLT for
  Hilbert space. We can check the conditions of the martingale CLT for Hilbert
  space. Denote
  \begin{equation*}
    \xi_{n,i}=\frac{1}{\sqrt{r}}M_{n,i},
  \end{equation*}
  and we have $\Exp(\xi_{n,i}\vert\mathcal{F}_{n,i-1})=0$. In addition, we
  have $\Exp(\|\xi_{n,i}\|^{2})<\infty$ because
  \begin{align*}
    &\left\{\Exp\left(\| M_{n,1}\|^{2}\right)\right\}^{\frac{1}{2}}\\
    &\leq\left(\Exp\left[\left\|\hat{m}\{b^{'}(X_{1}^{*{T}}\beta_{0})-Y_{1}^{*}\}X_{1}^{*}\right\|^{2}\right]\right)^{\frac{1}{2}}+\left(\Exp\left[\left\|\onen\sumjn\hat{w}_{j}\left\{b^{'}(X_{j}^{T}\beta_{0})-Y_{j}\right\}X_{j}\right\|^{2}\right]\right)^{\frac{1}{2}}\\
    &:=I_{1}^{\frac{1}{2}}+I_{2}^{\frac{1}{2}}<\infty.
  \end{align*}
  We know $I_{1}<\infty$ because
  \begin{align*}
    I_{1}&=\Exp\left[\hat{m}^{2}\{b^{'}(X_{1}^{*{T}}\beta_{0})-Y_{1}^{*}\}^{2}\left\| X_{1}^{*}\right\|^{2}\right]=\Exp\left(\Exp\left[\hat{m}^{2}\{b^{'}(X_{1}^{*{T}}\beta_{0})-Y_{1}^{*}\}^{2}\left\| X_{1}^{*}\right\|^{2}\bigg{\vert}\mathcal{F}_{n,0}\right]\right)\\
         &=\Exp\left[\hat{m}\onen\sumjn\hat{w}_{j}\{b^{'}(X_{j}^{T}\beta_{0})-Y_{j}\}^{2}\| X_{j}\|^{2}\right]\\
         &\lesssim\Exp\left(\left(\onen\sumjn\| X_{j}\|\right)\left[\onen\sumjn\{b^{'}(X_{j}^{T}\beta_{0})-Y_{j}\}^{2}\| X_{j}\|^{3}\right]\right)\\
         &=\Exp\left\{\Exp\left(\left(\onen\sumjn\| X_{j}\|\right)\left[\onen\sumjn\{b^{'}(X_{j}^{T}\beta_{0})-Y_{j}\}^{2}\| X_{j}\|^{3}\right]\bigg{\vert} X_{1}^{n}\right)\right\}\\
         &=\Exp\left[\left(\onen\sumjn\|
           X_{j}\|\right)\left\{\onen\sumjn b^{''}(X_{j}^{T}\beta_{0})\| X_{j}\|^{3}\right\}\right]\\
         &\lesssim\Exp\left\{\left(\onen\sumn\| X_{i}\|\right)\left(\onen\sumjn\| X_{j}\|^{3}\right)\right\}\\
         &=\Exp\left(\frac{1}{n^{2}}\sumn\| X_{i}\|^{4}+\frac{1}{n^{2}}\sum_{i\neq j}\| X_{i}\|\| X_{j}\|^{3}\right)\\
         &=\onen\Exp(\| X\|^{4})+\frac{n(n-1)}{n^{2}}\Exp(\| X\|)\Exp(\| X\|^{3})<\infty.
  \end{align*}
  Similarly, using Minkowski inequality, we know $I_{2}<\infty$ because
  \begin{align*}
    I_{2}&\leq\onen\sumjn\Exp\left[\left\|\hat{w}_{j}\left\{b^{'}(X_{j}^{T}\beta_{0})-Y_{j}\right\}X_{j}\right\|^{2}\right]\\
         &\lesssim\onen\sumjn\Exp\left[\left\{b^{'}(X_{j}^{T}\beta_{0})-Y_{j}\right\}^{2}\left\| X_{j}\right\|^{4}\right]\\
         &=\onen\sumjn\Exp\left\{b^{''}(X_{j}^{T}\beta_{0})\left\| X_{j}\right\|^{4}\right\}\\
         &\lesssim\Exp(\left\| X\right\|^{4})<\infty.
  \end{align*}
  Therefore $\Exp(\| M_{n,1}\|^{2})<\infty$, which means
  $\Exp(\|\xi_{n,i}\|^{2})<\infty$. We then verify the three
  conditions of Lemma \ref{lm:MartingaleCLT}.  Condition \ref{c1} and \ref{c3}
  are trivial to verify. For condition \ref{c1}, we can see
  \begin{align*}
    \sumr\Exp\left(\left\|\xi_{n,i}\right\|^{2}\vert\mathcal{F}_{n,i-1}\right)&=\Exp\left(\| M_{n,1}\|^{2}\vert\mathcal{F}_{n,0}\right)\\
                                                                             &=\onen\sum_{k=1}^{n}\hat{w}_{k}\onen\sumjn\hat{w}_{j}\{b^{'}(X_{j}^{T}\beta_{0})-Y_{j}\}^{2}\| X_{j}\|^{2}\\
                                                                             &\quad-\onen\sum_{k=1}^{n}\hat{w}_{k}\{b^{'}(X_{k}^{T}\beta_{0})-Y_{k}\}X_{k}^{T}\onen\sumjn\hat{w}_{j}\{b^{'}(X_{j}^{T}\beta_{0})-Y_{j}\}X_{j}.
  \end{align*}
  We have already known that
  $\| X_{j}\|\{b^{'}(X_{j}^{T}\beta_{0})-Y_{j}\}X_{j}$ is integrable and
  also
  \begin{align*}
    \Exp\left[\{b^{'}(X_{j}^{T}\beta_{0})-Y_{j}\}^{2}\| X_{j}\|^{3}\right]&=\Exp\left(\Exp\left[\{b^{'}(X_{j}^{T}\beta_{0})-Y_{j}\}^{2}\| X_{j}\|^{3}\bigg{\vert} X_{j}\right]\right)\\
                                                                       &=\Exp\left\{b^{''}(X_{j}^{T}\beta_{0})\| X_{j}\|^{3}\right\}\\
                                                                       &\lesssim\Exp\left(\| X_{j}\|^{3}\right)<\infty.
  \end{align*}
  Therefore, $\| X_{j}\|$,
  $\| X_{j}\|\{b^{'}(X_{j}^{T}\beta_{0})-Y_{j}\}X_{j}$ and
  $\{b^{'}(X_{j}^{T}\beta_{0})-Y_{j}\}^{2}\| X_{j}\|^{3}$ are all
  integrable. Then, applying Lemma \ref{lm:L1_in_MoreEffi} and the Law of Large
  Numbers, we respectively have that
  \begin{align*}
    \onen\sum_{k=1}^{n}\hat{w}_{k}&=\onen\sum_{k=1}^{n}\sqrt{b^{''}(X_{k}^{T}\beta_{0})}\|L\Phi^{-1}X_{k}\|+o_{p}(1)\\
                                  &=\Exp\left\{\sqrt{b^{''}(X^{T}\beta_{0})}\|L\Phi^{-1}X\|\right\}+o_{p}(1),
  \end{align*}
  and
  \begin{align*}
    &\onen\sumjn\hat{w}_{j}\{b^{'}(X_{j}^{T}\beta_{0})-Y_{j}\}^{2}\| X_{j}\|^{2}\\
    &=\onen\sumjn\sqrt{b^{''}(X_{j}^{T}\beta_{0})}\|L\Phi^{-1}X_{j}\|\{b^{'}(X_{j}^{T}\beta_{0})-Y_{j}\}^{2}\| X_{j}\|^{2}+o_{p}(1)\\
    &=\Exp\left[\left\{b^{''}(X_{j}^{T}\beta_{0})\right\}^{\frac{3}{2}}\|L\Phi^{-1}X_{j}\|\| X_{j}\|^{2}\right]+o_{p}(1),
  \end{align*}
  and
  \begin{align*}
    &\onen\sumjn\hat{w}_{j}\{b^{'}(X_{j}^{T}\beta_{0})-Y_{j}\} X_{j}\\
    &=\onen\sumjn\sqrt{b^{''}(X_{j}^{T}\beta_{0})}\|L\Phi^{-1}X_{j}\|\{b^{'}(X_{j}^{T}\beta_{0})-Y_{j}\} X_{j}+o_{p}(1)\\
    &=o_{p}(1).
  \end{align*}
  Therefore, we have
  \begin{align*}
    \sumr\Exp\left(\left\|\xi_{n,i}\right\|^{2}\vert\mathcal{F}_{n,i-1}\right)
    &=\onen\sum_{k=1}^{n}\hat{w}_{k}\onen\sumjn\hat{w}_{j}\{b^{'}(X_{j}^{T}\beta_{0})-Y_{j}\}^{2}\| X_{j}\|^{2}\\
    &\quad-\onen\sum_{k=1}^{n}\hat{w}_{k}
      \{b^{'}(X_{k}^{T}\beta_{0})-Y_{k}\}X_{k}^{T}\onen\sumjn\hat{w}_{j}
      \{b^{'}(X_{j}^{T}\beta_{0})-Y_{j}\}X_{j}\\
    &\xrightarrow{p}\Exp\left\{\sqrt{b^{''}(X^{T}\beta_{0})}\|L\Phi^{-1}X\|\right\}\Exp\left[\left\{b^{''}(X^{T}\beta_{0})\right\}^{\frac{3}{2}}\|L\Phi^{-1}X\|\| X\|^{2}\right]\\
    &=\text{tr}(m\Gamma),
  \end{align*}
  which means condition \ref{c1} is verified. For condition \ref{c3}, similarly
  as condition \ref{c1}, we can prove that
  \begin{align}
    \sumr\Exp\left(\xi_{n,i}\xi_{n,i}^{T}\vert\mathcal{F}_{n,i-1}\right)
    &=\Exp\left( M_{n,1}M_{n,1}^{T}\vert\mathcal{F}_{n,0}\right)\notag\\
    &=\onen\sum_{k=1}^{n}\hat{w}_{k}\onen\sumjn\hat{w}_{j}\{b^{'}(X_{j}^{T}\beta_{0})-Y_{j}\}^{2}X_{j}X_{j}^{T}\notag\\
    &\quad-\onen\sum_{k=1}^{n}\hat{w}_{k}\{b^{'}(X_{k}^{T}\beta_{0})-Y_{k}\}X_{k}\onen\sumjn\hat{w}_{j}\{b^{'}(X_{j}^{T}\beta_{0})-Y_{j}\}X_{j}^{T}
      \notag\\
    &\xrightarrow{p}\Exp\left\{\sqrt{b^{''}(X^{T}\beta_{0})}\|L\Phi^{-1}X\|\right\}\Exp\left[\left\{b^{''}(X^{T}\beta_{0})\right\}^{\frac{3}{2}}\|L\Phi^{-1}X\| XX^{T}\right]\notag\\
    &=m\Gamma.\label{eq:2}
  \end{align}
  Let the orthonormal basis $e_{i}$ be $e_{1}=(1,0,0,...,0)^{T}$,
  $e_{2}=(0,1,0,...,0)^{T}$,..., then condition \ref{c3} is the same as the
  convergence in probability of each entry of
  $\Exp\left[M_{n,1}M_{n,1}^{T}\vert\mathcal{F}_{n,0}\right]$, which is
  guaranteed. To prove condition \ref{c2}, we first show that
  $\sup_{n}\Exp\left[\|
    M_{n,1}\|^{2+\delta}\vert\right]<\infty$. We have
  \begin{align*}
    &\left\{\Exp\left(\left\| M_{n,1}\right\|^{2+\delta}\right)\right\}^{\frac{1}{2+\delta}}\\
    &\leq\left(\Exp\left[\left\|\hat{m}\{b^{'}(X_{1}^{*{T}}\beta_{0})-Y_{1}^{*}\}X_{1}^{*}\right\|^{2+\delta}\right]\right)^{\frac{1}{2+\delta}}+\left(\Exp\left[\left\|\onen\sumjn\hat{w}_{j}\{b^{'}(X_{j}^{T}\beta_{0})-Y_{j}\}X_{j}\right\|^{2+\delta}\right]\right)^{\frac{1}{2+\delta}}\\
    &:=I_{1}^{\frac{1}{2+\delta}}+I_{2}^{\frac{1}{2+\delta}}.
  \end{align*}
  For $I_{1}$, we know
  \begin{align*}
    I_{1}&=\Exp\left[\left\|\hat{m}\{b^{'}(X_{1}^{*{T}}\beta_{0})-Y_{1}^{*}\}X_{1}^{*}\right\|^{2+\delta}\right]\\
         &=\Exp\left(\Exp\left[\left\|\hat{m}\{b^{'}(X_{1}^{*{T}}\beta_{0})-Y_{1}^{*}\}X_{1}^{*}\right\|^{2+\delta}\bigg{\vert}\mathcal{F}_{n,0}\right]\right)\\
         &=\Exp\left\{\hat{m}^{1+\delta}\onen\sumjn\hat{w}_{j}\left\vert b^{'}(X_{j}^{T}\beta_{0})-Y_{j}\right|^{2+\delta}\| X_{j}\|^{2+\delta}\right\}\\
         &\lesssim\Exp\left[\left(\onen\sumjn\| X_{j}\|\right)^{1+\delta}\left\{\onen\sumjn\left| b^{'}(X_{j}^{T}\beta_{0})-Y_{j}\right|^{2+\delta}\| X_{j}\|^{3+\delta}\right\}\right]\\
         &\leq\Exp\left[\left(\onen\sumn\| X_{i}\|^{1+\delta}\right)\left\{\onen\sumjn\left| b^{'}(X_{j}^{T}\beta_{0})-Y_{j}\right|^{2+\delta}\| X_{j}\|^{3+\delta}\right\}\right]\\
         &=\Exp\left\{\frac{1}{n^{2}}\sumn\left\vert b^{'}(X_{j}^{T}\beta_{0})-Y_{j}\right|^{2+\delta}\| X_{j}\|^{4+2\delta}+\frac{1}{n^{2}}\sum_{i\neq j}\| X_{i}\|^{1+\delta}\left\vert b^{'}(X_{j}^{T}\beta_{0})-Y_{j}\right|^{2+\delta}\| X_{j}\|^{3+\delta}\right\}\\
         &=\onen\Exp\left\{\left\vert b^{'}(X^{T}\beta_{0})-Y\right|^{2+\delta}\| X\|^{4+2\delta}\right\}+\frac{n(n-1)}{n^{2}}\Exp\left(\| X\|^{1+\delta}\right)\Exp\left\{\left\vert b^{'}(X^{T}\beta_{0})-Y\right|^{2+\delta}\| X\|^{3+\delta}\right\}.
  \end{align*}
  For $I_{2}$, we know
  \begin{align*}
    I_{2}&\leq\onen\sumjn\Exp\left[\left\|\hat{w}_{j}\left\{b^{'}(X_{j}^{T}\beta_{0})-Y_{j}\right\}X_{j}\right\|^{2+\delta}\right]\\
         &\lesssim\onen\sumjn\Exp\left\{\left| b^{'}(X_{j}^{T}\beta_{0})-Y_{j}\right|^{2+\delta}\left\| X_{j}\right\|^{4+2\delta}\right\}\\
         &=\Exp\left\{\left| b^{'}(X^{T}\beta_{0})-Y\right|^{2+\delta}\left\| X\right\|^{4+2\delta}\right\}.
  \end{align*}
  Since we include the intercept in the model, we have $\| X\|\geq 1$. Then,
  it is easy to know that
  $\Exp\left\{\left\vert
      b^{'}(X^{T}\beta_{0})-Y\right|^{2+\delta}\|
    X\|^{3+\delta}\right\}\leq \Exp\left\{\left\vert
      b^{'}(X^{T}\beta_{0})-Y\right|^{2+\delta}\|
    X\|^{4+2\delta}\right\}$. Then,
  $\sup_{n}\Exp\left(\| M_{n,1}\|^{2+\delta}\right)<\infty$, and
  condition \ref{c2} is due to
  $\sup_{n}\Exp\left(\| M_{n,1}\|^{2+\delta}\vert\right)<\infty$,
  because
  \begin{equation*}
    \sumr\Exp\left[\|\xi_{n,i}\|^{2}I\{\|\xi_{n,i}\|^{2}>\varepsilon\}\vert\mathcal{F}_{n,i-1}\right]=\Exp\left[\| M_{n,1}\|^{2}I\{\| M_{n,1}\|^{2}>\varepsilon r\}\vert\mathcal{F}_{n,0}\right],
  \end{equation*}
  and applying Markov inequality, we have
  \begin{align*}
    \Exp\left(\Exp\left[\| M_{n,1}\|^{2}I\{\| M_{n,1}\|^{2}>\varepsilon r\}\vert\mathcal{F}_{n,0}\right]\right)&=\Exp\left[\| M_{n,1}\|^{2}I\{\| M_{n,1}\|^{2}>\varepsilon r\}\right]\\
                                                                                                                                           &\leq\frac{1}{\varepsilon^{\frac{\delta}{2}}r^{\frac{\delta}{2}}}\Exp\left[\| M_{n,1}\|^{2+\delta}I\{\| M_{n,1}\|^{2}>\varepsilon r\}\right]\\
                                                                                                                                           &\leq\frac{1}{\varepsilon^{\frac{\delta}{2}}r^{\frac{\delta}{2}}}\Exp\left(\| M_{n,1}\|^{2+\delta}\right)\\
                                                                                                                                           &\leq\frac{1}{\varepsilon^{\frac{\delta}{2}}r^{\frac{\delta}{2}}}\sup_{n}\Exp\left(\| M_{n,1}\|^{2+\delta}\right)\to0.
  \end{align*}
  Therefore,
  \begin{equation*}
    \frac{1}{\sqrt{r}}\sumr M_{n,i}\xrightarrow{d}N(0,m\Gamma).
  \end{equation*}
  Then, the three conditions of Lemma \ref{lm:MartingaleCLT} have been verified
  and we know
  \begin{equation*}
    \sqrt{r}\Psi_{\ruw}^{*\mathrm{m}}(\beta_{0})\xrightarrow{d}N(0,m\Gamma).
  \end{equation*}	
\end{proof}

Next, for the more general case of $r/n\to\rho\in[0,1)$, we prove that
\begin{equation*}
  \sqrt{r}\Psi_{\ruw}^{*\mathrm{m}}(\beta_{0})
  \xrightarrow{d}N(0,m\Gamma+\rho\Omega),		
\end{equation*}
under a stronger moment condition:
\begin{equation*}
	\Exp\left\{\left| b^{'}(X^{T}\beta_{0})-Y\right|^{4+2\delta}
    \left\| X\right\|^{8+4\delta}\right\}<\infty.
\end{equation*}
We prove the result with the following martingle CLT:
\begin{lemma}[Multivariate version of martingale
  CLT]\label{lm:MultiMartingaleCLT}
	For $k=1,2,3,...,$ let $\{\xi_{ki}; i=1,2,...,N_{k}\}$ be a martingale
  difference sequence in $\mathbb{R}^{p}$ relative to the filtration
  $\{\mathcal{F}_{ki};i=0,1,...,N_{k}\}$ and let $Y_{k}\in\mathbb{R}^{p}$ be an
  $\mathcal{F}_{k0}$-measurable random vector. Set
  $S_{k}=\sum_{i=1}^{N_{k}}\xi_{ki}$. Assume the following conditions.
  \begin{enumerate}
	\item\label{cd1}
    $\lim\limits_{k\to\infty}\sum_{i=1}^{N_{k}}\Exp(\|\xi_{ki}\|^{4})=0$.
	\item\label{cd2}
    $\lim\limits_{k\to\infty}\Exp\left\{\left\|\sum_{i=1}^{N_{k}}\Exp(\xi_{ki}\xi_{ki}^{T}\vert\mathcal{F}_{k,i-1})-B_{k}\right\|^{2}\right\}=0$
    for some sequence of positive definite matrices $\{B_{k}\}_{k=1}^{\infty}$
    with $\sup_{k}\lambda_{\max}(B_{k})<\infty$ i.e. the largest eigenvalue is
    uniformly bounded.
	\item\label{cd3} For some probability distribution $L_{0}$, $*$ denotes
    convolution and $L(\cdot)$ denotes the law of random variables:
    \begin{equation*}
      L(Y_{k})*N(0,B_{k})\xrightarrow{d}L_{0}.
    \end{equation*}
  \end{enumerate}
  Then we have
  \begin{equation*}
    L(Y_{k}+S_{k})\xrightarrow{d}L_{0}.
  \end{equation*}
\end{lemma}
We refer to Lemma 5 of \cite{Zhang2020optimal}. Then we can present the
proof.%
\begin{proof}[\textbf{Proof for the case of $\rho\in[0,1)$}]
  We know that
  \begin{equation*}
    \sqrt{r}\Psi_{\ruw}^{*\mathrm{m}}(\beta_{0})=\frac{1}{\sqrt{r}}\sumr M_{n,i}+\sqrt{r}Q.
  \end{equation*}
  Denote
  \begin{equation*}
    B_{k}=m\Gamma=\Exp\left\{\sqrt{b^{''}(X^{T}\beta_{0})}\|L\Phi^{-1}X\|\right\}\Exp\left[\left\{b^{''}(X^{T}\beta_{0})\right\}^{\frac{3}{2}}\|L\Phi^{-1}X\| XX^{T}\right].
  \end{equation*}
   In \eqref{eq:2}, we have already proved that
  \begin{align*}
    \sumr\Exp\left(\xi_{n,i}\xi_{n,i}^{T}\vert\mathcal{F}_{n,i-1}\right)\xrightarrow{p}m\Gamma.
  \end{align*}
  If we can prove $\Exp\left(M_{n,1}M_{n,1}^{T}\vert\mathcal{F}_{n,0}\right)$ is
  $L^{2}$ uniformly integrable, then we can prove
  $\Exp\left(M_{n,1}M_{n,1}^{T}\vert\mathcal{F}_{n,0}\right)\xrightarrow{L^{2}}m\Gamma$,
  which is condition \ref{cd2}. Since
  \begin{equation*}
    \Exp\left\{\left\|\Exp\left( M_{n,1}M_{n,1}^{T}\vert\mathcal{F}_{n,0}\right)\right\|^{2+\delta}\right\}\leq\Exp\left(\left\| M_{n,1}M_{n,1}^{T}\right\|^{2+\delta}\right)=\Exp\left(\left\| M_{n,1}\right\|^{4+2\delta}\right),
  \end{equation*}
  it is sufficient to show
    $\sup_{n}\Exp\left(\left\|
      M_{n,1}\right\|^{4+2\delta}\right)<\infty$. This is true because
  \begin{align*}
    &\left\{\Exp\left(\left\| M_{n,1}\right\|^{4+2\delta}\right)\right\}^{\frac{1}{4+2\delta}}\\
    &\leq\left(\Exp\left[\left\|\hat{m}\{b^{'}(X_{1}^{*{T}}\beta_{0})-Y_{1}^{*}\}X_{1}^{*}\right\|^{4+2\delta}\right]\right)^{\frac{1}{4+2\delta}}+\left(\Exp\left[\left\|\onen\sumjn\hat{w}_{j}\{b^{'}(X_{j}^{T}\beta_{0})-Y_{j}\}X_{j}\right\|^{4+2\delta}\right]\right)^{\frac{1}{4+2\delta}}\\
    &:=I_{1}^{\frac{1}{4+2\delta}}+I_{2}^{\frac{1}{4+2\delta}},
  \end{align*}
  and for $I_{1}$, we know
  \begin{align*}
    I_{1}&=\Exp\left[\left\|\hat{m}\{b^{'}(X_{1}^{*{T}}\beta_{0})-Y_{1}^{*}\}X_{1}^{*}\right\|^{4+2\delta}\right]\\
         &=\Exp\left(\Exp\left[\left\|\hat{m}\{b^{'}(X_{1}^{*{T}}\beta_{0})-Y_{1}^{*}\}X_{1}^{*}\right\|^{4+2\delta}\bigg{\vert}\mathcal{F}_{n,0}\right]\right)\\
         &=\Exp\left\{\hat{m}^{3+2\delta}\onen\sumjn\hat{w}_{j}\left\vert b^{'}(X_{j}^{T}\beta_{0})-Y_{j}\right|^{4+2\delta}\| X_{j}\|^{4+2\delta}\right\}\\
         &\lesssim\Exp\left[\left(\onen\sumjn\| X_{j}\|\right)^{3+2\delta}\left\{\onen\sumjn\left\vert b^{'}(X_{j}^{T}\beta_{0})-Y_{j}\right|^{4+2\delta}\| X_{j}\|^{5+2\delta}\right\}\right]\\
         &\leq\Exp\left[\left(\onen\sumn\| X_{i}\|^{3+2\delta}\right)\left\{\onen\sumjn\left\vert b^{'}(X_{j}^{T}\beta_{0})-Y_{j}\right|^{4+2\delta}\| X_{j}\|^{5+2\delta}\right\}\right]\\
         &=\Exp\left\{\frac{1}{n^{2}}\sumn\left\vert b^{'}(X_{j}^{T}\beta_{0})-Y_{j}\right|^{4+2\delta}\| X_{j}\|^{8+4\delta}+\frac{1}{n^{2}}\sum_{i\neq j}\| X_{i}\|^{3+2\delta}\left\vert b^{'}(X_{j}^{T}\beta_{0})-Y_{j}\right|^{4+2\delta}\| X_{j}\|^{5+2\delta}\right\}\\
         &=\onen\Exp\left\{\left\vert b^{'}(X^{T}\beta_{0})-Y\right|^{4+2\delta}\| X\|^{8+4\delta}\right\}+\frac{n(n-1)}{n^{2}}\Exp\left(\| X\|^{3+2\delta}\right)\Exp\left\{\left\vert b^{'}(X^{T}\beta_{0})-Y\right|^{4+2\delta}\| X\|^{5+2\delta}\right\},
  \end{align*}
  and for $I_{2}$, we know
  \begin{align*}
    I_{2}&\leq\onen\sumjn\Exp\left[\left\|\hat{w}_{j}\left\{b^{'}(X_{j}^{T}\beta_{0})-Y_{j}\right\}X_{j}\right\|^{4+2\delta}\right]\\
         &\lesssim\onen\sumjn\Exp\left\{\left| b^{'}(X_{j}^{T}\beta_{0})-Y_{j}\right\rvert^{4+2\delta}\left\| X_{j}\right\|^{8+4\delta}\right\}\\
         &=\Exp\left\{\left| b^{'}(X^{T}\beta_{0})-Y\right|^{4+2\delta}\left\| X\right\|^{8+4\delta}\right\}.
  \end{align*}
  Then, we have
  $\sup_{n}\Exp\left(\left\|
      M_{n,1}\right\|^{4+2\delta}\right)<\infty$. In addition, we also know
  that condition \ref{cd1} is guaranteed, since
  \begin{equation*}
    \sumr\Exp(\|\xi_{n,i}\|^{4})=\oner\Exp\left(\| M_{n,1}\|^{4}\right)\to0.
  \end{equation*}
  Now, we verify condition \ref{cd3}. We know that $Q$ is
  $\mathcal{F}_{n,0}$-measurable and
  \begin{align*}
    \sqrt{r}Q&=\sqrt{r}\onen\sumjn\sqrt{b^{''}(X_{j}^{T}\hbeta_{\rp})}\|L
               \hat{\Phi}_{\rp}^{-1}X_{j}\|\left\{b^{'}(X_{j}^{T}\beta_{0})-Y_{j}\right\}X_{j}\\
             &=\sqrt{\frac{r}{n}}\frac{1}{\sqrt{n}}\sumjn\sqrt{b^{''}(X_{j}^{T}\hbeta_{\rp})}\|L\hat{\Phi}_{\rp}^{-1}X_{j}\|\left\{b^{'}(X_{j}^{T}\beta_{0})-Y_{j}\right\}X_{j}.
  \end{align*}
  Since $r_{\rp}/\sqrt{n}\to0$, we can use the same deduction in
    the case of $\rho=0$ and know that the data points used to estimate
  $\hbeta_{\rp}$ and $\hat{\Phi}_{\rp}$ can be ignored. Therefore, we
  assume $\hbeta_{\rp}$ and $\hat{\Phi}_{\rp}$ are independent with $X_{j},
  Y_{j}$'s. Now we denote
  \begin{equation*}
    \tilde{\Omega}=\Exp\left[\left\{b^{''}(X^{T}\beta_{0})\right\}^{2}\|L\Phi^{-1}X\|^{2}l^{T}XX^{T}l\right],
  \end{equation*}
  for every $l\in\mathbb{R}^{p}$, and
  \begin{equation*}
    \tau_{n_{i}}=\tilde{\Omega}^{-\frac{1}{2}}\sqrt{b^{''}(X_{i}^{T}\hbeta_{\rp})}\|L\hat{\Phi}_{\rp}^{-1}X_{i}\|\left\{b^{'}(X_{i}^{T}\beta_{0})-Y_{i}\right\}l^{T}X_{i}.
  \end{equation*}
  Then, $\tau_{n_{i}}$'s are i.i.d. conditional on $\hbeta_{\rp}$
  and $\hat{\Phi}_{\rp}$. Thus, they are interchangeable due to Theorem
  7.3.2 of \cite{ChowTeicher2003}. We now apply the central limit theorem of
  interchangeable random variables in Theorem 2 of \cite{blum1958central}. We
  should verify the three conditions of the theorem. The first condition is
  trivial to verify, $\forall i\neq j$,
  \begin{equation*}
    \Exp(\tau_{n_{i}}\tau_{n_{j}})=\Exp\left\{\Exp(\tau_{n_{i}}\tau_{n_{j}}\vert\hbeta_{\rp},\hat{\Phi}_{\rp})\right\}=0.
  \end{equation*}
  Then, we show $\Exp(|\tau_{n_{i}}|^{3})=o(\sqrt{n})$ because
  \begin{align*}
    \Exp(|\tau_{n_{i}}|^{3})&=\tilde{\Omega}^{-\frac{3}{2}}\Exp\left[\left\{b^{''}(X_{i}^{T}\hbeta_{\rp})\right\}^{\frac{3}{2}}\|L\hat{\Phi}_{\rp}^{-1}X_{i}\|^{3}\left| b^{'}(X_{i}^{T}\beta_{0})-Y_{i}\right|^{3}(l^{T}X_{i})^{3}\right]\\
                           &\lesssim \tilde{\Omega}^{-\frac{3}{2}}\| l\|^{3}\Exp\left\{\left| b^{'}(X_{i}^{T}\beta_{0})-Y_{i}\right|^{3}\| X_{i}\|^{6}\right\}\\
                           &=o(\sqrt{n}).
  \end{align*}
  Next, we verify that $\forall i\neq j$,
  $\Exp\left(\tau_{n_{i}}^{2}\tau_{n_{j}}^{2}\right)\to1$. We can see
  \begin{align*}
    \tau_{n_{i}}^{2}\tau_{n_{j}}^{2}&=\tilde{\Omega}^{-2}\hat{w}_{i}^{2}\left| b^{'}(X_{i}^{T}\beta_{0})-Y_{i}\right|^{2}(l^{T}X_{i})^{2}\hat{w}_{j}^{2}\left| b^{'}(X_{j}^{T}\beta_{0})-Y_{j}\right|^{2}(l^{T}X_{j})^{2}\\
                                 &\lesssim\tilde{\Omega}^{-2}\| l\|^{4}\left| b^{'}(X_{i}^{T}\beta_{0})-Y_{i}\right|^{2}\| X_{i}\|^{4}\left| b^{'}(X_{j}^{T}\beta_{0})-Y_{j}\right|^{2}\| X_{j}\|^{4},
  \end{align*}
  and
  \begin{equation*}
    \Exp\left\{\left| b^{'}(X_{i}^{T}\beta_{0})-Y_{i}\right|^{2}\| X_{i}\|^{4}\left| b^{'}(X_{j}^{T}\beta_{0})-Y_{j}\right|^{2}\| X_{j}\|^{4}\right\}=\Exp\left\{b^{''}(X^{T}\beta_{0})\| X\|^{4}\right\}^{2}<\infty.
  \end{equation*}
  Therefore, by the dominating convergence theorem, we have
  \begin{align*}
    &\Exp\left(\tau_{n_{i}}^{2}\tau_{n_{j}}^{2}\right)\\
    &\to\Exp\bigg{\{}\tilde{\Omega}^{-2}b^{''}(X_{i}^{T}\beta_{0})\|L\Phi^{-1}X_{i}\|^{2}\left| b^{'}(X_{i}^{T}\beta_{0})-Y_{i}\right|^{2}(l^{T}X_{i})^{2}\times\\
    &\quad\quad\quad b^{''}(X_{j}^{T}\beta_{0})\|L\Phi^{-1}X_{j}\|^{2}\left| b^{'}(X_{j}^{T}\beta_{0})-Y_{j}\right|^{2}(l^{T}X_{j})^{2}\bigg{\}}\\
    &=1.
  \end{align*}
  We then use Theorem 2 of \cite{blum1958central} and obtain
  \begin{equation*}
    \frac{1}{\sqrt{n}}\sumn\tau_{n_{i}}\xrightarrow{d}N(0,1).
  \end{equation*}
  Thus, from Cram\'er-Wold device, we know that
  \begin{equation*}
    \sqrt{r}Q\xrightarrow{d}N(0,\rho \Omega).
  \end{equation*}
  We use $\phi_{X}(t)$ to denote the characteriatic function of random vector
  $X$, then
  \begin{equation*}
    \phi_{\sqrt{r}Q}(t)=\Exp e^{it^{T}\sqrt{r}Q}\to e^{-\frac{1}{2}t^{T}\rho\Omega t},
  \end{equation*}
  and therefore
  \begin{equation*}
    \phi_{\sqrt{r}Q}(t)\cdot\phi_{N(0,m\Gamma)}(t)=\Exp e^{it^{T}\sqrt{r}Q}e^{-\frac{1}{2}t^{T}(m\Gamma)t}\to e^{-\frac{1}{2}t^{T}(m\Gamma+\rho\Omega)t}.
  \end{equation*}
  Let $L_{0}=N(0,m\Gamma+\rho\Omega)$. We then have
  \begin{equation*}
    L\left(\sqrt{r}Q\right)*N(0,m\Gamma)\xrightarrow{d}L_{0},
  \end{equation*}
  which means condition \ref{cd3} is true. Therefore, we have already verified
  three conditions of Lemma \ref{lm:MultiMartingaleCLT} and know that
  \begin{equation*}
    \sqrt{r}\Psi_{\ruw}^{*\mathrm{m}}(\beta_{0})=\sumr\xi_{n,i}+\sqrt{r}Q\xrightarrow{d}N(0,m\Gamma+\rho\Omega).
  \end{equation*}
\end{proof}

\subsubsection{Proof of Theorem \ref{th:Normality}}\label{sec:Normality}

\begin{proof}[\textbf{Proof of Theorem \ref{th:Normality}}]
  We now show the asymptotic normality of $\hbeta_{\ruw}$. We know that the
  estimator $\hbeta_{\ruw}$ is defined as \eqref{eq:def_beta_uw}; therefore
  it is the minimizer of
  \begin{equation*}
    \lambda_{\ruw}^{*}(\beta)=\hat{m}\sumr\left\{b(X_{i}^{*{T}}\beta)-Y_{i}^{*}X_{i}^{*{T}}\beta\right\}.
  \end{equation*}
  Thus, $\sqrt{r}(\hbeta_{\ruw}-\beta_{0})$ is the minimizer of
  \begin{equation*}
    \gamma(s)=\lambda_{\ruw}^{*}(\beta_{0}+s/\sqrt{r})-\lambda_{\ruw}^{*}(\beta_{0}).
  \end{equation*}
  Applying Taylor's Theorem, we have
  \begin{align*}
    \gamma(s)&=\frac{1}{\sqrt{r}}s^{T}\dot{\lambda}_{\ruw}^{*}(\beta_{0})+\frac{1}{2r}s^{T}\Ddot{\lambda}_{\ruw}^{*}\left(\beta_{0}+\acute{s}/\sqrt{r}\right)s\\
             &=s^{T}\sqrt{r}\Psi_{\ruw}^{*\mathrm{m}}(\beta_{0})+\frac{1}{2}s^{T}\Dot{\Psi}_{\ruw}^{*\mathrm{m}}\left(\beta_{0}+\acute{s}/\sqrt{r}\right)s\\
             &=s^{T}\sqrt{r}\Psi_{\ruw}^{*\mathrm{m}}(\beta_{0})+\frac{1}{2}s^{T}\left\{\Gamma\left(\beta_{0}\right)+o_{p}(1)\right\}s\\
             &=s^{T}\sqrt{r}\Psi_{\ruw}^{*\mathrm{m}}(\beta_{0})+\frac{1}{2}s^{T}\Gamma s+o_{p}(\| s\|^{2}).
  \end{align*}
  The third equation is due to Lemma \ref{lm:Convergence_of_Gamma} and
  $\beta_{0}+\acute{s}/\sqrt{r}\xrightarrow{p}\beta_{0}$. Then we apply the
  Basic Corollary in page 2 of \cite{hjort2011asymptotics} and obtain
  \begin{equation*}
    \sqrt{r}(\hbeta_{\ruw}-\beta_{0})=-\Gamma^{-1}\sqrt{r}\Psi_{\ruw}^{*\mathrm{m}}(\beta_{0})+o_{p}(1).
  \end{equation*}
  Now, due to Lemma \ref{lm:Normality_of_Psi} and Slutsky's theorem, we have
  \begin{equation*}
    \sqrt{r}(\hbeta_{\ruw}-\beta_{0})\xrightarrow{d}N(0,\Sigma_{\ruw}^{\rho}),
  \end{equation*}
  under the conditions of Lemma \ref{lm:Normality_of_Psi}. This completes the
  proof.
\end{proof}

\subsection{Proofs of Efficency Comparison}\label{sec:proof_effi}

In this section, we present the technique details in Section
\ref{sec:efficiency}. The proof of Equation \eqref{eq:weightedvariance} is
presented in Section \ref{sec:eq_weightedvar}. The proof of Theorem
\ref{th:Efficiency} is presented in Section \ref{sec:th_effi}.

\subsubsection{Proof of Equation
  \eqref{eq:weightedvariance}}\label{sec:eq_weightedvar}

\begin{proof}[\textbf{Proof of Equation \eqref{eq:weightedvariance}}]
  Inserting in the optimal probability defined in \eqref{eq:pi}, we obtain that
  \begin{align*}
    &\Exp\left[\frac{1}{n^{2}}\sumn b^{''}(X_{i}^{T}\beta_{0})X_{i}X_{i}^{T}\left\{\frac{1}{r\pi_{i}}-\oner+1\right\}\right]\\
    &=\Exp\left\{\frac{1}{n^{2}}\sumn b^{''}(X_{i}^{T}\beta_{0})X_{i}X_{i}^{T}\oner\frac{\sumjn\sqrt{b^{''}(X_{j}^{T}\beta_{0})}\left\|L\Phi^{-1}X_{j}\right\|}{\sqrt{b^{''}(X_{i}^{T}\beta_{0})}\left\|L\Phi^{-1}X_{i}\right\|}\right\}-\Exp\left\{\frac{1}{rn^{2}}\sumn b^{''}(X_{i}^{T}\beta_{0})X_{i}X_{i}^{T}\right\}\\
    &\quad + \Exp\left\{\frac{1}{n^{2}}\sumn b^{''}(X_{i}^{T}\beta_{0})X_{i}X_{i}^{T}\right\}\\
    &=\oner\Exp\left\{\onen\sumjn\sqrt{b^{''}(X_{j}^{T}\beta_{0})}\left\|L\Phi^{-1}X_{j}\right\|\onen\sumn\frac{b^{''}(X_{i}^{T}\beta_{0})X_{i}X_{i}^{T}}{\sqrt{b^{''}(X_{i}^{T}\beta_{0})}\left\|L\Phi^{-1}X_{i}\right\|}\right\}-\frac{1}{rn}\Phi+\onen\Phi\\
    &=\frac{1}{rn^{2}}\sumn\Exp\left\{b^{''}(X_{i}^{T}\beta_{0})X_{i}X_{i}^{T}\right\}+\frac{1}{rn^{2}}\sum_{i\neq j}\Exp\left\{\sqrt{b^{''}(X_{i}^{T}\beta_{0})}\left\|L\Phi^{-1}X_{i}\right\| \frac{b^{''}(X_{j}^{T}\beta_{0})X_{j}X_{j}^{T}}{\sqrt{b^{''}(X_{j}^{T}\beta_{0})}\left\|L\Phi^{-1}X_{j}\right\|}\right\}\\
    &\quad -\frac{1}{rn}\Phi+\onen\Phi\\
    &=\frac{1}{rn}\Phi+\oner\frac{n(n-1)}{n^{2}}m\Lambda-\frac{1}{rn}\Phi+\onen\Phi\\
    &=\oner\frac{n(n-1)}{n^{2}}m\Lambda+\onen\Phi=\oner\frac{n-1}{n}m\Lambda+\onen\Phi.
  \end{align*}
  This completes the proof
\end{proof}

\subsubsection{Proof of Theorem \ref{th:Efficiency}}\label{sec:th_effi}

\begin{proof}[\textbf{Proof of Theorem \ref{th:Efficiency}}]
  First, we prove that $\Gamma^{-1}\leq\Phi^{-1}\Lambda\Phi^{-1}$.
  Denote $\mathbf{v}=\sqrt{b^{''}(X^{T}\beta_{0})}X$ and
  $h=\sqrt{b^{''}(X^{T}\beta_{0})}\|L\Phi^{-1}X\|$. We then only
  need to
  prove
  $$\Exp(h\mathbf{v}\mathbf{v}^{T})^{-1}\leq\Exp(\mathbf{v}\mathbf{v}^{T})^{-1}\Exp(h^{-1}\mathbf{v}\mathbf{v}^{T})\Exp(\mathbf{v}\mathbf{v}^{T})^{-1}.$$
  Denote
  \begin{equation*}
    \mathbf{f}:=\sqrt{h}\Exp(h\mathbf{v}\mathbf{v}^{T})^{-1}\mathbf{v}-\frac{1}{\sqrt{h}}\Exp(\mathbf{v}\mathbf{v}^{T})^{-1}\mathbf{v}.
  \end{equation*}
  Since $\mathbf{f}\mathbf{f}^{T}\geq0$, we have
  \begin{equation*}
    0\leq\Exp(\mathbf{f}\mathbf{f}^{T})=\Exp(\mathbf{v}\mathbf{v}^{T})^{-1}\Exp(h^{-1}\mathbf{v}\mathbf{v}^{T})\Exp(\mathbf{v}\mathbf{v}^{T})^{-1}-\Exp(h\mathbf{v}\mathbf{v}^{T})^{-1}.
  \end{equation*}
  Next, we prove $\Phi^{-1}\leq\Gamma^{-1}\Omega\Gamma^{-1}$.
  It is straightforward to verify that
  $\Phi=\Exp(\mathbf{v}\mathbf{v}^{T})$ and
  $\Gamma^{-1}\Omega\Gamma^{-1}=\Exp(h\mathbf{v}\mathbf{v}^{T})^{-1}\Exp(h^{2}\mathbf{v}\mathbf{v}^{T})\Exp(h\mathbf{v}\mathbf{v}^{T})^{-1}$. Denote
  \begin{equation*}
    \mathbf{g}:=\Exp(\mathbf{v}\mathbf{v}^{T})^{-1}\mathbf{v}-h\Exp(h\mathbf{v}\mathbf{v}^{T})^{-1}\mathbf{v}.
  \end{equation*}
  Since $\mathbf{g}\mathbf{g}^{T}\geq0$, we have
  \begin{equation*}
    0\leq\Exp(\mathbf{g}\mathbf{g}^{T})\leq\Exp(h\mathbf{v}\mathbf{v}^{T})^{-1}\Exp(h^{2}\mathbf{v}\mathbf{v}^{T})\Exp(h\mathbf{v}\mathbf{v}^{T})^{-1}-\Exp(\mathbf{v}\mathbf{v}^{T})^{-1}.
  \end{equation*}
  This completes the proof.
\end{proof}

  \subsection{Additional numerical experiments}\label{sec:add_expriments}
  In this section, we present additional numerical experiment results.
  \subsubsection{Numerical results using L-optimality criterion under Linear
    regression}\label{sec:Lop_linear}
  We now present the numerical results for linear model when using sampling
  probabilities obtained under L-optimality criterion,
  $\pi_i^{\mvc}$. In this scenario, the subsampling probabilities in
  Algorithm \ref{alg:algolinear} are
  \begin{equation*}
    \pi_{i}^{\mvc}=\frac{\left\| X_{i}\right\|}{\sum_{k=1}^{n}\left\| X_{k}\right\|}.
  \end{equation*}
  We use the same settings as those  in Section~\ref{sec:simulation} for linear models to
  investigate the performance of $\pi_i^{\mvc}$. The results are
  presented in Figure \ref{fig:unconditionMSE_linear_L_optimal}. For the T1 case, we
  use a trimmed mean with $\alpha=0.05$ to calculate the empirical MSEs. From
  Figure~\ref{fig:unconditionMSE_linear_L_optimal}, the
  results are similar to those in Section~\ref{sec:simulation} for
  A-optimality criterion. The unweighted estimator outperforms the weighted
  estimator for all the cases as well when we use $\pi_i^{\mvc}$.
  \begin{figure}[H]
    \centering
    \begin{subfigure}{0.48\textwidth}
      \includegraphics[width=\textwidth]{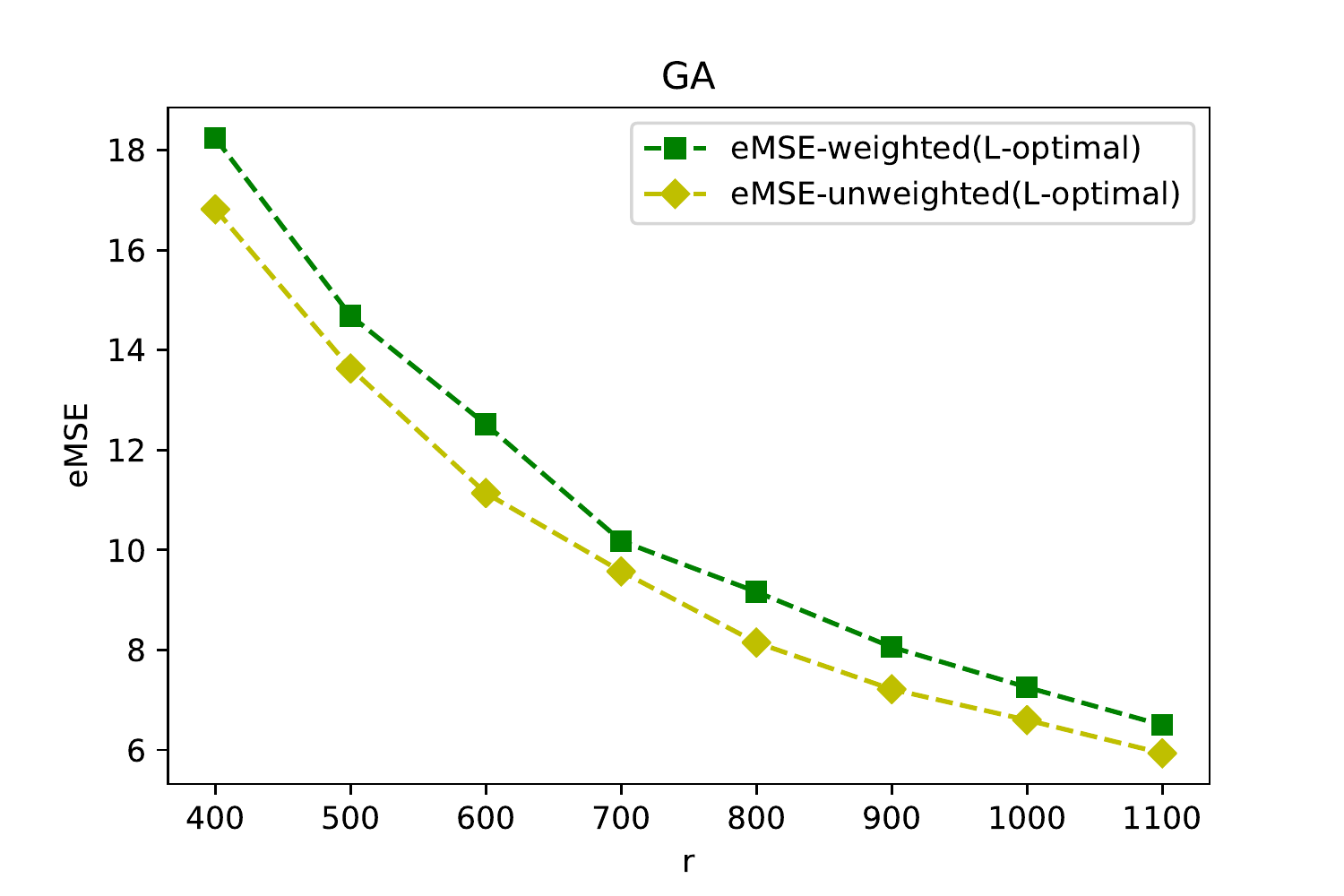}
      \caption{GA}
    \end{subfigure}
    \begin{subfigure}{0.48\textwidth}
      \includegraphics[width=\textwidth]{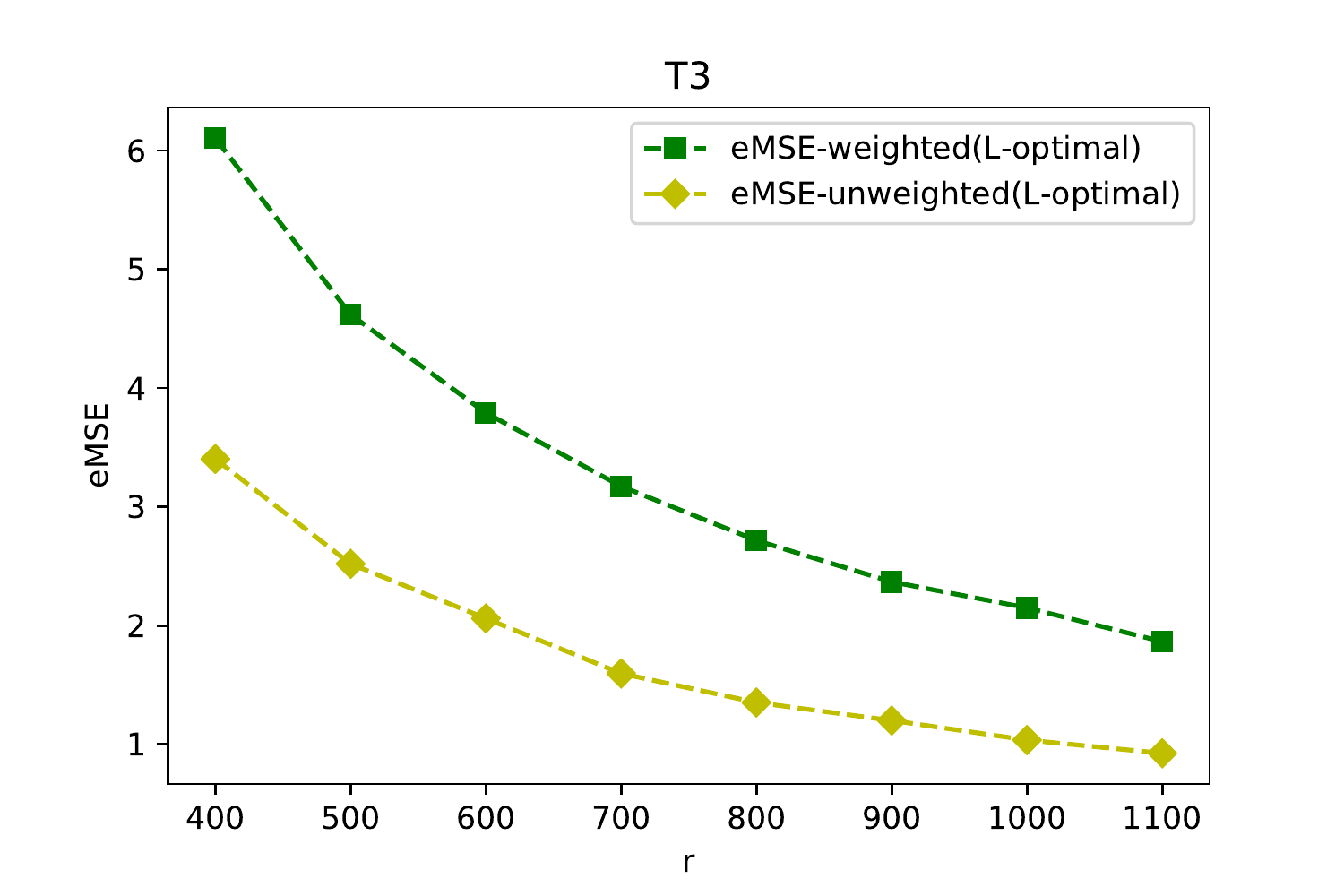}
      \caption{T3}
    \end{subfigure}
    \begin{subfigure}{0.48\textwidth}
      \includegraphics[width=\textwidth]{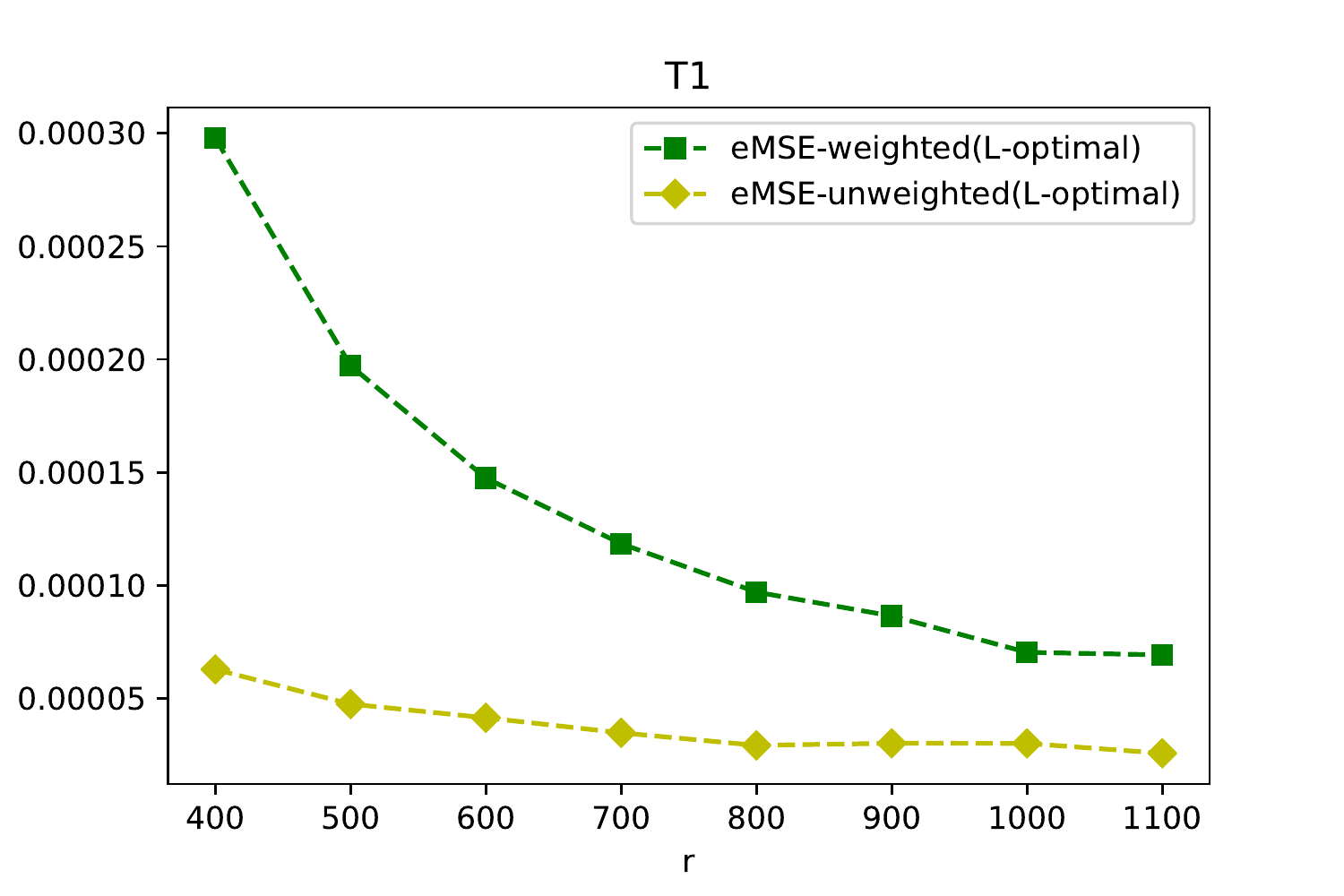}
      \caption{T1}
    \end{subfigure}
    \begin{subfigure}{0.48\textwidth}
      \includegraphics[width=\textwidth]{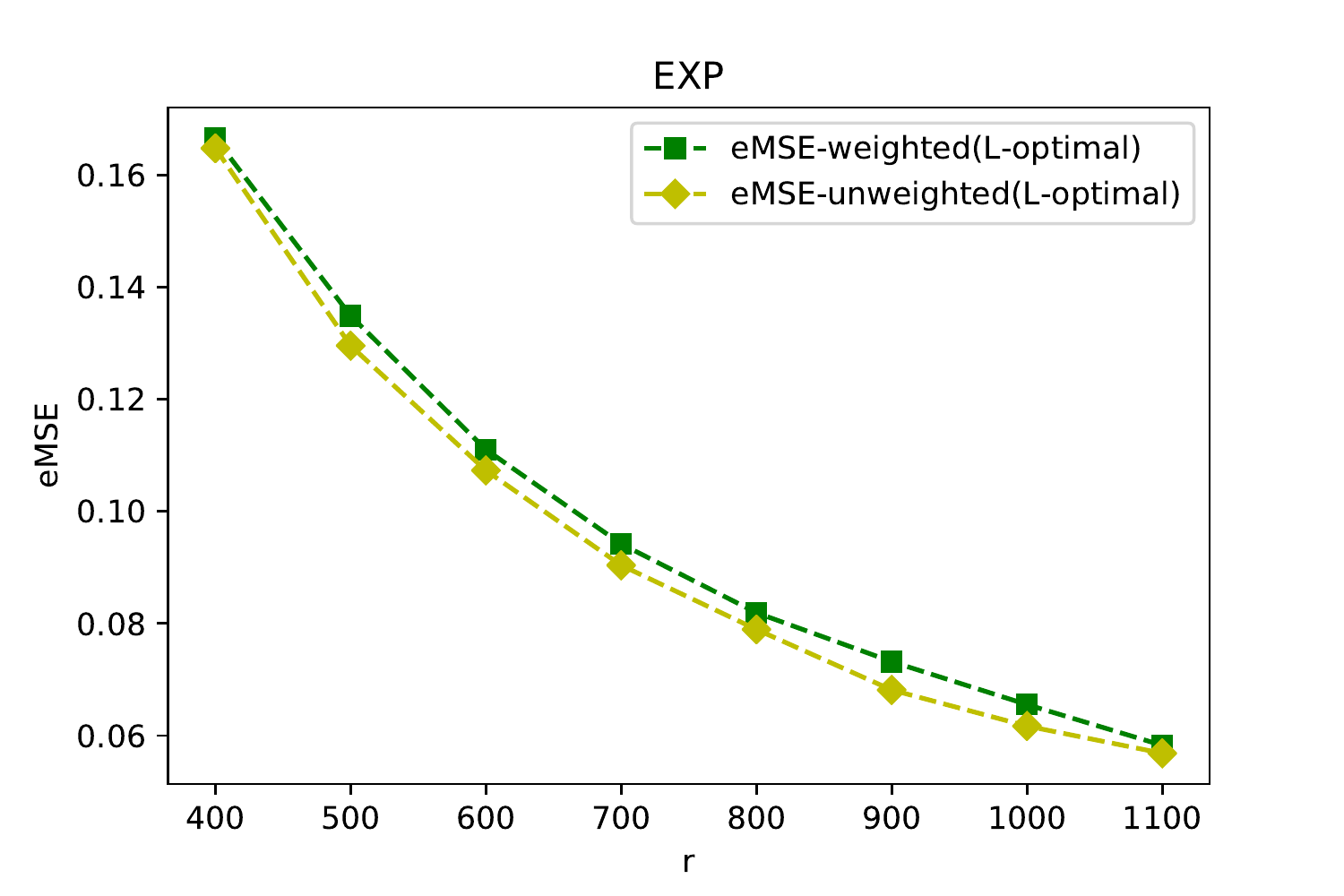}
      \caption{EXP}
    \end{subfigure}
    \caption{eMSE for different subsampe sizes $r$ for linear regression under
      different settings.}
    \label{fig:unconditionMSE_linear_L_optimal}
  \end{figure}

\subsubsection{Influence of the pilot estimation}\label{sec:pilot}
In this section, we use numerical experiments to evaluate the effect of the
pilot estimation method. We use the same settings as in Section
\ref{sec:simulation} for logistic regression and take $N=100000$. We set the
pilot subsample size to be $r_{\rp}=500$. To compare the performances of using
different methods to obtain the pilot estimates, we consider the simple random
sampling and the case-control sampling \citep{fithian2014local,
  wang2019more}. The case-control sampling method use the following probabilities:
\begin{equation*}
\pi_{0i}=\frac{c_0(1-y_i)+c_1y_i}{n}
\end{equation*}
to take data points, where $c_0$ and $c_1$ are constants to balance the
responses. We choose $c_0=\{ 2(1-p_m) \}^{-1}$ and $c_1=(2p_m)^{-1}$, where $p_m$
denotes the prior maginal probability $\Pr(y=1)$.
We present the results in Figure \ref{fig:unconditionMSE_logistic_pilot}.
It is seen that the eMSEs of the final subsample estimators obtained by using
simple random sampling and case-control sampling as the pilot estimation method are
similar for both A-optimal and L-optimal probabilities. This indicates that the
influence of the methods used to obtain the pilot estimates is not significant
for the setting considered.
\begin{figure}[H]
  \centering
  \begin{subfigure}{0.48\textwidth}
    \includegraphics[width=\textwidth]{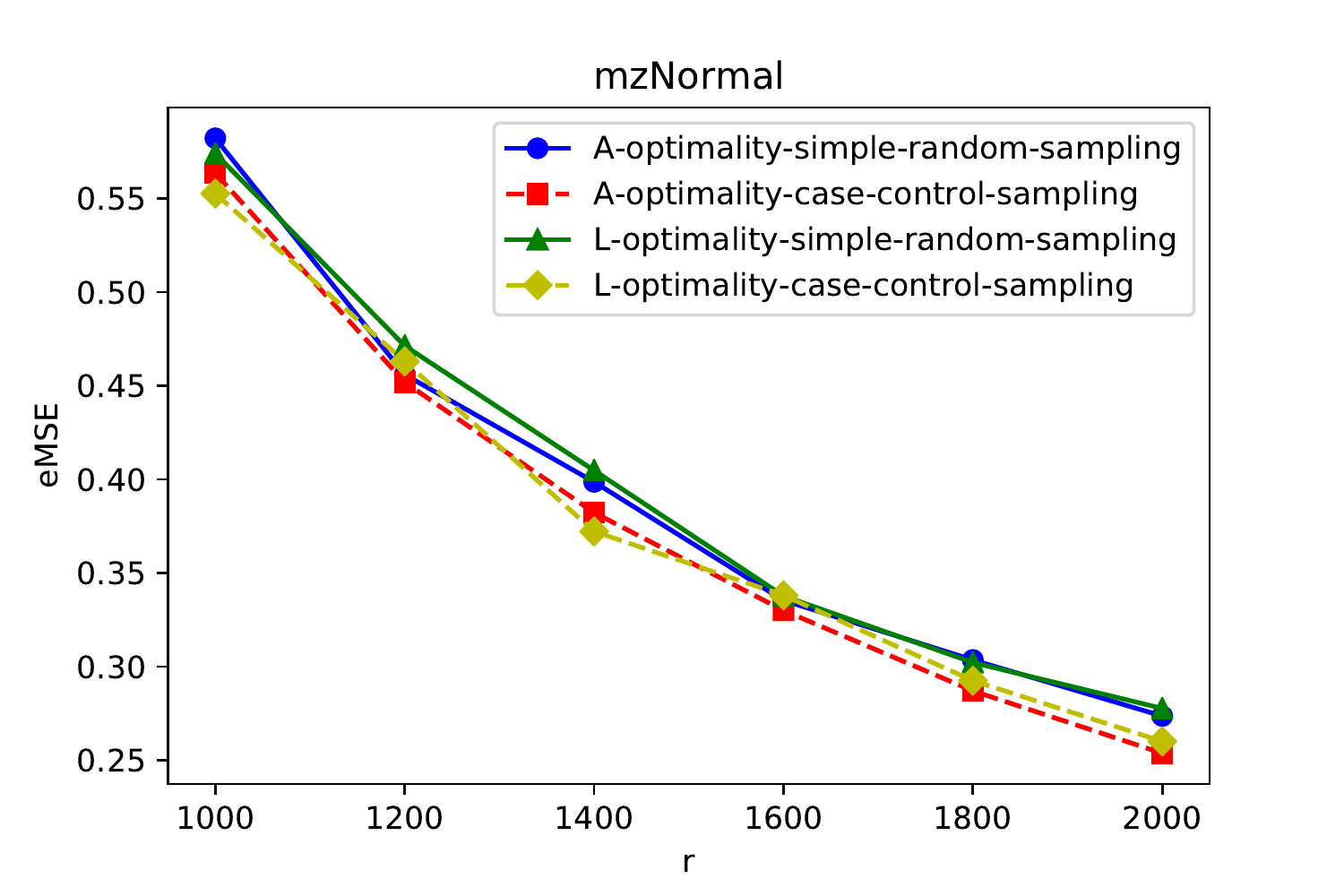}
    \caption{mzNormal}
  \end{subfigure}
  \begin{subfigure}{0.48\textwidth}
    \includegraphics[width=\textwidth]{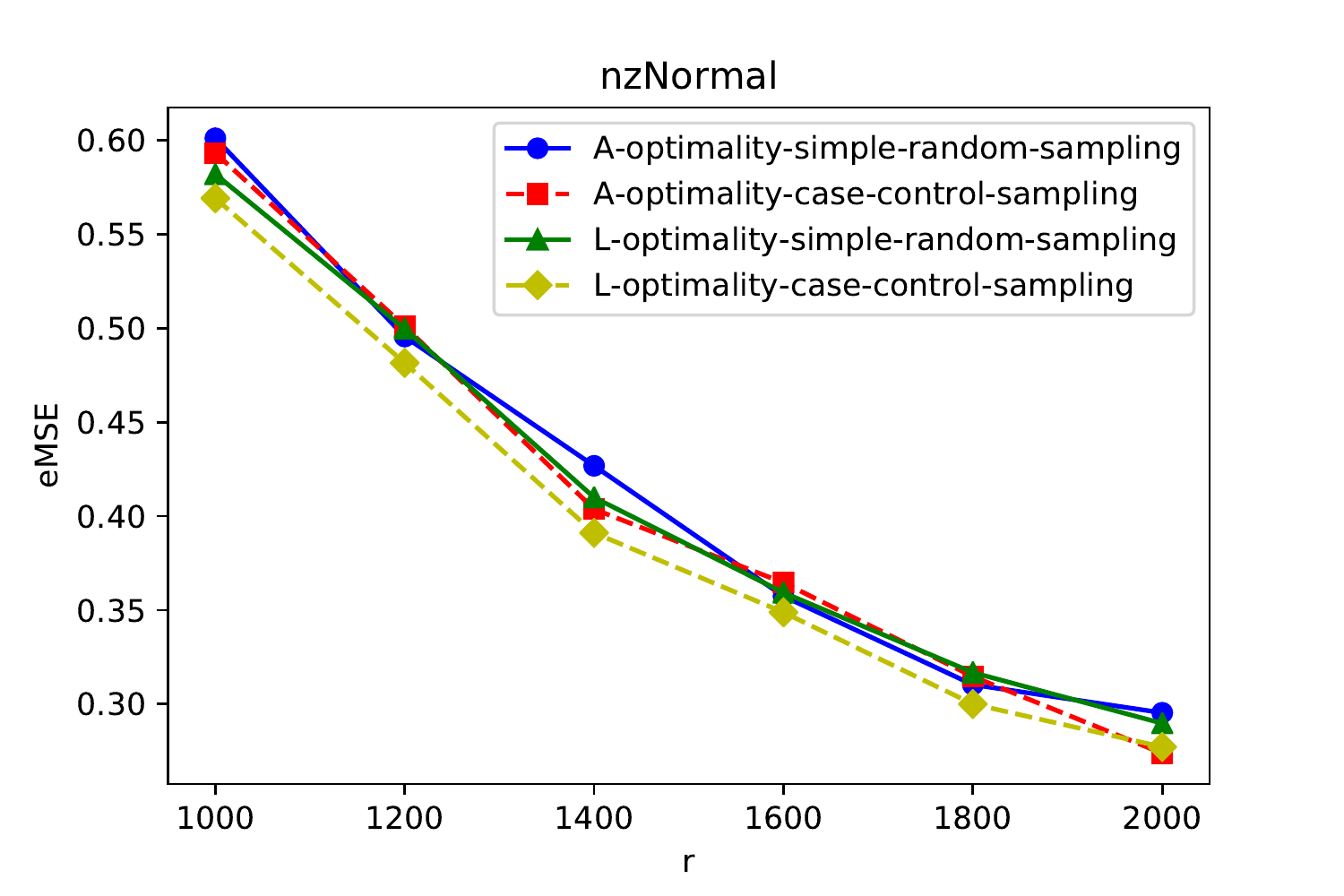}
    \caption{nzNormal}
  \end{subfigure}
  \begin{subfigure}{0.48\textwidth}
    \includegraphics[width=\textwidth]{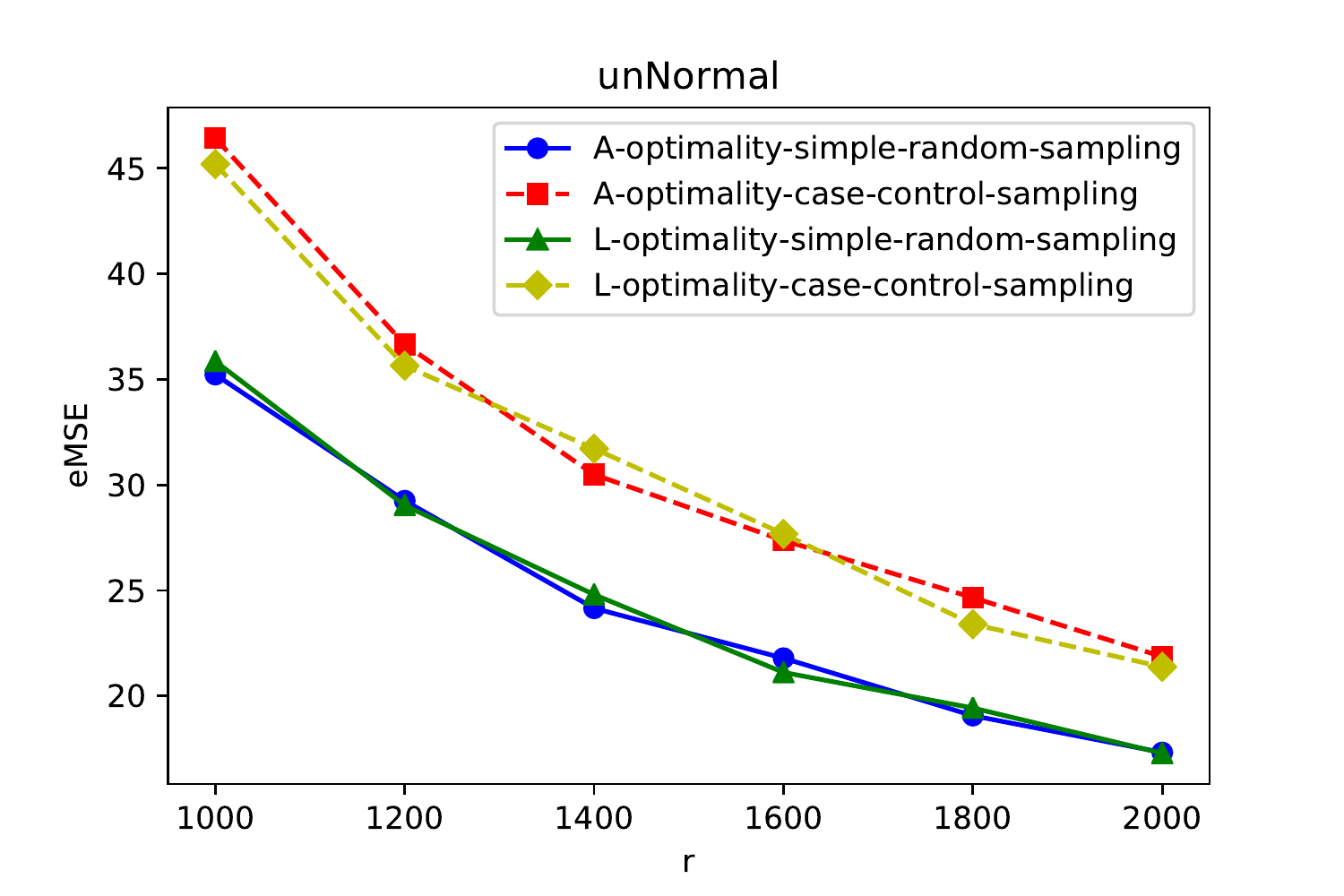}
    \caption{unNormal}
  \end{subfigure}
  \begin{subfigure}{0.48\textwidth}
    \includegraphics[width=\textwidth]{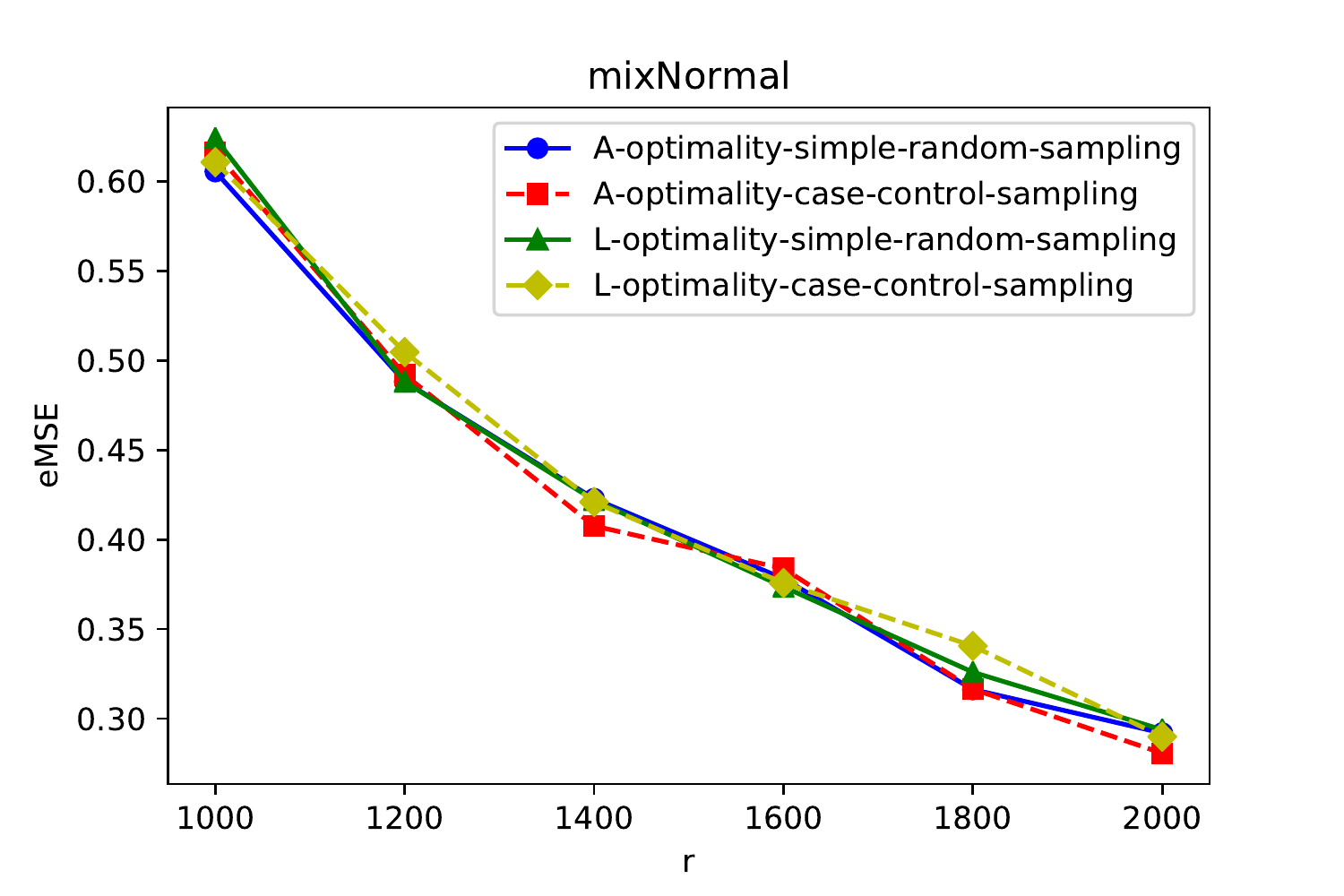}
    \caption{mixNormal}
  \end{subfigure}
  \caption{eMSE for different subsample sizes $r$ with a pilot sample size
    $r_{\rp}=500$ for logistic regression under different settings and
    different sampling methods for pilot estimation.}
  \label{fig:unconditionMSE_logistic_pilot}
\end{figure}

\bibliographystyle{myapa}
\bibliography{references}
\end{document}